%% file: main.tex
\documentclass{article}

\usepackage[top=.8in,bottom=1in,left=.6in,right=.6in]{geometry}

\usepackage{hyperref}
\usepackage{url}
\usepackage{float}
\usepackage{amsmath, amssymb, amsthm}
\usepackage[dvipsnames]{xcolor}
\usepackage{titlesec}
\usepackage{graphicx}
\usepackage{tikz-cd}
\usepackage{comment}
\usepackage{enumitem}
\usepackage{float}
\usepackage{Definitions/macros_gpi}
\usepackage{multicol}

\usepackage{xcolor}
\hypersetup{
    colorlinks,
    linkcolor={orange!60!black},
    citecolor={blue!50!green},
    urlcolor={blue!40!green}
}

\author{Olivier Peltre}
\date{\small
Univ. Artois, CNRS, CRIL, F-62300 Lens, France}

\title{{\bf Local Max-Entropy and Free Energy Principles,\\
Belief Diffusions and their Singularities}}

\input{commands.tex}
\input{template.tex}

\begin{document}

\maketitle

\begin{multicols}{2}

\begin{abstract}
\input{abstract.tex}

\end{abstract}

\input{introduction.tex}

\input{I-graphicalModels.tex}
\input{II-i-Boltzmann-Gibbs.tex}
\input{II-ii-Bethe-Kikuchi.tex}

\input{III-localised.tex}
\input{IV-diffusion.tex}

\input{V-i-theorems.tex}
\input{V-ii-singularities.tex}
\input{V-iii-graphs.tex}
\input{VI-conclusion.tex}

\section*{Acknowledgements}

This work has benefited from the support of the AI Chair EXPEKCTATION (ANR-19-CHIA-0005-01) 
of the French National Research Agency (ANR).

I am very grateful to Yaël Frégier, Pierre Marquis and Frédéric Koriche for their support in Lens,
and to Daniel Bennequin, Grégoire Sergeant-Perthuis 
and Juan-Pablo Vigneaux for fostering my interest in graphical models. 

\bibliography{biblio/biblio.bib}
\bibliographystyle{myunsrt}

\appendix
\input{sectionA-gauss.tex}
\input{sectionA-fix.tex}
\input{sectionA.tex}

\input{sectionB.tex}

\end{multicols}
\end{document}

%% file: commands.tex
\newcommand{\N}{\mathbb{N}}
\newcommand{\R}{\mathbb{R}}
\newcommand{\Z}{\mathbb{Z}}

\DeclareMathOperator{\e}{\mathrm{e}}

\newcommand{\Pow}{\mathcal{P}}
\newcommand{\incl}{\subseteq}
\newcommand{\cont}{\supseteq}
\newcommand{\eqvl}{\Leftrightarrow}

\newcommand{\Set}{\mathbf{Set}}
\newcommand{\Alg}{\mathbf{Alg}}
\newcommand{\Vect}{\mathbf{Vect}}

\DeclareMathOperator{\Ker}{Ker}
\DeclareMathOperator{\Img}{Im}

\newcommand{\cork}{{\bf cork}}

\newcommand{\gl}{\mathfrak{gl}}
\newcommand{\Perm}{\mathfrak{S}}

\renewcommand{\Prob}{\mathrm{Prob}}
\newcommand{\E}{\mathbb{E}}

\newcommand{\Tg}{\mathrm{T}}

\newcommand{\F}{F}
\newcommand{\varF}{\mathcal{F}}

\newcommand{\U}{\mathcal{U}}
\newcommand{\DF}{\mathcal{D}}

\renewcommand{\aa}{\mathrm{a}}
\renewcommand{\bb}{\mathrm{b}}
\newcommand{\cc}{\mathrm{c}}
\newcommand{\dd}{\mathrm{d}}

\newcommand{\Om}{\Omega}
\newcommand{\ph}{\varphi}
\newcommand{\eps}{\varepsilon}

\newcommand{\Fix}{\mathcal{M}}
\newcommand{\Fixt}{\mathcal{M}_+}
\newcommand{\fix}{{\cal N}}

\newcommand{\Beliefs}{{\cal B}}
\newcommand{\trito}{\triangleright}

%% file: template.tex
\renewcommand{\thesection}{\Roman{section}}
\renewcommand{\thesubsection}{\Alph{subsection}}
\renewcommand{\thesubsubsection}{\arabic{subsubsection}}

\titleformat{\section}%
  {\centering\uppercase}{{\thesection}.}{0.5em}{}
\titleformat{\subsection}%
  {\it}{\thesubsection}{0.5em}{}
\titleformat{\subsubsection}
         {\it\raggedright}{\hspace{1.5em}{\thesubsubsection}.}{0.5em}{}

\makeatletter
\renewcommand{\p@subsection}{\thesection.}
\makeatother

\newcounter{mytheorem}
\newtheorem{MyTheorem}[mytheorem]{Theorem}

\newcounter{theorem}
\newtheorem{Theorem}[theorem]{Theorem}
\newtheorem{Corollary}[theorem]{Corollary}

\newcounter{definition}
\newtheorem{Definition}[definition]{Definition}

\newtheorem{Proposition}[definition]{Proposition}
\newtheorem{Lemma}[definition]{Lemma}

\newcounter{problem}
\newtheorem{Problem}[problem]{Problem}

%% file: abstract.tex
A comprehensive picture of three Bethe-Kikuchi variational principles 
including their relationship to belief propagation (BP) algorithms on 
hypergraphs is given. 
The structure of BP equations is generalized to define continuous-time 
diffusions, solving 
localized versions of
the max-entropy principle (A), the variational free energy principle (B), 
and a less usual equilibrium free energy principle (C), Legendre dual to A. 
Both critical points 
of Bethe-Kikuchi functionals and stationary beliefs
are shown to lie at the non-linear intersection of two constraint surfaces, enforcing
energy conservation and marginal consistency respectively.
The hypersurface of singular beliefs, 
accross which equilibria 
become unstable as the constraint surfaces meet tangentially, 
is described by polynomial equations in the convex polytope 
of consistent beliefs. This polynomial is expressed 
by a loop series expansion for graphs of binary variables.

%% file: introduction.tex
\section{Introduction}

Boltzmann-Gibbs principles describe the equilibrium state
$p_\Om = \frac 1 {Z_\Om} \e^{- H_\Om}$ of a statistical sytem, given its hamiltonian 
or energy function $H_\Om : E_\Om \to \R$, as solution to 
a collection of variational problems  \cite{Jaynes-57, Marle-16}.
Each one corresponds to a different set of external constraints 
(energy or temperature, volume or pressure...)
yet they are all related by Legendre transforms on the constraint parameters. 
However, evaluating the
partition function $Z_\Om$ is impossible for large configuration spaces 
$E_\Om$, and the design of efficient algorithms for estimating 
$p_\Om$ or a subset of its marginals is a challenge with countless
applications. 

Given a hypergraph $K \subseteq \Pow(\Om)$ with vertices in the set of 
variables $\Om$,
the Bethe-Kikuchi principles \ref{S-crit}, \ref{varF-crit} and 
\ref{F-crit} below yield tractable variational 
problems, where instead of the global distribution $p_\Om$,
one optimizes over a field of {\it consistent local beliefs} 
$(p_\aa)_{\aa \in K}$.
Controlling the range of interactions independently 
of the size of the system, they exploit the asymptotic 
additivity of \emph{extensive} thermodynamic
functionals. 
The free energy principle~\ref{varF-crit}, 
also known as the cluster variation method 
(CVM)~\cite{Bethe-35,Kikuchi-51,Morita-57}, 
was notably shown to have exponential convergence 
on the Ising model when $K$ grows 
coarse \cite{Schlijper-83}, the choice of $K$ thus
offers a compromise between precision and complexity.
It was already known that CVM solutions 
may be found by the Generalized Belief Propagation (GBP) algorithm 
of Yedidia, Freeman and Weiss~\cite{Yedidia-2005,gsi19,phd} \nocite{gsi21}
(and known for even longer when $K$ is a graph~\cite{Ikeda-04}). 
This algorithm is however far from being optimal, 
and there lacked a comprehensive understanding of their correspondence 
including the relationship with a Bethe-Kikuchi max-entropy principle 
(\ref{S-crit}) and its Legendre-dual free energy principle (\ref{F-crit}). 

We here describe continuous-time diffusion equations on belief networks 
which smooth out most convergence issues of GBP, recovered as 
a time-step 1 Euler integrator.  
We then show 
that they solve the three 
different Bethe-Kikuchi variational problems, \ref{S-crit},
\ref{varF-crit} and \ref{F-crit}, whose 
critical points are shown to lie at the intersection of 
two constraint manifolds, enforcing {\it energy conservation} 
and {\it belief consistency} respectively. 
The former consists of homology classes in a chain 
complex of local observables $(C_\bullet, \delta)$, 
the latter consists of cohomology classes in 
the dual cochain complex of local measures $(C_\bullet^*, d)$, 
but these two are related by the non-linear correspondence 
mapping energy functions to local Gibbs states.  
While solutions to the max-entropy 
principle \ref{S-crit} 
are stationary 
states of \emph{adiabatic} diffusion algorithms, 
preserving the mean energy $\E[H_\Om]$, 
the variational free energy principle \ref{varF-crit} 
(CVM) is solved by \emph{isothermal} diffusions, 
preserving the inverse temperature $\beta = 1 / T$.
The equilibrium free energy principle \ref{F-crit} is 
dual 
to the other two,  
as it optimizes 
over fibers of gauge-transformed energy functions, 
and not over consistent beliefs. 
Its solutions retract onto those of \ref{varF-crit} 
and are also be found by isothermal diffusions. 

From a physical perspective, Bethe-Kikuchi principles 
could be viewed as mere tools
to approximate 
the global and exact Boltzmann-Gibbs principles, although
precise and powerful.
The possible coexistence of multiple energy values 
at a fixed temperature still reminds counter-intuitive yet 
physical phenomena,
such as surfusion and metastable equilibria.  
Recently, free energy principles, Bethe-Kikuchi approximations 
and BP algorithms
also made their way to neuroscience 
\cite{Friston-Parr-2017, Parr-19,Rudrauf-17, Sanchez-21}, 
a context where it is not very clear what a global 
probability distribution $p_\Om$ ought to describe.  
The locality of interactions is yet somehow materialized 
by neuron dendrites and larger scale brain connectivity.
Message-passing in search  
of consensus is
an interesting and working
metaphor of neuronal behaviours,
which might in this case hold more reality 
than global variational principles. 
The remarkable success of BP algorithms in decoding 
and stereovision 
applications demonstrates the potential of message-passing schemes 
on graphs 
to solve difficult problems. 
However the success of BP algorithms on loopy networks is often 
presented as an empirical coincidence, 
and a deeper understanding of their different regimes 
could provide the missing theoretical guarantees. 

Whatever the perspective, singularities 
of Bethe-Kikuchi functionals and belief diffusions
(in finite size)
are an important and interesting 
feature. 
They happen when the two constraint manifolds meet tangentially. 
A stationary state crossing the singular surface will generically become 
unstable, attracted towards a different sheet of the intersection. 
This would appear as a discontinuous jump in 
the convex polytope  $\Gamma_0$ of consistent beliefs. 
We show that the singular stratification 
$\bar \Sigma^1 = \bigsqcup_{k \geq 1} \Sigma^k$ 
is described by polynomial equations in $\Gamma_0$. 
They compute 
the corank $k$ of linearized diffusion, restricted to 
the subspace of infinitesimal gauge transformations. 
For graphs of binary variables, 
this polynomial is written explicitly 
in terms of a loop series expansion.

\nocite{Bethe-35,Kikuchi-51,Morita-57}
\nocite{Gallager-63, Pearl-82, Yedidia-2005}

\subsection{Related work}

The first occurence 
of BP as an approximate bayesian inference scheme dates back 
to Gallager's 1962 thesis \cite{Gallager-63} on decoding, 
although it is often attributed to Pearl's 1982 article 
on bayesian trees \cite{Pearl-82}, where it is exact. 
BP has received a lot of attention and new applications 
since then although it is still mostly famous in the 
decoding context, on this see for instance 
\cite{Weiss-97, Murphy-Weiss-99, Kschischang-01, Yedidia-2001,
Tehrani-Jego-08} and \cite{Kschischang-01} for an excellent review. 
In telecommunications, BP is thus used to reconstruct a 
parity-check encoded signal by iteratively updating beliefs 
until all the local constraints are satisfied. 
Although the associated factor graph has loops, 
BP works surprisingly well at reconstructing the signal. 
See \cite{Knoll-17} for a numerical study of loopy BP and its singularities. 
As a marginal estimation algorithm, usecases for BP 
and its generalizations to hypergraphs are quite universal. 
Other interesting applications for instance include (but are not limited to) 
computer stereovision \cite{Sun-03} and conditional Boltzmann machines \cite{Ping-2017}. 
Gaussian versions of BP also exist \cite{Mezard-2017}, from 
which one could for instance recover the well known K{\'a}lm{\'a}n filter 
on a hidden Markov chain \cite{Kschischang-01}. 

The relationship with Bethe-Kikuchi approximations 
is covered in the reference works
\cite{Yedidia-2001, Ikeda-04, Mezard-Montanari} in the case of graphs,  
yet a true correspondence with the CVM on hypergraphs 
$K \subseteq \Pow(\Om)$ could bot be stated 
before the GBP algorithm of Yedidia, Freeman and Weiss 
in 2005 \cite{Yedidia-2005},  
whose work bridged two subjects with a long history.  
The idea to replace the partition function $Z_\Om$ by 
a sum of local terms $\sum_{\aa \in K} c_\aa Z_\aa$, where 
coefficients $c_\aa \in \Z$ take care of eliminating redundancies, 
was first introduced 
by Bethe in 1935,  
and generalized by Kikuchi in 1951 \cite{Kikuchi-51}. 
The truncated Möbius inversion formula was only recognized by Morita 
in 1957 \cite{Morita-57}, laying the CVM on systematic 
combinatorial foundations. 
Among recent references, see \cite{Pelizzola-05} 
for a general introduction to the CVM. 
The convex regions of Bethe-Kikuchi free energies 
and loopy BP stability are studied in \cite{Mooij-07,Watanabe-09}, while 
very interesting loop series expansions 
may be found in \cite{Mori-13} and \cite{Sudderth-07}.
For applications of Bethe-Kikuchi free energies  
to neuroscience and active bayesian learning, see also 
\cite{Friston-Parr-2017, Parr-19,Rudrauf-17}.

The unifying notion of graphical model
describes Markov random fields by their factorization properties, which are
in general stronger than their conditional independence properties 
obtained by the 
Hammersley-Clifford theorem \cite{Hammersley-Clifford-71}.
This work sheds a different light from the usual probabilistic 
interpretation,  so as to make the most of 
the local structure of observation. 
The mathematical constructions below
thus mostly borrow from algebraic topology
and combinatorics. The reference on combinatorics is Rota \cite{Rota-64}, 
and the general construction 
for the cochain complex $(C_\bullet^*, d)$ dates back to 
Grothendieck and Verdier \cite{SGA-4-V}. It has been given 
a very nice description by Moerdijk in his short book \cite{Moerdijk}.
When specializing the theory to localized statistical systems, 
one quickly arrives at the so-called \emph{pseudo-marginal extension problem}
\cite{Kellerer-64,Vorobev-62,Abramsky-2011} 
whose solution is closely related to the \emph{interaction decomposition theorem} 
\cite{Kellerer-64,Matus-88}. This fundamental result 
yields direct sum decompositions for the functor of 
local observables, used in the proof of theorem \ref{thm:Gauss} below.

\subsection{Methods and Results}

This work brings together concepts and methods from 
algebraic topology, combinatorics,
statistical physics and information theory. 
It should particularly interest users of belief propagation algorithms, 
although we hope it will also motivate a broader use of Bethe-Kikuchi information
functionals beyond their proven decoding applications. 
It is intended as a comprehensive but high-level reference on
the subject for a pluridisciplinary audience. 
We expect some readers might lack specific vocabulary from 
homological algebra, although we do not believe it necessary for understanding 
the correspondence theorems. We provide 
specific references to theory and applications where needed, and 
longer proofs are laid in appendix to avoid burdening the main text 
for readers mostly interested by the results. 

The main object of theory here consists 
in what we call the {\it combinatorial chain complex} 
$(C_\bullet, \delta, \zeta)$. 
This is a graded vector space $\bigoplus_{r\geq 0} C_r$, 
a codifferential $\delta : C_r \to C_{r-1}$, 
and a combinatorial automorphism $\zeta : C_r \to C_r$, 
attached to any hypergraph $K$ 
with vertices in the set of variables $\Om$.  
The operators $\zeta$ and $\delta$ acting 
on $C_0$ and $C_1$ generate belief propagation equations, 
and an efficient implementation is made
available at \href{https://github.com/opeltre/topos}{github.com/opeltre/topos}.  
Although deeper numerical 
experiments are not in the scope of this article, 
this library was used to produce the level curves of a Bethe-Kikuchi free energy 
in figure \ref{fig:singular-contours} 
and the benchmarks of figure \ref{fig:benchmark}.

Initial motivations were to arrive at a concise factorization 
of GBP algorithms, and at a rigorous proof of the correspondence 
between GBP fixed points and CVM solutions. Although 
described earlier \cite{gsi19, phd, gsi21},
there lacked
their relationship to a Bethe-Kikuchi max-entropy
principle \ref{S-crit} and its dual (equilibrium) 
free energy principle \ref{F-crit}.
The polynomial description 
of singular sets $\Sigma^k \subseteq \Gamma_0$ is also new, 
and their explicit formulas on binary graphs yield 
a very satisfying and most expected relationship with 
the subject of loop series expansions \cite{Watanabe-09,Mori-13,Sudderth-07}. 

The article is structured as follows. 
\begin{itemize}[label=$-$,leftmargin=1em]
\item \ref{section:GMs}. {\small\uppercase{Graphical Models and Generalized Belief Propagation}}
defines Gibbsian ensembles 
as positive Graphical Models (GMs). The factorization property 
of the probability density translates as a linear spanning property of 
its log-likelihood, called a $K$-local observable. 
GBP equations are then provided along with classical examples. 

\item \ref{section:Principles}. {\small\uppercase{Max-Entropy and Free Energy Principles}}
briefly reviews Boltzmann-Gibbs and Bethe-Kikuchi principles 
so as to formulate the variational problems \ref{S-crit}, \ref{varF-crit} 
and \ref{F-crit}, localized versions 
of the fundamental principles defining thermodynamic equilibrium in statistical physics. 
 
\item\ref{section:Local}. {\small\uppercase{Local Statistical Systems}} is 
the core technical section, 
where we define the chain complex $(C_\bullet, \delta)$ 
of local observables, its dual complex $(C_\bullet^*, d)$ of
local densities, and the combinatorial automorphisms 
$\zeta$ and $\mu = \zeta^{-1}$ acting on all the degrees of $C_\bullet$.
These higher degree combinatorics of $C_r$ for ${r > 1}$ 
were described previously in \cite[chap. 3]{phd}.
We here propose an integral notation for the zeta transform 
making the analogy with geometry more intuitive, 
in the spirit of Rota \cite{Rota-64}. 

\item \ref{section:Diffusions}. {\small\uppercase{Belief Diffusions}}
uses the codifferential $\delta$ and its conjugate 
under $\zeta$ to generate diffusion equations on the 
complex $(C_\bullet, \delta, \zeta)$.
We explain under which 
conditions a flux functional $\Phi : C_0 \to C_1$
yields solutions to problems \ref{S-crit}, \ref{varF-crit} and \ref{F-crit} 
as stationary states of the diffusion $\frac{dv}{dt} = \delta \Phi(v)$
on $C_0$. Its purpose is to explore a homology class $[v]$ 
({\it energy conservation})
until meeting the preimage $\Fix^\beta \subseteq C_0$ 
of cohomology classes under the Gibbs state map 
at inverse temperature $\beta$ ({\it marginal consistency}).  

\item \ref{section:Equilibria}. {\small\uppercase{Message-Passing Equilibria}}
states the rigorous correspondence between critical points 
and stationary beliefs with theorems \ref{thm:S}, \ref{thm:varF} and \ref{thm:F}.
Singular subsets $\Sigma^k \subseteq \Gamma_0$ are defined 
by computing the dimension of the intersection of the two 
tangent constraint manifolds, almost everywhere transverse. 
The fact that it may be described by polynomial equations 
allows for numerical exploration of singularities, and motivate 
a deeper study relating them to the topology of $K$.
\end{itemize}

\subsection{Notations} 

Let $\Om$ denote a finite set of indices, 
which we may call the \emph{base space}.
We view the partial order of \emph{regions} $(\Pow(\Om), \subseteq)$ 
as a category with a unique arrow $\bb \to \aa$ whenever $\bb \subseteq \aa$. 
We write $\bb \subset \aa$ only when $\bb$
is a strict subset of $\aa$, and use consistent 
alphabetical order in notations as possible.  
The opposite category, denoted $\Pow(\Om)^{op}$, 
has arrows $\aa \to \bb$ 
for $\bb \subseteq \aa$. 

A free sheaf of \emph{microstates} $E : \Pow(\Om)^{op} \to \Set$ 
will then map every region 
$\aa \subseteq \Om$ to a finite local configuration
space $E_\aa = \prod_{i \in \aa} E_i$. 
In other words, the 
sections of $E_\aa$ are vertex colourings $x_\aa = (x_i)_{i\in \aa}$ of 
a subset $\aa \subseteq \Om$, with colours $x_i \in E_i$ for $i \in \aa$ 
(called local microstates or configurations in physical terminology). 
As a contravariant functor, $E$ maps every inclusion 
$\bb \incl \aa$ to the canonical restriction $\pi^{\aa\to\bb}$,
projecting $E_\aa$ onto  $E_\bb$.
Given a local section 
$x_\aa \in E_\aa$, we write 
$x_{\aa | \bb} = \pi^{\aa\to\bb}(x_\aa)$.

For every region $\aa \subseteq \Om$, we will write $\R^{E_\aa}$ for the finite dimensional
algebra of real observables on $E_\aa$,
write $\R^{E_\aa *}$ for its linear dual, 
and write $\bar \Delta_\aa = \Prob(E_\aa) \subset \R^{E_\aa *}$ 
for the topological simplex of probability measures on $E_\aa$, also called 
{\it states} of the algebra $\R^{E_\aa}$.
The open simplex of {\it positive} states
will be denoted $\Delta_\aa \subset \bar \Delta_\aa$.

%% file: I-graphicalModels.tex
\section{Graphical Models and Generalized Belief Propagation} 
\label{section:GMs}

\subsection{Graphical Models}

In this paper, our main object of study is a special class of 
Markov Random Fields (MRFs), called Gibbsian ensembles in  
\cite{Hammersley-Clifford-71}
although they are now more commonly called 
Graphical Models (GMs) in the decoding and machine learning contexts. 
Given a set of vertices $\Om$ and a random colouring 
$x_\Om = (x_i)_{i \in \Om}$, 
a GM essentially captures the locality of interactions 
by a hypergraph $K \subseteq \Pow(\Om)$ over which the 
density of $x_\Om$ should factorize. 
This property is in general stronger than the Markov properties 
obtained by the Hammersley-Clifford theorem
(as conditional independence of separated regions only ensures factorization 
over cliques \cite{Hammersley-Clifford-71}).

\begin{Definition}
    Given $K \subseteq \Pow(\Om)$ and a collection of {\em factors} 
    $f_\aa : E_\aa \to \R$ for $\aa \in K$, 
    the {\em graphical model} parameterized by $(f_\aa)_{\aa \in K}$ 
    is the probability distribution
    \begin{equation} 
    p_\Om(x_\Om) =  \frac 1 {Z_\Om} 
    \prod_{\aa \in K} f_\aa(x_{\Om| \aa}),
    \end{equation}
    where $Z_\Om = \sum_{E_\Om} \prod_\aa f_\aa$ 
    is an (unknown) integral over $E_\Om$, 
    called the {\em partition function}. 
    \end{Definition}

It is common to represent GMs by their
{\it factor graph} (figure \ref{fig:region-graphs}.b), a bipartite 
graph where variable nodes $i \in \Om$ carry variables $x_i \in E_i$, 
and factor nodes $\aa \in K$ carry local functions $f_\aa : E_\aa \to \R$. 
Factors $\aa \in K$ are then linked to all nodes $i \in \aa \subseteq \Om$. 
However, this graph structure should not be confused with the partial order 
structure we use for message-passing 
(figures \ref{fig:region-graphs}.c and \ref{fig:region-graphs}.d). 

In statistical physics, where the notion of graphical model originates from, 
it is more common to write $p_\Om$ as a normalized 
exponential density $\frac 1 {Z_\Om} \e^{- \beta H_\Om}$ called the 
{\it Gibbs density}. The function $H_\Om : E_\Om \to \R$ 
is called {\it hamiltonian} or {\it total energy} and 
the scalar parameter $\beta = 1 / T$ is called {\it inverse temperature}. 
This variable energy scale is often set to 1 and omitted. 
We recommend references \cite{Jaynes-57} and \cite{Marle-16} 
for deeper thermodynamic background. 

Assuming positivity of $p_\Om$, 
the factorization property of a graphical model $p_\Om$ 
translates to a linear spanning property on the 
global hamiltonian $H_\Om = \sum_\aa - \ln f_\aa$. 
    
\begin{Definition}\label{def:K-local}
We say that a global observable $H_\Om : E_\Om \to \R$ is $K$-{\em local} 
with respect to the hypergraph $K \subseteq \Pow(\Om)$, when there exists 
a family of {\em interaction potentials} $h_\aa : E_\aa \to \R$ 
such that for all $x_\Om \in E_\Om$,
\begin{equation} \label{eq:K-local}
H_\Om(x_\Om) = \sum_{\aa \in K} h_\aa(x_{\Om|\aa}).
\end{equation}
The {\em Gibbs state} of the hamiltonian $H_\Om$ at inverse temperature $\beta > 0$ is 
the positive probability distribution
\begin{equation}
p_\Om(x_\Om) = \frac 1 {Z_\Om} \e^{- \beta H_\Om(x_\Om)},
\end{equation}
and we denote the surjective Lie group morphism from global observables 
to Gibbs states by
\begin{equation} 
    \rho^\beta_\Om : \R^{E_\Om} \to \Delta_\Om.
\end{equation}
The $K$-local {\em Gibbsian ensemble} is the 
image of $K$-local observables under $\rho^\beta_\Om$ (for any $\beta$). 
\end{Definition}

The notions of Gibbsian ensembles \cite{Hammersley-Clifford-71} and graphical models are only equivalent 
up to a positivity assumption on $p_\Om$. 
We always assume positivity of $p_\Om$, 
although this is not always the case in 
decoding applications.  

Let $C_0(K) = \prod_{\aa \in K} \R^{E_\aa}$ denote the 
space of interaction potentials, and write $C_0 = C_0(K)$ by 
assuming $K \subseteq \Pow(\Om)$ to be fixed. 
In general, the map from $(h_\aa) \in C_0$ to 
the global hamiltonian $H_\Om \in \R^{E_\Om}$ has a low-dimensional source space, 
but fails to be injective.
In section \ref{section:Local} we construct 
a chain complex $(C_\bullet, \delta)$, i.e. a graded 
vector space $C_\bullet = \bigoplus_{r=0}^n C_r$ 
and a degree -1 square-null operator $\delta$,
\begin{equation} \label{dgm:resolution}
    \begin{tikzcd}
\R^{E_\Om}  & C_0 \lar[dashed,swap]{\zeta_\Om} & C_1 \lar[swap]{\delta} & \dots \lar[swap]{\delta} 
& C_n \lar[swap]{\delta},
    \end{tikzcd}
\end{equation}
We showed that this complex is acyclic 
is an exact sequence in previous work \cite{gsi19, phd}. 
In other words, $\Ker(\delta) = \Img(\delta)$ and $\Ker(\zeta_\Om) = \delta C_1$.  
Degree-0 homology classes $[h] \in C_0 / \delta C_1$ 
thus yield a bijective parameterization of $K$-local observables, 
this is the statement of theorem \ref{thm:Gauss}. 
The complex structure moreover underlies belief propagation algorithms, 
which explore a homology class $h + \delta C_1 \subseteq C_0$ of 
interaction potentials, and 
GBP is generalized by diffusion equations of the form $\dot{v} = \delta \Phi(v)$ 
in section \ref{section:Diffusions}.

\input{fig-region-graph.tex}

\subsection{Generalized Belief Propagation}

The purpose of Belief Propagation (BP) or sum-product algorithms is to estimate 
the marginals $(p_\aa)_{\aa \in K}$ 
of a GM $p_\Om \in \Delta_\Om$.
Let us precise 
what is meant by {\it beliefs} early: this terminology 
distinguishes them from {\it true marginals} 
of a global distribution $p_\Om$, 
which would admit a genuine probabilistic interpretation,
were they practically computable by partial integration. 
{\it Consistent beliefs} are also called fields of {\it pseudo-marginals}, 
because they satisfy the same gluing axioms as true marginals. 
Their extension to a global measure on $E_\Om$ may not always be 
chosen positive: 
this is known as the {\it pseudo-marginal extension problem}
(see \cite{Vorobev-62,Kellerer-64} and
\cite{Abramsky-2011} for a modern approach with sheaves).
Beliefs relax the consistency constraint, although being purposed 
to reach a consensual state. 

\begin{Definition} \label{def:beliefs}
{\em Beliefs} (over $K$) are collection 
of local positive probability densities $q \in \Delta_0 = \prod_{\aa \in K} \Delta_\aa$. 

Beliefs are said {\em consistent}, written $q \in \Gamma_0$,
if $q_\bb$ is the pushforward of $q_\aa$ under 
$\pi^{\aa \to \bb} : E_\aa \to E_\bb$ for all $\bb \subset \aa$ in $K$. 
\end{Definition}

The Generalized Belief Propagation (GBP) algorithm of Yedidia et al.~\cite{Yedidia-2005} is 
an important generalization of BP equations
to hypergraphs $K \subseteq \Pow(\Om)$ 
(region graphs in their terminology). 
Unlike BP on factor graphs, where messages $M_{\aa \to i}(x_i)$ 
respect the bipartite structure 
and flow only from factors nodes to vertices,  
GBP iterates upon messages $M_{\aa\to \bb}(x_\bb)$ between any two 
regions such that $\aa \supset \bb$, see figure~\ref{fig:region-graphs}. 

\begin{Definition}
Given positive {\it factors} $f \in C_0$ 
(the model parameters, representing conditional kernels and local priors)
and positive {\it initial messages} 
$M^{(0)} \in C_1$ (which may be set to $1$ or absorbed in $f$), the GBP equations
\begin{equation}\label{eq:GBP-beliefs}
    q^{(t)}_\bb(x_\bb) = \frac 1 {Z_\bb} 
    \:\prod_{\cc \subseteq \bb} f_\cc (x_{\bb | \cc})
    \:\prod_{\substack{
            \aa \supseteq \cc\\[.2em]
            \aa \not\subseteq \bb}} M^{(t)}_{\aa\to\cc}(x_{\bb | \cc}),
    \end{equation}
    \begin{equation}\label{eq:GBP-messages}
    M^{(t+1)}_{\aa \to \cc}(x_\cc)
    = M^{(t)}_{\aa \to \cc}(x_\cc) \cdot \frac
    {\displaystyle \sum_{x_{\aa|\cc} = x_\cc} q_\aa^{(t)}(x_\aa)}
    {q^{(t)}_\cc(x_\cc)}
    \end{equation}
define a sequence $(q, M) \in (\Delta_0 \times C_1)^\N$.
\end{Definition}

The evolution of $q^{(t)}$ only depends on the geometric increment of messages 
$M^{(t+1)} / M^{(t)}$. 
Setting $M^{(0)} = 1$, GBP equations therefore also define a dynamical system 
${\rm GBP} : \Delta_0 \to \Delta_0$, such that
$q^{(n)} = {\rm GBP}^{n}(q^{(0)})$.
It is clear from \eqref{eq:GBP-messages} that {\it consistency of
beliefs} $q \in \Gamma_0$ is equivalent to {\it stationarity of messages}. 
However it is not obvious that {\it stationarity 
of beliefs} implies stationarity of messages, and this explains why 
GBP is usually viewed as a dynamical system on messages. 

We showed that stationarity of beliefs does imply 
consistency and stationarity of messages \cite{gsi19,phd}, and
called this property {\it faithfulness 
of the GBP diffusion flux} 
(see definition~\ref{def:proj-faithful} and proposition \ref{prop:GBP-faithfulness} below). 
This property is crucial to retrieve the consistent polytope
$\Gamma_0 \subseteq \Delta_0$ as stationary states of the 
dynamical system on beliefs, 
and thus properly draw the analogy between GBP and diffusion.

\subsection{Examples}

\subsubsection{Markov chains and trees}

A length-$n$ {\em Markov chain} on $\Om = \{ 0, \dots, n \}$ is local with respect 
to the 1-graph $K = K_0 \sqcup K_1 \subseteq \Pow(\Om)$ 
linking every vertex $0 \leq i < n$ 
to its successor $i+1$. Individual states of the Markov chain are denoted $x_i \in E_i$ for $i \in \Om$. 
This data extends to a contravariant functor $E : (K, \subseteq)^{op} \to \Set$
mapping edges $ij \in K$ to pairwise joint states $E_{ij} = E_i \times E_j$, 
and with canonical projections as arrows. 

Given an input prior $f_0(x_0) := \mathbb{P}(x_0)$ and Markov transition kernels 
$f_{{j},{j-1}}(x_j, x_{j-1}) := \mathbb{P}(x_{j} | x_{j-1})$ for $1 \leq j \leq n$, 
set other factors and messages to 1. BP and GBP then coincide to exactly compute in $n$-steps
the output and hidden posteriors 
$q_j(x_j) = \mathbb{P}(x_j)$ for $j \leq n$ as:
\begin{equation} \label{eq:Markov-chain}
\mathbb{P}(x_j) = \sum_{x_{j-1} \in E_{j-1}} \mathbb{P}(x_j | x_{j-1}) \: \mathbb{P}(x_{j-1}).
\end{equation}
In this particular case, the sum-product update rule 
thus simply consists in a recurrence of matrix-vector multiplications 
for $M_{j-1, j \to j}(x_j) = q_j(x_j) = \mathbb{P}(x_j)$.
The integrand of \eqref{eq:Markov-chain} 
also yields the exact pairwise probability
$\mathbb{P}(x_{j-1}, x_j)$ by the Bayes rule. 

The situation is very much similar for {\em Markov trees}, i.e.
when $K$ is an acyclic graph. 
In this case, choosing any leaf node $x_0$ as root, BP recursively
computes the bayesian posteriors $\mathbb{P}(x_j)$ exactly,
by integrating the
product $\mathbb{P}(x_j|x_i) \, \mathbb{P}(x_i)$
over the parent state $x_i \in E_i$ 
when computing the message $M_{ij \to j}(x_j)$ 
by \eqref{eq:GBP-messages}. 
This is the content of Pearl's original article on 
Bayesian trees \cite{Pearl-82}. 

A famous particular class of Markov trees consists 
of hidden Markov chain models, for instance used by 
Friston in neuroscience \cite{Friston-Parr-2017}, 
and of which the K{\'a}lm{\'a}n filter is 
a continuous-valued gaussian version \cite{Kschischang-01}. 

\subsubsection{Spin glasses and Hopfield networks}

A {\em spin glass} is a 1-graph $K \subseteq \Pow(\Om)$
whose vertices carry a binary variable $x_i \in \{ -1, +1\} = E_i$. 
In this case, any $K$-local hamiltonian $H_\Om : E_\Om \to \R$ 
may be uniquely decomposed as: 
\begin{equation} \label{eq:SpinGlass}
H_\Om(x_\Om) = h_\varnothing + \sum_{i \in K_0} b_i x_i 
+ \sum_{ij \in K_1} w_{ij} x_i x_j 
\end{equation}
Uniqueness of this decomposition is a consequence of the {\it interaction decomposition theorem}, 
see \cite{Kellerer-64} and \cite{Matus-88}. 
One often drops the constant factor $h_\varnothing$ 
for its irrelevance in the Gibbs state definition. 

Spin glasses are formally equivalent to the Ising model 
of ferromagnetism, the only conceptual 
differences residing in 
a random sampling of {\it biases} $b_i \sim \mathbb{P}(b_i)$ 
and {\it weights} $w_{ij} \sim \mathbb{P}(w_{ij})$, 
and in allowing for other graphs than cubic lattices which describe 
homogeneous crystals. 
Bethe produced his now famous combinatorial approximation scheme 
in 1935 \cite{Bethe-35} to estimate the Ising model free energy. 

BP works surprisingly well 
at estimating the true likelihoods $\mathbb{P}(x_j)$ and 
$\mathbb{P}(x_i, x_j)$ by 
consistent beliefs $q_i(x_i)$ and $q_{ij}(x_i, x_j)$ 
even on loopy graphs.
It may however show converge issues 
and increased errors as the temperature decreases 
when $K$ has loops, while the number of stationary states grows 
quickly with the number of loops. 
In the cyclic case, each belief aggregates a 
product of messages $\prod_i M_{ij \to j}(x_j)$ 
according to \eqref{eq:GBP-beliefs} 
before being summed over in the computation of messages \eqref{eq:GBP-messages}, 
and therefore does not reduce to a simple matrix-vector multiplication. 
See \cite{Weiss-97,Murphy-Weiss-99,Knoll-17} for more information 
on loopy BP behaviour. 

Spin glasses are also called Boltzmann machines \cite{Ackley-85}
in the context of generative machine learning. 
It is well known that bipartite spin glasses, also called 
restricted Boltzmann machines (RBMs), are equivalent to 
the Hopfield model of associative memory 
whose phase transitions \cite{Bovier-01, Barra-18} 
have received a lot of attention. 
These phase diagrams are generally obtained by the replica method 
of Javanmard and Montanari \cite{Javanmard-13} in the thermodynamic limit. 
We believe a broader understanding of such phase transitions 
in neural networks would be highly beneficial for artificial intelligence, 
and a bayesian framework seems more suitable for such a program 
than one-way parameterized functions and feed-forward neural networks.

\subsubsection{Higher-order relations}

When $K \subseteq \Pow(\Om)$ is a general hypergraph,
we may write $K = \bigsqcup_{r} K_r$ by grading hyperedges 
$\aa = \{i_0, \dots, i_r\}$ according to their dimension $r$,
and write $K_{-1} = \{ \varnothing \}$ when $K$ 
contains the empty region
(which we usually assume along with the $\cap$-closure of $K$).
We call {\it dimension} of $K$ the greatest $n$ such that $K_n \neq \varnothing$,
and call $K$ an $n$-{\it graph}.  

Working with a coarse $n$-graph $K$ with $n \geq 2$ is useful even when
the hamiltonian is local with respect to a 
$1$-subgraph $K' \subseteq K$. 
This is done 
by simply extending $K'$-local factors $(f_\bb)_{\bb \in K'}$ by 
$f_\aa = 1$ for all $\aa \in K \smallsetminus K'$.
The dimension of the coarser hypergraph $K$ used for message-passing  
thus provides greater precision at the cost of local complexity, 
which is the cost of partial integrals on $E_\aa$ for maximal $\aa \in K$.
From a physical perspective, this corresponds to applying Kikuchi's Cluster 
Variation Method (CVM) \cite{Kikuchi-51,Morita-57} to a spin glass 
hamiltonian $H_\Om$, given by \eqref{eq:SpinGlass}. 

There is also greater opportunity in considering higher order
interactions $h_\aa(x_\aa) = h_{i_0, \dots, i_r}(x_{i_0}, \dots, x_{i_r})$ 
for $r \geq 2$ to capture more subtle relations. Taking $r = 2$ with binary variables, 
the interaction terms $w_{ijk} x_i x_j x_k$ mimics attention in transformer 
networks. In a continuous-valued case, 
where $x_i \in \R^3$ for instance describes atomic positions in a 
molecule or crystal, energy couldn't be made dependent on
bond angles without third-order interactions 
$h_{ijk}(x_i, x_j, x_k)$ \cite{a3i-23}. 
Continuous variables are outside the scope of this article 
but we refer the reader to \cite{Weiss-01, Gouillart-13, Mezard-2017} 
the Gaussian version of BP algorithms. 
See also \cite{Sejnowski-87} on third-order Boltzmann machines 
and \cite{Goh-22,Hajij-23} for more recent higher-order attention network architectures. 

%% file: fig-region-graph.tex
\begin{figure*}[t] 
\begin{center}

\includegraphics[width=.85     \textwidth]
{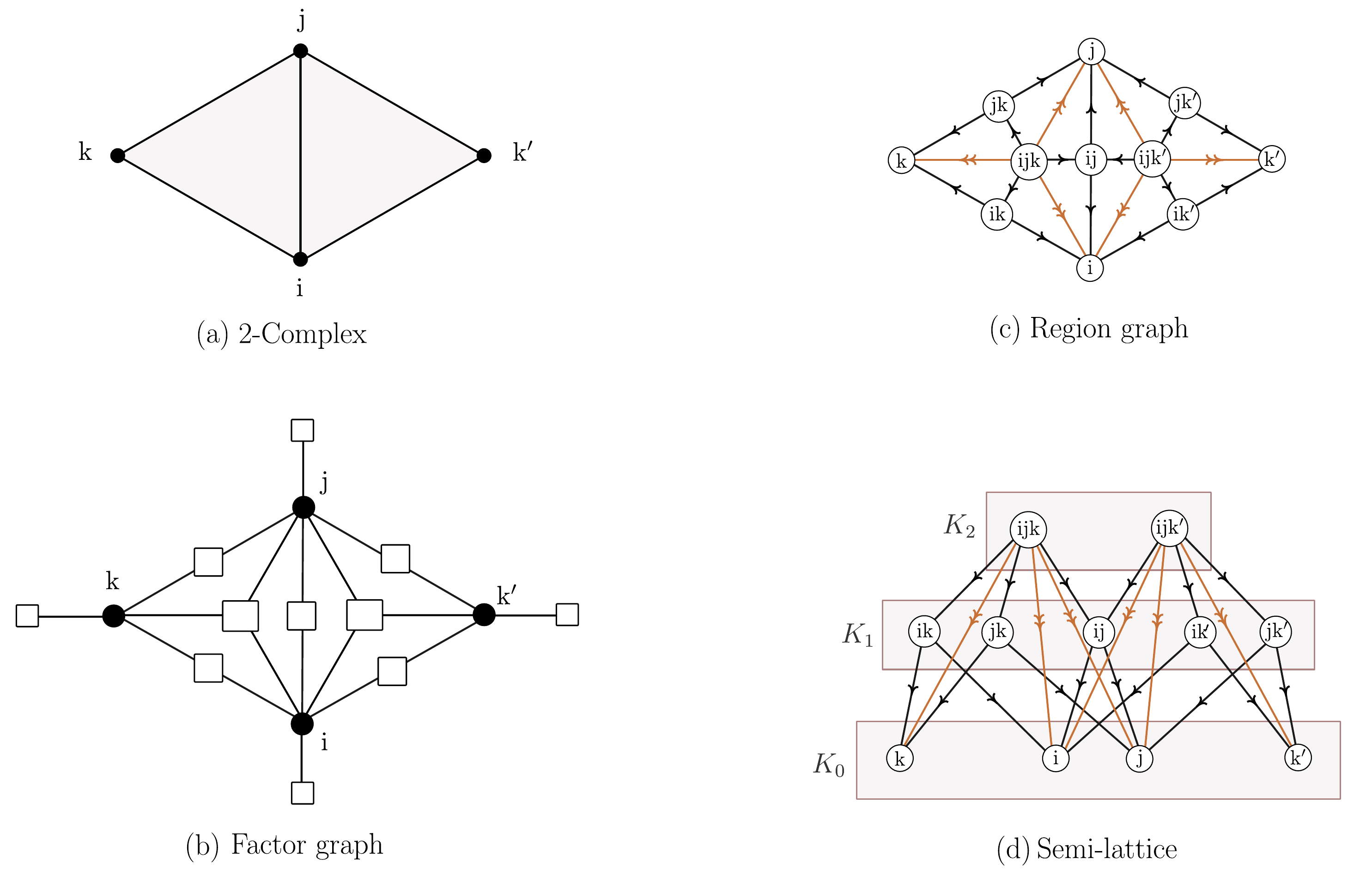}

\centering
\vspace{.2em}
    \caption{\label{fig:region-graphs}{\small
        Different representations of $K \subseteq \Pow(\Om)$. 
        Here $K$ is a 2-dimensional simplicial 
        complex (a). The factor graph (b) is constructed by joining 
        every factor node $\aa \in K$ (squares) to its variables $i \in \aa$ (dots). 
        The region graph (c) instead inserts a directed edge 
        between any two region nodes $\aa, \bb \in K$ such that $\aa \supseteq \bb$.
        The partial order structure is equivalently represented by (d), 
        where height gives a better impression of the ordering. 
        In (c) and (d), non-primitive arrows (i.e. having a non-trivial factorization) 
        are in orange, and the terminal region $\varnothing$ is not represented. 
    }}
\end{center}
\end{figure*}

%% file: II-i-Boltzmann-Gibbs.tex
\section{Entropy and Free Energy Principles} \label{section:Principles}

\subsection{Boltzmann-Gibbs variational principles}

In classical thermodynamics, a statistical system $\Om$ is described by 
a configuration space $E$ (assumed finite in the following)
and a \emph{hamiltonian} 
$H : E \to \R$ measuring the energy level of each configuration. 
Thermal equilibrium with a reservoir at \emph{inverse temperature}
$\beta = 1 / T$ defines the so-called 
\emph{Gibbs state} $p^\beta \in \Prob(E)$ 
by renormalization of the Boltzmann-Gibbs density $\e^{-\beta H}$:
\begin{equation}\label{eq:gibbs}
    p^\beta(x) = \frac 1 Z \e^{-\beta H(x)} 
    \quad{\mathrm{where}}\quad Z = \sum_{x \in E} \e^{-\beta H(x)}. 
\end{equation}
Two different kinds of variational principles characterise the equilibrium state 
$p^\beta$:
\begin{itemize}
\item the \emph{max-entropy principle}
asserts that the Shannon entropy $S(p) = - \sum_E p \ln(p)$ is 
maximal under an internal energy constraint $\E_p[H] = {\U}$ at equilibrium;
\item the \emph{free energy principle} 
asserts that the 
variational free energy ${\cal F}^\beta(p, H) = \E_p[H] - T S(p)$
is minimal under the temperature constraint $T = 1 / \beta$ at equilibrium; 
it is then equal to the free energy $F^\beta(H) = - \ln \sum_E \e^{- \beta H}$. 
\end{itemize}
Although equation \eqref{eq:gibbs} gives the solution to both optimisation problems, 
computing $Z$ naively is usually impossible for 
the size of $E$ grows exponentially with the number of 
microscopic variables in interaction. 

The Legendre duality relating
Shannon entropy and free energies is illustrated by theorems 
\ref{thm:S-global} and \ref{thm:F-global} below, implying 
the equivalence between entropy and free energy Boltzmann-Gibbs principles
~\cite{Jaynes-57,Marle-16}. 
Properties of thermodynamic functionals may be abstracted 
from the global notion of physical equilibrium, and 
stated for every region $\aa \subseteq \Om$ as below. 
Local functionals will then be recombined 
following the Bethe-Kikuchi approximation scheme in the next 
subsection.  

\begin{Definition} \label{def:functionals}
For every $\aa \incl \Om$ and $\beta > 0$, 
define: 
\begin{itemize}
    \item 
    the {\em Shannon entropy} 
    $S_\aa : \Delta_\aa \to \R$ by 
    \begin{equation}
    S_\aa(p_\aa) = - \sum p_\aa \ln(p_\aa),
    \end{equation}

    \item 
    the {\em variational free energy}
    $\varF^\beta_\aa : \Delta_\aa \times \R^{E_\aa} \to \R$
     by 
    \begin{equation}
    \varF^\beta_\aa(p_\aa, H_\aa) = \E_{p_\aa}[H_\aa] 
    - \frac 1 \beta S_\aa(p_\aa),
    \end{equation}
    \item 
    the {\em free energy}
    $\F^\beta_\aa : \R^{E_\aa} \to \R$ by 
    \begin{equation}
    \F^\beta_\aa(H_\aa) = - \frac 1 \beta \ln \sum \e^{- \beta H_\aa}.
    \end{equation}
\end{itemize}
\end{Definition}

Legendre transforms may be carried with respect to local observables 
and beliefs directly, instead of the usual one-dimensional 
temperature or energy parameters. 
To do so, one should describe tangent fibers of 
$\Delta_\aa$ as 
$\Tg_{p_\aa} \Delta_\aa \simeq \R^{E_\aa} {\rm\; mod\;} \R$, 
so that additive energy constants span Lagrange multipliers for 
the normalisation constraint $\langle p_\aa, 1_\aa \rangle = 1$. 
For a more detailed study of thermodynamic functionals, 
we refer the reader to \cite[chap. 4]{phd} and \cite{Marle-16}.

\begin{Theorem} \label{thm:S-global}
Under the mean energy constraint 
$\langle p_\aa, H_\aa \rangle = {\cal U}$, 
the maximum of Shannon entropy is reached on a Gibbs state 
$p_\aa^* = \frac 1 {Z_\aa} \e^{-\beta H_\aa}$, 
\begin{equation}
    S_\aa(p_\aa^*) = \max_{\substack{ 
        p_\aa \in \Delta_\aa \\
        \langle p_\aa, H_\aa \rangle = {\cal U} }}
    S_\aa(p_\aa),
\end{equation}
for some univocal value of the Lagrange multiplier $\beta \in \R$, 
called {\em inverse temperature}. 
\end{Theorem}

\begin{Theorem} \label{thm:F-global}
Under the temperature constraint $T = 1 / \beta$, given a hamiltonian 
$H_\aa \in \R^{E_\aa}$, the minimum 
of variational free energy is reached on the Gibbs state 
$p_\aa^* = \frac 1 {Z_\aa} \e^{-\beta H_\aa}$,
\begin{equation} \label{eq:min-varF}
\varF_\aa^\beta(p_\aa^*, H_\aa)
= \min_{p_\aa \in \Delta_\aa} \varF_\aa^\beta(p_\aa, H_\aa). 
\end{equation}
It moreover coincides with the equilibrium free energy $F^\beta_\aa(H_\aa)$,
\begin{equation}
F_\aa^\beta(H_\aa) = \min_{p_\aa \in \Delta_\aa} \varF_\aa^\beta(p_\aa, H_\aa)
\end{equation}
\end{Theorem}

Although stated for every region $\aa \subseteq \Om$,
theorem \ref{thm:F-global} only makes sense physically when 
applied to the global region $\Om$. Indeed, 
local free energy principles \eqref{eq:min-varF} 
on regions $\aa \subseteq \Om$ totally neglect interactions 
with their surroundings. The local beliefs $(p_\aa)_{\aa \in K}$ 
they define thus have very little chance of being consistent.

%% file: II-ii-Bethe-Kikuchi.tex
\subsection{Bethe-Kikuchi variational principles}

We now proceed to define localised versions 
of the max-entropy and free energy variational principles \ref{thm:S-global} and \ref{thm:F-global}, 
attached to any hypergraph $K \subseteq \Pow(\Om)$. 
We recall that $C_0$ stands for the space of interaction 
potentials $\prod_{\aa \in K} \R^{E_\aa}$, 
and that $\Gamma_0 \subseteq \Delta_0 \subseteq \R^{E_\aa *}$ denotes the convex polytope 
of consistent beliefs (definition \ref{def:beliefs}).
Bethe-Kikuchi principles
will characterise finite sets 
of consistent local beliefs
$\{ p^1, \dots, p^m \} \subseteq \Gamma_0$,
in contrast with their global counterparts
defining the true global Gibbs state 
$p_\Om \in \Delta_\Om$. 

{\em Extensive} thermodynamic functionals
such as entropy and free energies 
satisfy an {\em asymptotic additivity} 
property, e.g. the entropy of a large piece of matter is 
the sum of entropies associated to any division into
large enough constituents. 
Bethe-Kikuchi approximations thus consist of
 computing only local terms 
(that is, restricted to a tractable number of variables) 
before {\em cumulating} them in a large 
weighted sum over regions $\aa \in K$ 
where integral coefficients 
$c_\aa \in \Z$ take care of eliminating redundancies. 
The coefficients $c_\aa$ are uniquely determined by the inclusion-exclusion 
principle $\sum_{\aa \supseteq \bb} c_\aa = 1$ for all $\bb \in K$ 
(corollary \ref{cor:BK-coeffs}).
For a recent introduction to the CVM 
\cite{Kikuchi-51,Morita-57}, we refer to 
Pelizzola's article \cite{Pelizzola-05}. 

Let us first introduce the local max-entropy principle
we shall be concerned with.
As in the exact global case, the max-entropy principle 
(problem \ref{S-crit}) describes an isolated system. It therefore 
takes place with constraints 
on the Bethe-Kikuchi mean energy 
$\check {\cal U} : \Delta_0 \times C_0 \to \R$, given 
for all $p \in \Delta_0$ and all $H \in C_0$ by:
\begin{equation}\label{eq:U-Bethe}
\check U(p, H) = \sum_{\aa \in K} c_\aa \: \E_{p_\aa}[H_\aa]
\end{equation}

\begin{Problem}\label{S-crit}
Let $H \in C_0$ denote local hamiltonians and 
choose a mean energy $\cal U \in \R$.
\\
Find beliefs $p \in \Delta_0$  critical for the {\em Bethe-Kikuchi entropy}
$\check S : \Delta_0 \to \R$ given by:
\begin{equation}
\check S(p) = \sum_{\aa \in K} c_\aa \: S_\aa(p_\aa),
\end{equation}
under the consistency constraint 
$p \in \Gamma_0$ and the mean energy constraint $\check\U(p, H) = \U$.
\end{Problem} 

The Bethe-Kikuchi variational free energy principle, 
problem~\ref{varF-crit} below, instead 
serves as a sound local substitute for describing 
a thermostated system
at temperature $T = 1/\beta$. 

\begin{Problem} \label{varF-crit}
Let $H \in C_0$ denote local hamiltonians
and choose an inverse temperature $\beta > 0$.
\\
Find beliefs $p \in \Delta_0$ critical for the 
{\em Bethe-Kikuchi variational free energy} 
$\check \varF : \Delta_0 \times C_0 \to \R$ given by:
\begin{equation}
\check\varF^\beta(p, H) = \sum_{\aa \in K} c_\aa \: \varF^\beta_\aa(p_\aa, H_\aa),
\end{equation}
under the consistency constraint $p \in \Gamma_0$. 
\end{Problem}

Problems \ref{S-crit} and \ref{varF-crit} are
both variational principles on $p \in \Gamma_0$
with the consistency constraint in common,
but with dual temperature and energy constraints respectively.
In contrast, the following free energy principle 
(problem \ref{F-crit}) explores a subspace of
local hamiltonians $V \in C_0$, satisfying what we shall view 
as a global energy conservation constraint
in the next section. Let us write $V \check\sim H$ if and 
only if $\sum_\aa c_\aa V_\aa = \sum_{\aa}c_\aa H_\aa$ 
as global observables on $E_\Om$. 

Problem \ref{F-crit} also describes a system at equilibrium with a thermostat
at fixed temperature $T = 1 / \beta$.

\begin{Problem} \label{F-crit}
Let $H \in C_0$ denote local hamiltonians and choose an inverse
temperature $\beta > 0$. 
\\
Find local hamiltonians $V \in C_0$
critical for the {\em Bethe-Kikuchi free energy}
$\check F^\beta : C_0 \to \R$ given by:
\begin{equation} \label{F_Bethe}
    \check \F^\beta(V) =
       \sum_{\aa \in K} c_\aa \: F^\beta_\aa(V_\aa),
\end{equation}
under the energy conservation constraint $V \check\sim H$. 
\end{Problem}

In contrast with the {\em convex} optimisation 
problems of theorems \ref{thm:S-global} and \ref{thm:F-global}, 
note that the concavity or convexity 
of information functionals 
is broken 
by the Bethe-Kikuchi coefficients $c_\aa \in \Z$. 
This explains why multiple solutions to problems 
\ref{S-crit}, \ref{varF-crit} and \ref{F-crit} 
might coexist, and why they cannot be found 
by simple convex optimisation algorithms. 
We will instead introduce continuous-time ordinary 
differential equations in section \ref{section:Diffusions}. 
Their structure is 
remarkably similar to diffusion or heat equations, 
although combinatorial transformations may again break 
stability and uniqueness of stationary states. 

On the Ising model, Schlijper showed that the CVM error decays exponentially 
as $K \incl \Pow(\Z^d)$ grows coarse with respect to 
the range of interactions~\cite{Schlijper-83}. This result 
confirms the heuristic argument on the extensivity of entropy, 
and reflects the fast decay of high-order mutual informations 
that are omitted in the Bethe-Kikuchi entropy \cite[chap. 4]{phd}.

%% file: III-localised.tex
\section{Local statistical systems} \label{section:Local}

In the following, we let $K \subseteq \Pow(\Om)$
denote a fixed hypergraph with vertices in 
$\Om = \{1, \dots, N\}$, and  
moreover assume that 
$(K, \subseteq, \cap)$ forms a semi-lattice.
The following constructions could be carried without 
the $\cap$-closure assumption, 
but theorem \ref{thm:Gauss} and the correspondence theorems 
of section~\ref{section:Equilibria} 
would become one-way.

This section carries the construction of what we may call a 
\emph{combinatorial chain complex} 
$(C_\bullet, \delta, \zeta)$ of local observables, 
recalling the necessary definitions and theorems from \cite{phd}.
The first ingredient is the codifferential $\delta$, 
satisfying $\delta^2 = \delta\delta = 0$ 
and of degree -1: 
\begin{equation} \label{eq:C-delta}
    C_0 \overset{\delta}{\longleftarrow} 
    C_1 \overset{\delta}{\longleftarrow} \;\dots\; \overset{\delta}{\longleftarrow}
    C_n.
\end{equation}
The first homology $[C_0] = {\rm H}_0(C_\bullet, \delta)$ 
will yield a bijective parameterization of $K$-local observables 
in $\R^{E_\Om}$, this is the statement of the Gauss theorem \ref{thm:Gauss}.
On the other hand, the dual cochain complex 
\begin{equation} \label{eq:C*-d}
    C_0^* \overset{d}{\longrightarrow} 
    C_1^* \overset{d}{\longrightarrow} \;\dots\; \overset{d}{\longrightarrow}
    C_n^*
\end{equation}
will allow us to describe consistent beliefs $p \in \Gamma_0$ 
by the cocycle equation $dp = 0$, living in the 
dual cohomology $[C_0^*] = {\rm H}^0(C_\bullet^*, d)$.
The construction of $(C_\bullet^*, d)$ may be traced back to 
Grothendieck and Verdier
\cite{SGA-4-V, Moerdijk}, yet we believe the interaction 
of algebraic topology with combinatorics presented here to be quite original.  

The \emph{zeta transform}  is 
here defined as a homogeneous linear automorphism 
$\zeta : C_\bullet \to C_\bullet$, and plays a role very similar 
to that of a discrete spatial integration, 
confirming an intuition of Rota. It satisfies remarkable commutation relations with $\delta$, 
the Gauss/Greene formulas \eqref{eq:Gauss-r} and \eqref{eq:mu-Gauss}.
Its inverse $\mu$ is called the \emph{Möbius transform} and  
the reciprocal pair $(\zeta, \mu)$ extends the 
famous Möbius inversion formulas \cite{Rota-64} to degrees higher than 1 
in the nerve of $K$.
All these operators will come into play when factorizing 
the GBP algorithm and defining Bethe-Kikuchi diffusions 
in section \ref{section:Diffusions}. 

Our localization procedure could be summarized as follows.
First, we restrict the sheaf $E$ 
to a contravariant 
functor $E_K$ over $K$;  
then,
we define a simplicial set $N_\bullet E_K$ 
extending $E_K$ to the categorical\footnote{
    Not to be confused with the {\v C}ech nerve of a covering, 
    often used to define sheaf cohomology. 
    They only differ by a barycentric subdivision so
    that the two cohomology theories 
    are isomorphic~\cite{agt23}.    
}
nerve $N_\bullet K$ by mapping every strictly ordered chain 
$\aa_0 \supset \dots \supset \aa_p$ to its {\it terminal} 
configuration space $E_{\aa_0\dots \aa_p} := E_{\aa_p}$.
In particular, the set $N_1 K$ describes the support of 
GBP messages, while higher order terms provide a 
projective resolution of $K$-local observables. 

\subsection{Algebraic topology}

The functor of {\it local observables}
$\R^E : \Pow(\Om) \to \Alg_{c}$ 
maps every region $\aa \subseteq \Om$ 
to the commutative algebra $\R^{E_\aa}$.
Its arrows consist of natural inclusions $\R^{E_\bb} \incl \R^{E_\aa}$,
when identifying each local algebra with a low dimensional subspace of 
$\R^{E_\Om}$.
We write $\tilde{h}^\aa_\bb : x_\aa \mapsto h_\bb(x_{\aa|\bb})$ when
the extension should be made explicit, and
identify $\tilde{h}^\aa_\bb$ with $h_\bb$ for all $\aa \supseteq \bb$
otherwise. 

One may then define a chain complex
of local observables $C_\bullet = C_\bullet(K, \R^E)$ 
indexed by the nerve of $K$ as follows. 
Its graded components $C_r$ are defined for $1 \leq r \leq n$, 
where $n$ denotes the maximal length of a strictly ordered chain in $K$, by:
\begin{equation} \label{eq:C-r}
    C_r = 
    \prod_{\aa_0 \supset \dots \supset \aa_r} \R^{E_{\aa_r}}\,.
\end{equation}
For every strict chain $\vec\aa = \aa_0 \supset \dots \supset \aa_r$ 
in $N_r K$, and every $0 \leq j \leq r$, let us denote
by $\vec\aa^{(j)}$ the $j$-\emph{face} of $\vec\aa$, 
obtained by removing $\aa_j$. 

\begin{Definition} \label{def:chainComplex}
    The chain complex of {\rm local observables} $(C_\bullet, \delta)$ 
    is defined by \eqref{eq:C-r} 
    and the degree -1 boundary operator $\delta$,
    whose action $\delta : C_1 \to C_0$ is given by 
    \begin{equation} \label{eq:delta-1}
        \delta\ph_\bb(x_\bb) = \sum_{\aa \supset \bb} \ph_{\aa \to \bb}(x_\bb)
        - \sum_{\bb \supset \cc} \ph_{\bb \to \cc}(x_{\bb|\cc}),
    \end{equation}
    while $\delta : C_{r+1} \to C_r$ acts by
    \begin{equation} \label{eq:delta-r}
        \delta\psi_{\vec\bb}(x_{\bb_r}) = 
        \sum_{j = 0}^{r+1}
        (-1)^j \sum_{\vec\cc^{(j)} = \vec\bb} 
        \psi_{\vec\cc}(x_{\bb_r | \cc_{r+1}}).
    \end{equation}
    Note that $\bb_r = \cc_{r+1}$ for all 
    $j < r + 1$ in the first $r$ sums of \eqref{eq:delta-r}.
\end{Definition}

The classical identity 
${\bar \aa}^{(i)(j)} = {\bar \aa}^{(j+1)(i)}$ for all $i < j$, 
together with linearity and functoriality of the inclusions
$\R^{E_\bb} \subseteq \R^{E_\aa}$,
implies the differential rule $\delta^2 = \delta \circ \delta = 0$. 

One may see in formula \eqref{eq:delta-1} 
an analogy between $\delta$ and a discrete graph divergence, 
or the divergence operator of geometry : 
theorem \ref{thm:A} below is a discrete yet statistical version of
the Gauss theorem on a manifold without boundary. 
It gives a local criterion for the 
global equality of $K$-local hamiltonians 
and is proven in appendix \ref{section:apx-Gauss}.
The chain complex $(C_\bullet, \delta)$ can moreover be proven acyclic 
(see \cite[thm 2.17]{phd} and \cite{SGA-4-V,Moerdijk}) 
when $K$ is $\cap$-closed, 
the exact sequence \eqref{eq:C-delta} then describes
the linear subspace of 
$K$-local energies,
through a \emph{projective resolution} of
the quotient $C_0 / \delta C_1$.
\begin{Theorem}[Gauss] \label{thm:A}\label{thm:Gauss}
Assume $K \subseteq \Pow(\Om)$ is $\cap$-closed. 
Then the following are equivalent for all $h, h' \in C_0$:

\begin{itemize}
\item[(i)] the equality 
$\sum_{\aa \in K} h'_\aa = \sum_{\aa \in K} h_\aa$ holds in $\R^{E_\Om}$,
\item[(ii)] there exists $\ph \in C_1$ such that 
$h' = h + \delta \ph$.
\end{itemize}
In other words $h'$ and $h$ are homologous, written 
$h' \sim h$ or $h' \in [h]$, if and only if they 
define the same global hamiltonian. 
\end{Theorem}

The functor of {\it local densities} 
$\R^{E*}: \Pow(\Om)^{op} \to \Vect$ is then defined by duality, 
mapping every $\aa \subseteq \Om$ to the vector space of linear forms on $\R^{E_\aa}$.
Its arrows consist of partial integrations 
$\pi_*^{\aa \to \bb} : \R^{E_\aa *} \to \R^{E_\bb *}$,
also called marginal projections. 
This functor generates a dual cochain complex $(C_\bullet^*, d)$
which shall serve to describe pseudo-marginals 
$(p_\aa)_{\aa \in K} \in \Gamma_0 \subset C_0^*$, 
used as substitute for global probabilities 
$p_\Om \in \Delta_\Om$. 

\begin{Definition}
The cochain complex of \emph{local densities} $(C_\bullet^*, d)$
is defined by $C_r^* = L(C_r)$
and the degree $+1$ differential $d$, whose 
action $d : C_0^* \to C_1^*$ is given by
\begin{equation}
    d p_{\aa \to \bb}(x_\bb) = p_\bb(x_\bb) - \sum_{x_{\aa|\bb} 
    =\, x_\bb} p_\aa(x_\aa).
\end{equation}
while $d : C_r^* \to C_{r+1}^*$ acts by
\begin{equation} 
    \begin{split}
    dq_{\vec \aa}(x_{\aa_{r+1}}) =& \sum_{j = 0}^{r} 
    (-1)^j \, q_{\vec \aa^{(j)}}(x_{\aa_{r+1}}) \\
    &+ 
    (-1)^{r+1} \sum_{x_{\aa_r|\aa_{r+1}} = x_{\aa_{r+1}}} 
    q_{\vec \aa^{(r+1)}}(x_{\aa_{r}})
    \end{split}
\end{equation}
Densities satisfying $dp = 0$ are called \emph{consistent}. 
The convex polytope of consistent beliefs 
is $\Gamma_0 = \Delta_0 \cap \Ker(d) \subseteq C_0^*$.
\end{Definition}

The purpose of BP algorithms and their generalizations 
is to converge towards consistent beliefs $p \in \Gamma_0$. 
The non-linear {\it Gibbs correspondence}
relating potentials $h \in C_0$ 
to beliefs $p \in C_0^*$ will thus
be essential to 
the dynamic of GBP:
\begin{equation}
    p_\aa(x_\aa) = \frac 1 {Z_\aa} 
    \e^{- \beta H_\aa(x_\aa)}
    \quad\mathrm{with}\quad
    H_\aa(x_\aa) = \sum_{\bb \incl \aa} h_\bb(x_{\aa|\bb}).
\end{equation}
The mapping $h \mapsto H$ above is an invertible Dirichlet convolution 
\cite{Rota-64}
on $C_0$, analogous to a discrete integration over cones $K^\aa \subseteq K$.
Although seemingly simple, this mapping $\zeta$ and its inverse 
$\mu$ surely deserve 
proper attention.

\subsection{Combinatorics}

The combinatorial automorphisms $\zeta$ and $\mu$ 
we describe below generalize well-known
\emph{Möbius inversion formulas} on $C_0$ and $C_1$. 
The convolution structure in degrees $r \leq 1$ 
originates from works of Dirichlet in number theory, 
and was thoroughly described by 
Rota \cite{Rota-64} on general partial orders; 
let us also mention
the interesting extension \cite{Leinster-08} to more general categories.

Heuristically, $\zeta$ and $\mu$ 
might be viewed as combinatorial mappings from intensive to extensive local 
observables and reciprocally, which 
systematically solve what are known as 
\emph{inclusion-exclusion principles}.

\input{fig-gauss.tex}

\begin{Definition} \label{def:zeta}
The \emph{zeta transform} $\zeta : C_\bullet \to C_\bullet$ 
is the linear homogeneous morphism acting on $C_0$ by letting
$\zeta : h \mapsto H$,
\begin{equation} 
    H_\aa(x_\aa) = \sum_{\bb \subseteq \aa} 
    h_\bb(x_{\aa|\bb}),
\end{equation}
and acting on $C_r$ by letting $\zeta : \ph \mapsto \Phi$,
\begin{equation}\label{eq:zeta-r}
\Phi_{\vec \aa}(x_{\aa_r}) 
= \sum_{\bb_r \subseteq \aa_r} \dots 
\sum_{\substack{\bb_0 \subseteq \aa_0 \\ \bb_0 \not\subseteq \aa_1}}
\ph_{\vec \bb}(x_{\aa_r|\bb_r})
\end{equation}
Note that $\bb_0 \supset \dots \supset \bb_r$ is also
implicitly assumed in \eqref{eq:zeta-r}.
\end{Definition}

Definition \ref{def:zeta}
extends to $C_\bullet$ the action of $\zeta$ on $C_0$. 
The action $\zeta : C_1 \to C_1$ 
defined by
\eqref{eq:zeta-r} should not be confused with the convolution 
product of the incidence algebra $\tilde C_1(K, \Z)$, 
obtained by including degenerate chains (identities) $\aa \supseteq \bb$ 
in $\tilde N_1 K \simeq N_0 K \sqcup N_1 K$, 
and restricting to 
integer coefficients.
See for instance Rota \cite{Rota-64} for details on Dirichlet convolution, 
and \cite[chap. 3]{phd} for the 
module structure considered here. 

Remember we assume $K \incl \Pow(\Om)$ to be $\cap$-closed\footnote{
    This coincides with the {\em region graph property} 
    of Yedidia, Freeman and Weiss \cite{Yedidia-2005}, when 
    describing the hypergraph $K$ in their language of 
    bipartite region graphs.
}
for theorem \ref{thm:A} to hold. This
is also necessary to obtain the explicit 
Möbius inversion formula \eqref{eq:mu-r} below. 

\begin{Theorem} \label{thm:Mobius}
    The \emph{Möbius transform} $\mu = \zeta^{-1}$ is given  
    in all degrees by a finite sum: 
    \begin{equation}\label{eq:mu-rec}
    \mu = \sum_{k = 0}^n (-1)^k (\zeta - 1)^k
    \end{equation}
    The action of $\mu : C_0 \to C_0$ may be written 
    $\mu : H \mapsto h$, 
    \begin{equation}\label{eq:mu-0}
        h_\aa(x_\aa) = \sum_{\bb \incl \aa} 
        \mu_{\aa \to \bb} \: H_{\bb}(x_{\aa|\bb}),
    \end{equation}
    and the action $\mu : C_r \to C_r$ may be written 
    $\mu : \Phi \to \ph$, 
    \begin{equation} \label{eq:mu-r}
    \ph_{\vec \aa}(x_{\aa_r})= 
    \sum_{\bb_r \subseteq \aa_r} 
    \mu_{\aa_r \to \bb_r}
    \dots \sum_{\substack{\bb_0 \subseteq \aa_0\\
    \bb_0 \not\subseteq \bb_1}}
    \mu_{\aa_0 \to \bb_0} \Phi_{\vec \bb^\cap}(x_{\aa_r | 
    \bb^\cap_r})
    \end{equation}
    where $\vec\bb^\cap := \bb_0 \supset 
    (\bb_0 \cap \bb_1) \supset \dots \supset (\bb_0 \cap \dots \cap \bb_r)$ 
    in \eqref{eq:mu-r}. 
\end{Theorem}

In \eqref{eq:mu-rec}, $n$ is the maximal length of a strict 
chain in $K$. 
In degree 0, this yields the usual recurrence formulas 
for $\mu_{\aa \to \bb}$ in terms of all the strict factorisations of $\aa \to \bb$.
In practice, the matrix $\mu$ can be computed efficiently 
in a few steps by using any 
sparse tensor library. 
When $K$ describes a graph, 
one for instance has $(\zeta - 1)^3 = 0$. 
The nilpotency of $\zeta - 1$ 
furthermore ensures that $\mu = \zeta^{-1}$ is given 
by \eqref{eq:mu-rec} even without the $\cap$-closure assumption. 
The reader is referred to \cite[thm 3.11]{phd}
for the detailed proof of~\eqref{eq:mu-r}, 
briefly summarized below.  

\begin{proof}[Sketch of proof]
The definitions of $\zeta$ and $\mu$ both exhibit 
a form of recursivity in the degree $r$. Letting
$i_{\aa} : C_r \to C_{r-1}$ denote evaluation of the 
first region on $\aa \in K$, one for instance has on $C_r$ for $r \geq 1$:
\begin{equation}
    \zeta(\ph)_{\aa_0 \dots \aa_r} = \sum_{\bb_0 \in K^{\aa_0}_{\aa_1}} 
    \zeta(i_{\bb_0} \ph)_{\aa_1 \dots \aa_r}. 
\end{equation}
The above extends to $C_0$ by 
agreeing to let 
$\zeta$ act as identity on 
$C_{-1} := \R^{E_\Om} \supseteq \R^{E_{\bb_0}}$, 
and letting $\aa_1 = \varnothing$ i.e. $K^{\aa_0}_{\aa_1} = K^{\aa_0}$.

On the other hand, the Möbius transform is recovered from operators $\nu_{\aa} : C_r \to C_{r-1}$ 
as $\mu(\Phi)_{\aa_0 \dots \aa_r} = \nu_{\aa_r} \dots \nu_{\aa_0} \Phi$, where:
\begin{equation} 
\nu_{\aa_0}(\Phi)_{\aa_1 \dots \aa_r} = 
\sum_{\bb_0 \in K^{\aa_0}_{\aa_1}} \mu_{\aa_0 \to \bb_0} \, \Phi_{\bb_0 \cap(\aa_1 \dots \aa_r)}
\end{equation}
The technical part of the proof then consists 
in proving $\nu_{\aa_0} \circ \zeta = \zeta \circ i_{\aa_0}$, using classical 
Möbius inversion formulas \cite[lemma 3.12]{phd}. 
One concludes that $\mu(\zeta \ph)_{\aa_0 \dots \aa_r}$ 
is computed as 
$\nu_{\aa_r} \dots \nu_{\aa_0}(\zeta \ph) = 
\zeta(i_{\aa_r} \dots i_{\aa_0} \ph) = \ph_{\aa_0 \dots \aa_r}$. 
\end{proof}

\begin{Corollary} \label{cor:BK-coeffs}
    \emph{Bethe-Kikuchi coefficients} 
    are equivalently defined by the 
    inclusion-exclusion principle \eqref{eq:inclusion-exclusion} 
    and the explicit Möbius inversion formula \eqref{eq:Bethe-numbers}:
    \begin{equation} \label{eq:inclusion-exclusion}
    \sum_{\aa \supseteq \bb} c_\aa = 1 
    \quad\hspace{1.4em}{\rm \;for\;all\;} \bb \in K,
    \end{equation}
    \begin{equation} \label{eq:Bethe-numbers}
    c_\bb = \sum_{\aa \cont \bb} \mu_{\aa \to \bb} 
    \quad{\rm \;for\;all\;} \bb \in K.
    \end{equation}
\end{Corollary}

\begin{proof} 
This classical result
is just $\zeta^*(c) = 1 \Leftrightarrow c = \mu^*(1)$, 
denoting by $\zeta^*$ and $\mu^*$ the dual automorphisms 
obtained by reversing the partial order on $K$. 
See \cite{Rota-64}. 
\end{proof}

For every $H \in C_0$ given by 
$H = \zeta  h \eqvl h = \mu H$, 
a consequence of~\eqref{eq:Bethe-numbers} is that for all $x_\Om \in E_\Om$, 
the total energy is exactly computed by the 
Bethe-Kikuchi approximation:
\begin{equation} \label{eq:U_Bethe}
    \sum_{\aa \in K} h_\aa(x_{\Om|\aa}) = 
    \sum_{\bb \in K} c_\bb \: H_\bb(x_{\Om|\bb}).
\end{equation} 
Including the maximal region $\Om$ 
in $\tilde K = \{ \Om \} \cup K$ yields the 
total energy of \eqref{eq:U_Bethe} 
as $H_\Om = \zeta(h)_\Om$, while $h_\Om = 0$ 
ensures the exactness of the Bethe-Kikuchi energy 
by \eqref{eq:Bethe-numbers}. 
In contrast, writing $S = \zeta s \in \R^K$ for the Möbius inversion of local entropies,
the total entropy $\sum c_\bb S_\bb$ computed by
problems \ref{S-crit} and \ref{varF-crit} 
neglects a global mutual information summand 
$s_\Om = \sum \mu_{\Om \to \aa} S_\aa$, which does not cancel on
loopy hypergraphs\footnote{
    Generalizing a classical result on trees, 
    we proved that $s_\Om$ also vanishes on acyclic or {\it retractable}
    hypergraphs~\cite[chap. 6]{phd}, although a deeper topological 
    understanding of this unusual notion of acyclicity in degrees $n \geq 2$ 
    is called for.
}. 
This relationship between Bethe-Kikuchi approximations and a 
truncated Möbius inversion formula was first recognized by 
Morita in \cite{Morita-57}. 

Let us now describe the remarkable commutation relations 
satisfied by $\zeta$ and $\delta$. 
Reminiscent of Greene formulas, they confirm the intuition of Rota
who saw in Möbius inversion 
formulas a discrete analogy with the fundamental theorem of calculus 
\cite{Rota-64}. They also strengthen the resemblence of GBP
with Poincaré's {\em balayage} algorithm \cite{Poincare-94}, 
used to find
harmonic forms on Riemannian manifolds, by solving 
local harmonic problems on subdomains and
iteratively updating boundary conditions.

The Gauss formula~\eqref{eq:Gauss-0} below
has been particularly useful to factorize the GBP algorithm
through $\zeta \circ \delta$ in \cite{gsi19}. 
Its simple proof will help understand the more general 
Greene formula \eqref{eq:Gauss-r} below. 

\begin{Proposition}[Gauss formula]\label{prop:Gauss}
For all $\ph \in C_1$ we have:
\begin{equation} \label{eq:Gauss-0}
\sum_{\cc \subseteq \bb} \delta \ph_\cc \;
= \sum_{\cc \subseteq \bb} 
\sum_{\substack{\aa \supset \cc \\[.3em] \aa \not \subseteq \bb}}
\ph_{\aa \to \cc}
\end{equation}
When $\Om \in K$, both sides of \eqref{eq:Gauss-0} involve the 
action of $\zeta$:
\begin{equation}
\zeta(\delta \ph)_\bb = \zeta(\ph)_{\Om \to \bb} 
\end{equation}
\end{Proposition}

\begin{proof}
Note that inclusions $\R^{E_\cc} \subseteq \R^{E_\bb}$ are implicit 
in \eqref{eq:Gauss-0} and below. By definition of 
$\zeta : C_0 \to C_0$ and $\delta : C_1 \to C_0$,
\begin{equation} 
\zeta(\delta \ph)_\bb = 
\sum_{\cc \subseteq \bb} \bigg(\, 
\sum_{\aa \supset \cc} \ph_{\aa \to \cc}
- \sum_{\dd \subset \cc} \ph_{\cc \to \dd} \,
\bigg).
\end{equation}
Inbound flux terms $\ph_{\aa \to \cc}$ such that 
$\aa \subseteq \bb$ compensate 
all outbound flux terms $\ph_{\cc \to \dd}$, which always
satisfy $\cc \subseteq \bb$. 
Therefore only inbound flux terms $\ph_{\aa \to \cc}$ such that 
$\aa \not \subseteq \bb$ and $\cc \subseteq \bb$ remain.
\end{proof}

For every $K_0, \dots , K_r \subseteq K$, 
let us write $K_0 \to \dots \to K_r$ 
for the subset
$N_r K \cap (K_0 \times \dots \times K_r)$
of strictly ordered chains. 
We then define the following subsets of $N_\bullet K$ 
to strengthen the analogy with geometry. They are depicted in 
figure \ref{fig:gauss}.

\begin{Definition} \label{def:hypercones}
For all $\aa, \bb$ in $K$ let us define: 
\begin{enumerate}[label=(\alph*)]
    \item $K^\bb = \{ \cc \in K \,|\, \cc \subseteq \bb \} \subseteq C_0$
    the 
    \emph{cone below} $\bb$,
    \item $K^\aa_\bb = K^\aa \smallsetminus K^\bb$ 
    the \emph{intercone from $\aa$ to $\bb$},
    \item $dK^\bb = \bigcup_{\aa \in K} (K^\aa_\bb \to K^\bb)$ 
    the \emph{coboundary} of $K^\bb$
\end{enumerate}
Then $K^\bb, K^\aa_\bb \subseteq N_0 K$ 
and $dK^\bb \subseteq N_1 K$. 
Given $\vec \aa = \aa_0 \dots \aa_r$ in $N_r K$ 
and $\vec \bb = \bb_1 \dots \bb_r$ in $N_{r-1} K$, 
we then recursively define the following subsets 
of $N_r K$ for $r \geq 1$:
\begin{enumerate}[label=(\alph*)]
    \setcounter{enumi}{3}
    \item $K^{\aa_0 \dots \aa_r} = 
    K^{\aa_0}_{\aa_1} \to K^{\aa_1 \dots \aa_r}$ 
    the \emph{$r$-hypercone below} $\vec \aa$, 
    \item $dK^{\bb_1 \dots \bb_r} = 
    \bigcup_{\bb_0 \in K} K^{\bb_0\bb_1 \dots \bb_r}$
    the \emph{coboundary} of $K^{\bb_1 \dots \bb_r}$.
\end{enumerate}

\end{Definition}

Following Rota, we want to think of $\zeta$ as a combinatorial 
form of spatial integration. 
For every $r$-field $\ph \in C_r$ and 
integration domain $\Sigma \subseteq N_r K$,
let us introduce the following notation:
\begin{equation}
\int_{\Sigma} \ph = \sum_{\vec \aa \in \Sigma} 
j_{\Sigma \leftarrow \vec \aa}(\ph_{\vec \aa})
\end{equation}
where $j_{\Sigma \leftarrow \vec \aa}$ embeds 
$\R^{E_{\aa_r}}$ into the linear colimit 
$\bigcup_{\bb \in \Sigma} \R^{E_{\bb_r}}$ 
of the local observables functor over $\Sigma$,
here a linear subspace of the algebra $\R^{E_\Om}$.

It follows that $\int_{K^{\aa_0 \dots \aa_r}}$ defines a map 
$C_r \to \R^{E_{\aa_r}}$, 
which coincides with the evaluation of $\zeta : C_r \to C_r$ on 
$\aa_0 \dots \aa_r$,  
\begin{equation} 
\int_{K^{\aa_0 \dots \aa_r}} \ph = \zeta(\ph)_{\aa_0 \dots \aa_r}.
\end{equation}
The Gauss formula \eqref{eq:Gauss-0} may then be generalized 
to all the degrees of $C_\bullet$ by the pleasant form below, 
which is only a rewriting of \cite[thm 3.14]{phd}.

\begin{Theorem}[Greene formula] \label{thm:Greene}
For all $\ph \in C_r$ we have:
\begin{equation} \label{eq:Gauss-r}
    \int_{K^{\bb_0 \dots \bb_r}} \delta \ph\;
    = \int_{dK^{\bb_0 \dots \bb_r}} \ph
\end{equation}
\end{Theorem}

Theorem \ref{thm:Greene} is proved in appendix \ref{section:apx-Gauss}.
When $\Om$ is in $K$, remark that 
the coboundary $dK^\bb$ of $K^\bb$ coincides with
$K^\Om_\bb$ by definition \ref{def:hypercones}. 
More generally, $dK^{\aa_0 \dots \aa_r}$ then coincides with 
$K^{\Om \aa_0 \dots \aa_r}$ 
and the Greene formula \eqref{eq:Gauss-r} takes the very succinct form 
$\zeta \circ \delta = i_\Om \circ \zeta$, where 
$i_\Om$ denotes evaluation of the first region on $\Om$.

Let us finally express the conjugate 
of $\delta$ by the graded automorphism $\zeta$ 
in terms of Bethe-Kikuchi numbers. 
The following proposition will yield 
a concise and efficient expression 
for the Bethe-Kikuchi diffusion algorithm 
in section \ref{section:Diffusions}. 

\begin{Theorem} \label{prop:mu-Gauss}
Let $\check \delta = \zeta \circ \delta \circ \mu$ 
and $\Phi \in C_r$ such that $\Phi = \zeta \phi$. 
Then for all $\aa_1 \dots \aa_r$ in $N_{r-1} K$, one has: 
\begin{equation} \label{eq:mu-Gauss}
\check \delta \Phi_{\aa_1 \dots \aa_r} = 
\int_{dK^{\aa_1 \dots \aa_r}} \phi \\[.3em]
= \check \Phi_{\Om \aa_1 \dots \aa_r}
\end{equation} 
where $\check \Phi_{\Om \aa_1 \dots \aa_r}$ 
denotes a Bethe-Kikuchi approximation of the total inbound flux 
to $K^{\aa_1 \dots \aa_r}$, given by
\begin{equation} \label{eq:check-Phi} 
\check \Phi_{\Om \aa_1 \dots \aa_r} = 
\sum_{\aa_0 \not\subseteq \aa_1} c_{\aa_0} 
    \Phi_{\aa_0 (\aa_0 \cap \aa_1) \dots (\aa_0 \cap \aa_r)}.
\end{equation}
\end{Theorem}

Remark that \eqref{eq:check-Phi} 
reduces to $\check \Phi_{\Om \aa_1 \dots \aa_r} = \Phi_{\Om \aa_1 \dots \aa_r}$ 
when $\Om\in K$, as in this case $c_\Om = 1$ 
and $c_\aa = 0$ for all $\aa \neq \Om$. 
The proof of theorem \ref{prop:mu-Gauss} is 
carried in appendix \ref{section:apx-Gauss}.

%% file: fig-gauss.tex
\begin{figure*}[t]
\begin{center}

\includegraphics[width=.8\textwidth]
{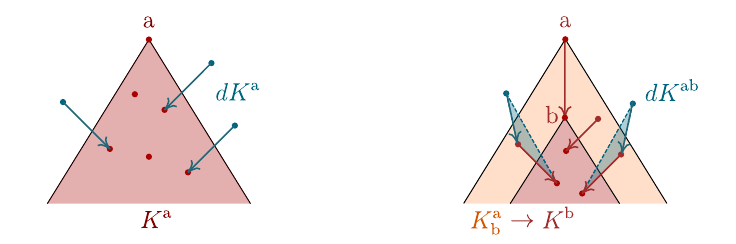}

\centering
\vspace{.2em}
    \caption{\label{fig:gauss}{\small
        Schematic pictures of the Gauss and Greene formulas. Left: 
        in degree 0, the red circles represent regions in the cone $K^\aa$ 
        being summed over when evaluating $\zeta(v)_\aa$, while blue arrows 
        represent coboundary terms of $dK^\aa$ being summed over when 
        evaluating $\zeta(\delta \ph)_\aa$, for $v \in C_0$ and $\ph \in C_1$.
        Right: in degree 1, red arrows represent flux terms in the 1-cone
        $K^{\aa\bb} = K^\aa_\bb \to K^\bb$ summed by $\zeta(\ph)_{\aa\to\bb}$, 
        while blue triangles represent the coboundary 2-chains 
        of $dK^{\aa\bb}$ summed by $\zeta(\delta\psi)_{\aa \to \bb}$, 
        for $\ph \in C_1$ and $\psi \in C_2$.
    }}
\end{center}
\end{figure*} 

%% file: IV-diffusion.tex
\section{Belief Diffusions} \label{section:Diffusions}

In this section, we describe dynamical equations 
that solve Bethe-Kikuchi variational principles. 
Their common structure will shed 
light on the correspondence theorems \ref{thm:S}, \ref{thm:varF} 
and \ref{thm:F} proven in section \ref{section:Equilibria}. 
Their informal statement is that 
solutions to free energy principles 
(problems \ref{varF-crit} and \ref{F-crit}) can both be found 
by conservative or \emph{isothermal} transport equations on $C_0$
which we described in  
\cite{gsi21}. 
In contrast, solving the max-entropy principle (problem~\ref{S-crit}) 
will require to let
temperature vary until equilibrium, 
so as to satisfy the mean energy constraint. 
This calls for a new kind of \emph{adiabatic} diffusion equations 
on $C_0$, satisfying energy conservation up to scalings only.

To make the most of the available linear structures on $(C_\bullet, \delta, \zeta)$, 
evolution is first described either at the level of potentials $v \in C_0$ 
or at the level of local hamiltonians $V = \zeta v \in C_0$. 
We conclude this section by translating the dynamic 
on beliefs $q = \rho(V) \in \Delta_0$,
for a clearer comparison with GBP. 
Bethe-Kikuchi diffusion will not only differ from GBP by an arbitrary choice of 
time step $\eps > 0$, but also from a degree-1 Möbius inversion on 
messages. 

\subsection{Isothermal diffusion}

The purpose of isothermal diffusion is to solve the local free energy principles 
\ref{varF-crit} and \ref{F-crit} by enforcing two 
different types of constraints simultaneously on $v \in C_0$: 
\begin{enumerate}[label=(\roman*)] 
\item \label{item:energy-conservation} \emph{energy conservation},
asking that $\sum_{\aa \in K} v_\aa$ is fixed to 
a given global hamiltonian $H_\Om \in \R^{E_\Om}$,
\item \label{item:consistency} \emph{belief consistency},
asking that the local Gibbs states $p = \rho^\beta(\zeta v)$ 
agree and satisfy $dp = 0$, i.e. $p \in \Gamma_0$.
\end{enumerate}
If a potential $h \in C_0$ satisfies~\ref{item:energy-conservation}, 
theorem~\ref{thm:A} implies 
that there should always exist a \emph{heat flux} $\ph \in C_1$ 
such that $v = h + \delta \ph$, in other words $v \in [h]$ 
is homologous to $h$. 
This naturally led us to generalize GBP update rules 
by continuous-time transport equations on $C_0$ \cite{gsi19,phd,gsi21}: 
\begin{equation}\label{eq:diffusion}
\frac {dv}{dt} = \delta \Phi(v).
\end{equation}

\begin{Definition}\label{def:isothermal-diffusion}
    Given a flux functional 
    $\Phi : C_0 \to C_1$ and an inverse temperature 
    $\beta \in \R$, we 
    call \emph{isothermal diffusion}
    the vector field $X^\beta_\Phi : C_0 \to C_0$ defined by:
    \begin{equation}\label{eq:X-Phi-beta}
    X^\beta_\Phi(v) = \frac 1 \beta \: \delta \Phi( \beta v)
    \end{equation}
\end{Definition}

The analytic submanifold $\Fix^\beta \subseteq C_0$ defined below 
should be stationary 
under~\eqref{eq:diffusion} for diffusion to solve 
Bethe-Kikuchi optimisation problems, as the
belief consistency constraint \ref{item:consistency} 
imposes restrictions 
on the flux functionals $\Phi : C_0 \to C_1$ suited for diffusion. 
The submanifold $\fix^\beta \subseteq \Fix^\beta$ describes
stronger constraints, by assuming Gibbs densities normalized to 
a common mass. 

\begin{Definition} \label{def:Fix}
Denote by $\e^{-V} \in C_0^*$ 
the unnormalized Gibbs densities of $V \in C_0$. 
For every inverse temperature $\beta > 0$, 
we call {\em consistent manifold} 
the subspace $\fix^\beta \subseteq C_0$ defined by
\begin{equation}\label{eq:fix-beta}
\fix^\beta= \{ v \in C_0 \,|\, 
\e^{- \beta \, \zeta v} \in \Ker(d) \}
\end{equation}
and call {\em projectively consistent manifold}
the larger space:
\begin{equation}\label{eq:Fix-beta}
\Fix^\beta= \{ v \in C_0 \,|\, 
\rho^{\beta}(\zeta v) \in \Gamma_0 \}
\end{equation}
We write $\fix = \fix^1$ and $\Fix = \Fix^1$.
\end{Definition}

Note that $\Fix^\beta$ only consists of a thickening of $\fix^\beta$ 
by the action of additive constants, which span a subcomplex $R_\bullet \subseteq C_\bullet$. 
Also remark that $\fix^\beta = \beta^{-1} \cdot \fix^1$ 
is a scaling of $\fix$ for all $\beta > 0$, and 
$\Fix^\beta = \beta^{-1} \cdot \Fix^1$ as well. It is hence sufficient to study isothermal 
diffusions at temperature 1, see appendix \ref{section:apx-fix}
for more details on the consistent manifolds. 

\input{fig-diffusion-traces.tex}

\begin{Definition} \label{def:proj-faithful}\label{def:faithful}
    We say that a flux functional $\Phi : C_0 \to C_1$ is
    \begin{enumerate}[label=(\alph*)]
    \item\label{item:Phi-consistent} {\em consistent} 
    at $\beta > 0$ if:
    \begin{equation}
    v \in \fix^\beta \Rightarrow \Phi(v) = 0
    \end{equation}
    \item\label{item:Phi-faithful} {\em faithful} 
    at $\beta > 0$ if moreover:
    \begin{equation} 
    \delta \Phi(v) = 0 \Leftrightarrow v \in \fix^\beta.
    \end{equation}
    \end{enumerate}
    We say that $\Phi$ is projectively consistent (resp. faithful) 
    if \ref{item:Phi-consistent} (resp. \ref{item:Phi-faithful}) 
    holds when $\Fix^\beta$ replaces $\fix^\beta$. 
    \end{Definition}

Let us now construct flux functionals 
$\Phi : C_0 \to C_1$ admissible for diffusion. 
Although consistency may be easily enforced via factorization 
through the following operator, 
as the next proposition shows, 
proving faithfulness remains a more subtle matter. 

\begin{Definition}(Free energy gradient) \label{def:DF}
Let $\DF : C_0 \to C_1$ denote the smooth functional
defined by:
\begin{equation}
\DF V_{\aa \to \bb}(x_\bb) = V_\bb(x_\bb) 
+ \ln \sum_{x_{\aa|\bb} = x_\bb} \e^{- V_\aa(x_\aa)}
\end{equation}
\end{Definition}

\begin{Proposition} \label{prop:d-D}
For all $V \in C_0$, let $q = \e^{-V} \in C_0^*$. Then 
\begin{equation}
dq = 0 \quad\Leftrightarrow\quad \DF V = 0.
\end{equation}
It follows that the 
flux functional $\Phi_f = f \circ \DF \circ \zeta$ 
is consistent for any smooth $f : C_1 \to C_1$.  
\end{Proposition}

\begin{proof}
For all $\aa \supset \bb$ in $K$, it is clear that 
$\e^{-V_\bb} = \pi^{\aa \to \bb}_*(\e^{-V_\aa})$ 
if and only if $V_\bb = -\ln \pi^{\aa \to \bb}_*(\e^{-V_\aa})$. 
\end{proof}

In section \ref{subsection:dynamic-beliefs}, 
we recover the discrete dynamic of GBP by using the 
flux functional $\Phi_{GBP} = - {\cal D} \circ \zeta$. 
The smooth dynamic on potentials integrates
the vector field $X_{GBP} = \delta \Phi_{GBP}$ on $C_0$,
computed by the diagram:
\begin{equation}
\begin{tikzcd} \label{eq:X-GBP}
    \hspace{-3em}X_{GBP}:\; C_0 \rar{\zeta} & C_0 \dar{-\DF} \\
    \, & C_1 \ular{\delta} 
\end{tikzcd}
\end{equation}
The flux functional $\Phi_{GBP}$ 
is faithful (at $\beta = 1$). 
See proposition \ref{prop:GBP-faithfulness} in appendix for a proof,
involving duality and monotonicity arguments 
which we already gave in \cite{gsi19} and 
~\cite{phd}.
Although seemingly optimal when $K$ is a graph, we argued 
in \cite[chap. 5]{phd} that the heat flux $\Phi_{GBP}$ introduces 
redundancies on higher-order hypergraphs, which explain 
the explosion of normalization constants. 

The {\it Bethe-Kikuchi diffusion flux} $\Phi_{BK} = 
- \mu \circ \DF \circ \zeta$ adds 
a degree-1 Möbius inversion \eqref{eq:mu-r} of GBP messages.
The isothermal vector field 
$X_{BK} = \delta \Phi_{BK}$ is then computed by:
\begin{equation}\label{eq:X-BK}
\begin{tikzcd}
    \hspace{-3.2em}X_{BK}:\;\, C_0 \rar{\zeta} & C_0 \dar{-\DF} \\
    C_1 \uar{\delta} & C_1 \lar{\mu} 
\end{tikzcd}
\end{equation}
The conjugate 
codifferential $\check{\delta} = \zeta \circ \delta \circ \mu$ is 
thus occurring in \eqref{eq:X-BK}, and one may substitute 
the result of theorem \ref{prop:mu-Gauss} to arrive 
at a very concise and efficient expression of 
the conjugate vector field $\zeta \circ X_{BK} \circ \mu$, 
governing the evolution of local hamiltonians.  

From the perspective of local hamiltonians $V = \zeta v$, 
first remark that 
in both cases 
the evolution $\dot{v} = \delta \phi$
under diffusion reads 
$\dot{V} = \zeta(\delta \phi)$, 
conveniently computed by 
the Gauss formula \eqref{eq:Gauss-0} on a cone $K^\bb$:
\begin{equation} 
\frac{dV_\bb}{dt}
= \zeta(\delta \phi)_\bb 
= \int_{dK^\bb} \phi
\end{equation}
We argue that the GBP flux $\Phi = -\DF(\zeta v) = - \DF V$ belongs to the 
"extensive" side, and should not be integrated as is on the coboundary 
of $K^\bb$. In fact, if one were able to compute the global free 
energy gradient $\DF V_{\Om \to \bb} = - \Phi_{\Om \to \bb}$ 
for all $\bb \in K$, 
the effective hamiltonians $V'_\bb = V_\bb + \Phi_{\Om \to \bb}$ 
would yield the sought for Gibbs state marginals exactly. 

Using the Bethe-Kikuchi flux $\phi = \mu \Phi$ thus allows 
to sum only "intensive" flux terms entering $K^\bb$.  
Their integral over $dK^\bb$ yields the 
Bethe-Kikuchi approximation $\check \Phi_{\Om \to \bb}$ 
of the global free energy gradient 
term $\Phi_{\Om \to \bb}$ by theorem \ref{prop:mu-Gauss}: 
\begin{equation} 
\frac {dV_\bb}{dt} 
= \int_{dK^\bb} \phi 
= \check \Phi_{\Om \to \bb} 
= \sum_{\aa \not\subseteq \bb} c_\aa \Phi_{\aa \to \aa \cap \bb}.
\end{equation}
Substituting $- \DF V$ for $\Phi$ yields, 
after small combinatorial rearrangements, the explicit formula 
\cite[thm 5.33]{phd}:
\begin{equation} \label{eq:BK-diffusion}
\frac{dV_\bb(x_\bb)}{dt}
= \sum_{\aa \in K}  c_\aa \bigg[- \ln 
\sum_{\substack{\hspace{.65em} 
      x_{\aa|\aa\cap\bb} \\ = x_{\bb|\aa \cap \bb}}}
\e^{-V_\aa(x_\aa)}  \bigg]
- V_\bb(x_\bb)
\end{equation}
In \cite{gsi21}, we empirically showed that the Bethe-Kikuchi diffusion flux
improves 
convergence and relaxes the need to normalize beliefs at each step.
Figure \ref{fig:traces} also shows that discrete integrators of $X_{BK}$ quickly converge
for large time steps $\lambda$ where integrators of $X_{GBP}$ require a value of $\lambda < 1/4$
on the 2-horn $K$ (the simplest 2-complex for which GBP is not exact). 
The flux $\Phi_{BK}$ may me proven faithful at least in a neighbourhood of 
$\fix$, see proposition \ref{prop:BK-faithfulness} in appendix.

\subsection{Adiabatic diffusion}

In order to solve the localised max-entropy principle \ref{S-crit}, 
one needs to let temperature vary so as to enforce the mean energy constraint. 
Theorem~\ref{thm2} will describe solutions as
dimensionless potentials $\bar v \in C_0$, 
satisfying both the consistency constraint~\ref{item:consistency}
at $\beta = 1$ and 
the following conservation contraint 
replacing~\ref{item:energy-conservation}:
\begin{enumerate}[label=(\roman*')]
\item \label{item:proj-energy}{\em projective energy conservation}, asking that 
$\sum_{\aa \in K} \bar v_\aa$ is fixed to a given line 
$\R H_\Om \subseteq \R^{E_\Om}$. 
\end{enumerate}
Fix a reference potential $h \in C_0$ 
and let $H_\Om = \sum_{\aa \in K} h_\aa$. 
By theorem \ref{thm:A}, constraint \ref{item:proj-energy}
is equivalent to the existence of
$\beta \in \R$ 
and $\ph \in C_1$ such that $\bar v = \beta h + \delta \ph$.
We therefore equivalently rewrite~\ref{item:proj-energy} 
as $\bar v \in \R[h] = \R h + \delta C_1$. 
Adding a radial source term to~\eqref{eq:diffusion}, 
we may enforce both \ref{item:proj-energy} and 
the mean energy constraint ${\cal U} = \E[H_\Om]$ 
as follows. 

Given an energy value $\U \in \R$ and 
a flux functional $\Phi : C_0 \to C_1$, 
we call 
\emph{adiabatic diffusion} 
an ordinary differential equation on $C_0$ of the form
\begin{equation}\label{eq:adiabatic-diffusion}
\begin{split}
\frac {d\bar v}{dt} = &\;\; \delta \Phi(\bar v) + 
\big({\cal U} - \langle p, h \rangle \big) \cdot \bar v \\[.3em]
&\;\; {\rm\; where\;} p = \rho(\zeta \bar v).
\end{split}
\end{equation}
It might seem paradoxical that the {\it adiabatic} diffusion includes 
what looks like an energy source term, when isothermal diffusion does not. 
The dynamical variable $\bar v \in C_0$ here describes 
a {\it dimensionless potential} $\bar v = \beta v$ 
given $v \in C_0$, 
measuring energies divided by an unknown temperature scale.  
The source term 
$\Psi^\U(\bar v) \in \R \bar v$ below therefore describes
a variation in temperature rather than a variation in energy.

\begin{Definition}\label{def:adiabatic-diffusion}
    Given a flux functional 
    $\Phi : C_0 \to C_1$, a potential $h \in C_0$ and a mean energy $\U \in \R$, we 
    call \emph{adiabatic diffusion}
    the vector field $Y^\U_\Phi : C_0 \to C_0$ defined by:
    \begin{equation}\label{eq:Y-Phi-U}
    Y^\U_\Phi(\bar v) = \delta \Phi( \bar v) 
    + \Psi^\U(\bar v)
    \end{equation} 
    \begin{equation}\label{eq:Psi-U}
    \Psi^\U(\bar v) = 
    \big[\U - \langle \rho(\zeta \bar v), h \rangle \big] \cdot \bar v
    \end{equation}
    It follows that $Y^\U_\Phi = X^1_\Phi + \Psi^\U$
    relates isothermal and adiabatic diffusions. 
\end{Definition}

The dependency in $h \in C_0$ is implicit in 
\eqref{eq:Y-Phi-U} and \eqref{eq:Psi-U}. 
It cannot be absorbed in the definition of $\Psi^\U$ as 
$\bar v \in C_0$
is unaware of the reference energy scale, 
and the pair $(\U, h) \in \R \times C_0$ is in fact necessary to 
define the mean energy constraint. 
One may write $\Psi^{\U, h}$ and $Y^{\U, h}_\Phi$ 
when the dependency in $h$ ought to be made explicit, 
although $h$ can be assumed fixed for present purposes. 

As in the isothermal case, it is necessary to restrict the 
flux functional $\Phi$ to enforce the belief consistency constraint 
\ref{item:consistency} at equilbrium. 
The same flux functionals may be used for both isothermal and adiabatic 
diffusions. 

\begin{Proposition}
Assume $\Phi : C_0 \to C_1$ is faithful. 
Then for all $\U \in \R$, the adiabatic 
vector field $Y^\U_\Phi$ is stationary on $\bar v \in C_0$ 
if and only if $p = \rho(\zeta \bar v)$ is consistent
and $\langle p, h \rangle = \U$. 
\end{Proposition}

\begin{proof}
Note that the r.h.s. of \eqref{eq:adiabatic-diffusion}
naturally decomposes into the direct sum $\delta C_1 \oplus \R \bar v {(t)}$ 
for all $t \in \R$, 
whenever the initial potential $h = \bar v {(0)}$ is not a boundary 
of $C_\bullet$ (in that case, total and mean energies vanish). 
It follows that 
$\bar v(t)$ is always homologous to a multiple of $h$, and 
that $\sum_{\aa \in K} \bar v_\aa(t)$ is always equal to 
a multiple of
$H_\Om = \sum_{\aa \in K} h_\aa$ by theorem \ref{thm:A}, 
and that $\bar v(t)$ satisfies 
the projective energy conservation constraint \ref{item:proj-energy}.  
When $\Phi$ is projectively faithful, 
the direct sum decomposition furthermore implies that adiabatic diffusion 
is stationary on $v \in C_0$ 
if and only if local Gibbs states $p = \rho(\zeta v)$ 
satisfy both the mean energy constraint 
$\langle p, v \rangle = {\cal U}$ and 
the belief consistency constraint \ref{item:consistency}, $dp = 0$.  
\end{proof}

The inverse temperature $\beta = T^{-1}$ of the system will be 
defined at equilibrium as Lagrange multiplier 
of the mean energy constraint,
by criticality of entropy 
(see e.g. \cite{Jaynes-57, Marle-16} and section \ref{section:Equilibria} below).
As we shall see, 
the Bethe-Kikuchi entropy may have multiple critical points, 
and therefore multiple values of temperature may coexist at
a given value of internal energy.

\subsection{Dynamic on beliefs} \label{subsection:dynamic-beliefs}

Let us now describe the dynamic on beliefs $q \in \Delta_0$ induced 
by isothermal and adiabatic diffusions on $C_0$. They both regularize and 
generalize the GBP algorithm of \cite{Yedidia-2005}. 
In the following, we denote by $R_\bullet \subseteq C_\bullet$ the 
subcomplex spanned by local constants. 

Starting from a potential $h \in C_0$, 
beliefs are recovered through a Gibbs state map 
$\rho^\beta \circ \zeta$ for some $\beta \in \R_+$.  
The affine subspace $[h] = h + \delta C_1$ satisfying 
\ref{item:energy-conservation} is therefore mapped 
to a smooth
submanifold of $\Delta_0$,
\begin{equation} \label{eq:Beliefs-h}
\Beliefs_{\beta h} := \rho^\beta \big(\zeta [h] \big)  \subset \Delta_0.
\end{equation}
On the other hand, the set of consistent beliefs $\Gamma_0 \subseteq \Delta_0$ 
satisfying~\ref{item:consistency} 
is a convex polytope of $C_0^*$ (whose 
preimage under $\rho^\beta \circ \zeta$ is $\Fix^\beta$).  
Conservation~\ref{item:energy-conservation} and
consistency~\ref{item:consistency} constraints
thus describe the intersection 
$\Beliefs_{\beta h} \cap \Gamma_0$ when viewed in $\Delta_0$. 
The non-linearity has shifted from \ref{item:consistency} 
to \ref{item:energy-conservation} in this picture. 

In order to recover well-defined dynamical systems on $\Delta_0$, 
trajectories of potentials $v \in C_0$ under diffusion 
should be assumed to consistently project on the quotient 
$C_0 / R_0$, i.e. to not depend on the addition of local energy constants 
$\lambda \in R_0 \simeq \R^K$.  
Classes of potentials $v + R_0$ are indeed in one-to-one correspondence 
with classes of local hamiltonians $\zeta(v + R_0) = V + R_0$, 
themselves in one-to-one correspondence with 
the local Gibbs states $q = \rho^\beta(V) = \frac 1 Z \e^{-\beta V}$.
Both GBP and Bethe-Kikuchi diffusions naturally define a dynamic on $C_0 / R_0 \simeq \Delta_0$. 

\input{fig-benchmark.tex}

Let us first relate the
GBP equations \eqref{eq:GBP-beliefs} and \eqref{eq:GBP-messages} 
to a discrete integrator of the GBP diffusion flux $X_{GBP}$, 
defined by \eqref{eq:X-GBP}.
The isothermal vector field $X_{GBP} = \delta 
\Phi_{GBP}$ can be approximately 
integrated by the simple Euler scheme
$v^{(t+ \lambda)} = (1 + \lambda X_{GBP}) \cdot v^{(t)}$, which yields
for $q = \rho(\zeta v)$
\begin{equation}\label{eq:GBP2-beliefs}
    q^{(t + \lambda)}_\bb(x_\bb) \propto q_\bb^{(t)}(x_\bb) \cdot
    \:\prod_{\substack{
            \aa \supseteq \cc\\[.2em]
            \aa \not\subseteq \bb}}
            \big[ m^{(t)}_{\aa\to\cc}(x_{\bb | \cc}) \big]^\lambda,
    \end{equation}
    \begin{equation}\label{eq:GBP2-messages}
    m^{(t)}_{\aa \to \cc}(x_\cc)
    = \frac
    {\displaystyle \sum_{x_{\aa|\cc} = x_\cc} q_\aa^{(t)}(x_\aa)}
    {q^{(t)}_\cc(x_\cc)}
    \end{equation}
Exponential versions of the Gauss formula 
\eqref{eq:Gauss-0} should be recognized in 
\eqref{eq:GBP-beliefs} as in \eqref{eq:GBP2-beliefs}.
Remark that the message term
$m^{(t)} = e^{-\Phi_{GBP}(v)} = \e^{-\DF V}$ 
of \eqref{eq:GBP2-messages}
is equal to the geometric increment of messages 
$M^{(t+1)} / M^{(t)}$ computed by the GBP update rule \eqref{eq:GBP-messages}, 
so that the two algorithms coincide for $\lambda = 1$. 
The parameter $\lambda$ otherwise appears as exponent of messages 
in \eqref{eq:GBP2-beliefs}, 
leaving \eqref{eq:GBP2-messages} unchanged.  

Because one is usually only interested in the long-term behaviour of 
diffusion and its convergence to stationary states, the parameter
$\lambda$ usually does not have to be infinitesimal. 
Figure \ref{fig:benchmark} shows the effect of varying $\lambda$ 
on a 50x50 Ising lattice. 
Choosing $\lambda \simeq 1/2$ instead of $1$ 
already greatly improves convergence \cite{gsi21}, 
making belief diffusions suitable for a wider class of initial conditions
than plain GBP. For low temperatures (high $\beta$), numerical instabilities 
occur that are not escaped by values of $\lambda \simeq 0.1$, as $q$ reaches 
distances of order $10^{-5}$ from the boundary of $\Delta_0$.  
Perhaps higher-order integration schemes would be beneficial in this regime.

On hypergraphs, the Bethe-Kikuchi diffusion flux $\Phi_{BK} = \mu \circ \Phi_{GBP}$
also behaves better than $\Phi_{GBP}$ (the evolution of beliefs 
is otherwise identical when $K$ is a graph). 
Integrating \eqref{eq:BK-diffusion} with a time step $\lambda \geq 0$ 
yields the following dynamic on $\Delta_0$:
\begin{equation}\label{eq:BK-beliefs}
    q_\bb^{(t+\lambda)} \propto 
    q_\bb^{(t)} \cdot \prod_{\aa \in K} 
    \Big[ m^{(t)}_{\aa \to \aa \cap \bb} \Big]^{\lambda c_\aa}
\end{equation}
\begin{equation}\label{eq:BK-messages}
m^{(t)}_{\aa \to \cc} = \frac 
{\pi_*^{\aa \to \cc}(q_\aa^{(t)})} {q_\cc^{(t)}}
\end{equation}
Note that Möbius inversion of messages
allows for convergent unnormalised algorithms.
Replacing the proportionality sign $\propto$ of \eqref{eq:BK-beliefs} 
by an equality assignment, the algorithm will converge to 
unormalized densities $q_\aa$ whose masses 
$\pi^{\aa \to \varnothing}_*(q_\aa)$ are all equal to $q_\varnothing$,  
the empty region taking care of harmonizing normalization constants. 

Adiabatic diffusions, on the other hand, 
only enforce the projective energy constraint \ref{item:proj-energy}, 
so as to fix an internal energy value $\U \in \R$. 
Energy scales acting as exponents in the multiplicative Lie group $\Delta_0$,
and the simple Euler scheme of \eqref{eq:adiabatic-diffusion} 
with the Bethe-Kikuchi diffusion flux yields:
$\Phi_{BK}$ yields 
\begin{equation} 
\begin{split}
    q_\bb^{(t+\lambda)} = &{\frac 1 {Z_\bb}}\,
    \Big[ q_\bb^{(t)} \Big]^{1 + \lambda \psi} \cdot \prod_{\aa \in K} 
    \Big[ m^{(t)}_{\aa \to \aa \cap \bb} \Big]^{\lambda c_\aa} \\
    &{\rm where\;}
    \psi = \U - \langle q^{(t)}, h \rangle.
\end{split}
\end{equation}
The messages remain given by \eqref{eq:BK-messages}, 
they still compute the exponential of the free energy gradient 
$\e^{-\DF \bar V}$ for $\bar V = \zeta \bar v$.

From a technical perspective, 
working with beliefs (bounded between 0 and 1) may help avoiding the numerical 
instabilities one may encounter with {\tt logsumexp} functions. However,
optimizing the product of messages computed by the exponential Gauss formula 
\eqref{eq:GBP2-beliefs} is not straightforward. 
Working with potentials and hamiltonians instead allows to use sparse matrix-vector
products to effectively parallelize diffusion on GPU, 
as implemented in the \href{https://github.com/opeltre/topos}{topos} library.

%% file: fig-diffusion-traces.tex
\begin{figure*}[t]
\begin{center}

\includegraphics[width=\textwidth]
{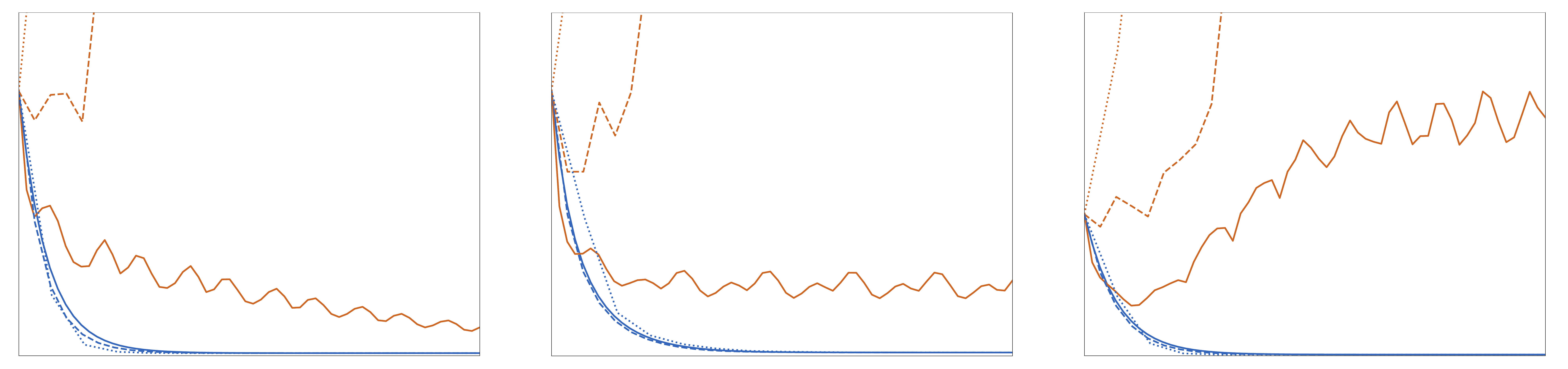}

\centering
\vspace{.2em}
    \caption{\label{fig:traces}
        {\small
        Norm of the free energy gradient $|| \DF V ||$ over 15 time units 
        for Bethe-Kikuchi diffusions (blue) and GBP diffusions (red), 
        for different values of the time step parameter $\lambda$. 
        For $\lambda = 1/4$ (solid line), 
        GBP diffusion sometimes converges in spite of oscillations; 
        for $\lambda = 1/2$ (dashed line), 
        GBP diffusion almost never converges; and for $\lambda = 1$ 
        (dotted line) GBP diffusion explodes even faster. 
        On the other hand, $\lambda$ has very little effect on the fast convergence of 
        Bethe-Kikuchi diffusion. 
        The hypergraph $K$ is the 2-horn -- join of 
        three 2-simplices $\{ (012), (013), (023) \}$ -- and coefficients 
        of the initial potential $h \in C_0$ are sampled from normal gaussians. 
    }}
\end{center}
\end{figure*} 

%% file: fig-benchmark.tex
\begin{figure}[H]
\begin{center}

\includegraphics[width=.45\textwidth]
{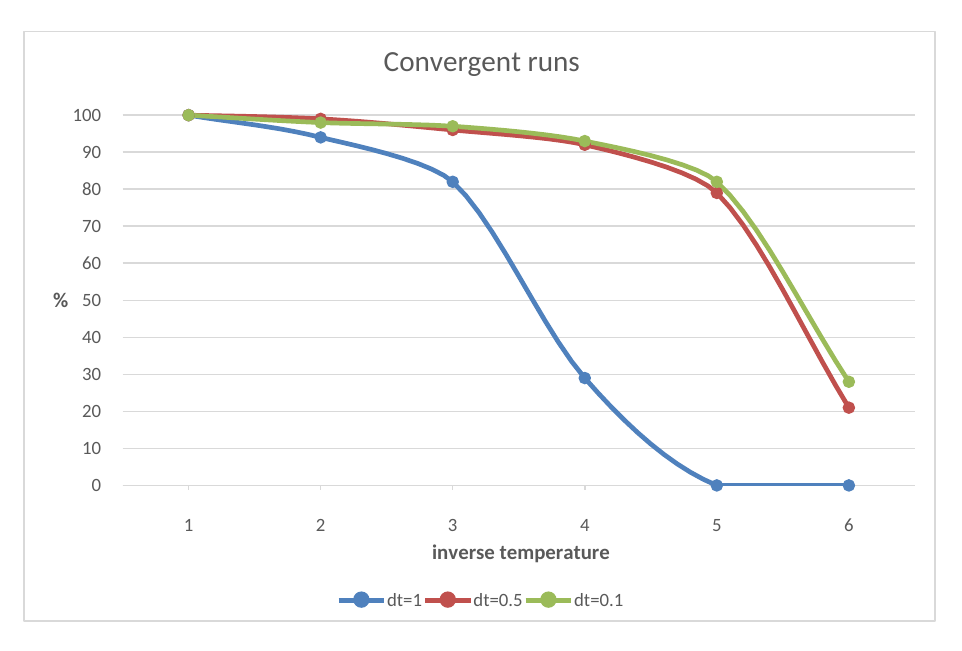}

\centering
\vspace{.2em}
    \caption{\label{fig:benchmark}{\small
        Effect of the time step parameter $\lambda$ on 
        convergence of diffusions on a 50x50 Ising lattice. 
        For $\lambda = 1$ (blue), we recover the GBP algorithm. 
        Results are very similar for $\lambda = 1 / 2$ (red) and $\lambda = 1 / 10$ 
        (green). Coefficients of the 100 initial batched hamiltonians $H \in C_0$
        are sampled from ${\cal N}(0, 1)$.  
        Note that $1 - {\rm tanh}(\beta)$
        estimates the typical distance to the boundary of $\Delta_0$ and  
        $1 - {\rm tanh}(6) \simeq 10^{-5}$ brings numerical instabilities. 
    }}
\end{center}
\end{figure} 

%% file: V-i-theorems.tex
\section{Message-Passing Equilibria} \label{section:Equilibria}

Let us now show that problems \ref{S-crit}, \ref{varF-crit} and \ref{F-crit}
are all equivalent to solving both \emph{local consistency constraints} 
and \emph{energy conservation constraints}. 
In other words, given a reference potential $h \in C_0$, all three problems
are equivalent to finding $v \in [h] \cap \Fix^\beta$.  

It will be more informative to study $[h] \cap \fix^\beta$, 
resolving the degeneracy of additive constants acting on $\Fix^\beta$.   
It may happen that this intersection is everywhere transverse, as is the 
case on acyclic or {\it retractable} hypergraphs, where true marginals 
may be computed in a finite number of steps \cite[thm. 6.26]{phd}.
In general, theorem \ref{thm:Sigma} shows that the singular subspace $\bar{\mathcal{S}_1} \subseteq \fix$
where both constraint surfaces are tangent can be described by polynomial equations 
on the convex polytope $\Gamma_0 \subseteq \Delta_0$. 
This important step towards a more systematic study of GBP equilibria 
suggests the existence of phase transitions, between both 
different regimes of diffusion 
and different free energy landscapes in Bethe-Kikuchi approximations. 

\input{fig-singularity.tex}

\subsection{Correspondence theorems}

The difference between problems 
\ref{S-crit}, \ref{varF-crit} and \ref{F-crit}  
mostly consists in which sets of variables 
are viewed as constraints and which are viewed as degrees of freedom. 
The max-entropy principle \ref{S-crit} treats $\beta$ as a free variable 
and ${\cal U}$ as a constraint,
while free energy principles \ref{varF-crit} and \ref{F-crit} 
treat $\beta$ as a constraint and 
${\cal U}$ as a free variable. 
Solutions to problems \ref{S-crit}, \ref{varF-crit} and \ref{F-crit} 
still share a common geometric description, involving intersections 
of the form $[h] \cap \Fix^\beta$ for a given potential $h \in C_0$. 

Another distinction emerges from the local structure 
of Bethe-Kikuchi approximations. It will be reflected in the 
duality between beliefs $p \in \Delta_0$ and potentials $v \in C_0$, 
and between the differential $d : C_0^* \to C_1^*$ 
and the codifferential $\delta : C_1 \to C_0$. Bethe-Kikuchi 
principles thus become reminiscent of harmonic problems, 
yet the non-linear mapping $\rho \circ \zeta : C_0 \to \Delta_0$ 
will allow for singular intersections of the 
two constraint surfaces, and multiple numbers of stationary states 
(see figure \ref{fig:singularity}).  

In the statements of theorems \ref{thm:S}, \ref{thm:varF} and \ref{thm:F} 
below, we assume a {\it projectively faithful} flux functional $\Phi : C_0 \to C_1$
(see definition \ref{def:faithful}) is chosen. 
The faithful flux $\Phi_{GBP}$ will fix $\fix^\beta$ 
instead of $\Fix^\beta$, but enforcing belief normalization 
at each step would turn it into a projectively faithful flux. 
We spare the reader with these technical details, 
covered in more details in \cite[chap. 5]{phd}.

Proofs are found in appendix \ref{apx:correspondence}. 
Recall that notations $[h] = h + \delta C_1$ 
and $[\R h] = \R h + \delta C_1$ stand for (lines of)
homology classes.

\begin{MyTheorem} \label{thm2}\label{thm:S}
    Let $\U \in \R$ and $h \in C_0$.
    Denoting by $A^{\cal U}\subseteq C_0^*$ 
    the affine hyperplane where
    $\langle p, h \rangle = {\cal U}$, 
    problem~\ref{S-crit} is equivalent to:
    \begin{equation}
    \begin{split}
    &\frac{\partial \check S(p)}{\partial p}
    \Big|_{\Tg_p (\Gamma_0 \cap\, A^{\cal U})} = 0 \\[.3em]
    \quad\eqvl\quad &
    p = \rho^1(\zeta \bar v) \in A^\U
    {\rm\;with\;} \
    \bar v \in [\R h] \cap \Fix^1
    \end{split} 
    \end{equation}
    In particular, critical potentials 
    $\bar v \in C_0$ coincide with fixed points
    of the adiabatic diffusion $Y_{\Phi}^{\cal U}$ restricted to ${[\R h]}$.
\end{MyTheorem}  

Note that any $v \in [h] \cap \Fix^\beta$ yields 
a reduced potential $\bar v = \beta v$ lying in $[\beta h] \cap \Fix^1$
by lemma \ref{apx:fix-beta}, and therefore a solution 
to problem \ref{S-crit} of energy 
$\U = \langle \rho(\zeta \bar v), h \rangle$. 

The form of theorem~\ref{thm2} is preferred because
it involves a simpler intersection problem in $C_0$:  
between the linear subspace $[\R h] = \R h + \delta C_1$,
the consistent manifold $\Fix^1$ that does not depend on $\beta$, 
and the non-linear mean energy constraint 
$\langle \rho^1(\zeta \bar v), h\rangle = \U$.
The possible multiplicity of solutions for a given energy 
constraint $\U$ may occur for different values of the Lagrange 
multiplier $\beta$.

\begin{MyTheorem} \label{thm1}\label{thm:varF}
    Let $\beta > 0, h \in C_0$.
    Problem \ref{varF-crit} is equivalent to:
    \begin{equation}
    \begin{split}
    &\frac{\partial \check\varF^\beta(p, H)}
          {\partial p}\Big|_{\Tg_p \Gamma_0} = 0 \\[.3em]
        \quad\eqvl\quad &
        p = \rho^\beta(\zeta v) 
        {\rm\;with\;} v \in [h] \cap \Fix^\beta
    \end{split}
    \end{equation}
    In particular, critical potentials 
    $v \in C_0$ coincide with fixed points
    of the isothermal diffusion $X_{\Phi}^{\beta}$ restricted to $[h]$.
\end{MyTheorem}

Theorem \ref{thm:varF} rigorously states the  
correspondence of Yedidia, Freeman and Weiss \cite{Yedidia-2005} 
between stationary states of GBP 
and critical points of the CVM, which generalized the well-known correspondence 
on graphs \cite{Ikeda-04}. 
The statement above makes the notion of stationary state more precise 
and our rigorous proof avoids any division by Bethe-Kikuchi coefficients 
thanks to lemma \ref{lemma1}.  

Before formulating the statement of theorem \ref{thm3}, 
let us denote by $c : C_0 \to C_0$ the multiplication by Bethe-Kikuchi coefficients.
One may show that $c_\bb = 0$ whenever $\bb$ 
is not an intersection of maximal regions $\aa_1, \dots, \aa_n \in K$. 
However, assuming that $K$ is the $\cap$-closure of 
a set of maximal regions does not always imply 
the invertibility of $c$. 
When $c$ is not invertible,
problem \ref{F-crit} will exhibit an affine degeneracy 
along $\Ker(c)$. 

\begin{Definition}
Let us call 
$\Fixt^\beta = \{v + b \:|\: 
v \in \Fix^\beta, b \in \Ker(c \zeta)\} \incl C_0$ 
the {\em weakly consistent manifold}. 
In particular, $\Fixt^\beta = \Fix^\beta$ when $c$ is invertible. 
\end{Definition}

A linear retraction $r^\beta : \Fixt^\beta \to \Fix^\beta$ will be defined
by equation \eqref{W} in the proof of theorem \ref{thm:F} below. 
It maps solutions of 
problem \ref{F-crit} 
onto those of \ref{varF-crit}. 
From the perspective of beliefs, 
this retraction simply consists of filling the blanks 
with the partial integration functor.

\begin{MyTheorem} \label{thm3} \label{thm:F}
    Let $\beta > 0, h \in C_0$, 
    problem \ref{F-crit} is equivalent to:  
    \begin{equation}
    \begin{split}
    &\frac{\partial \check F^\beta(V)}{\partial V}
    \Big|_{\Tg_V\,\zeta [h]} = 0 \\[.3em]
    \quad\eqvl\quad &
    V = \zeta w {\rm\;with\;} w \in [h] \cap \Fix^\beta_+ 
    \end{split}
    \end{equation}
    The weakly consistent potentials 
    $w \in [h] \cap \Fixt^\beta$
    can be univocally mapped onto $[h] \cap \Fix^\beta$
    by a retraction $r^\beta$.
    They coincide with fixed points
    of the isothermal diffusion $X_{\Phi}^{\beta}$ 
    restricted to ${[h]}$ 
    when $c$ is invertible, and with its preimage 
    under $r^\beta$ otherwise. 
\end{MyTheorem}

To prove theorems 
\ref{thm:S}, \ref{thm:varF} and \ref{thm:F}
we will need the following lemma, which is rather subtle 
in spite of its apparent simplicity (see \cite[Section 4.3.3]{phd} for
detailed formulas). 
The lemma states that multiplication by $c$ 
is equivalent to Möbius inversion up to a boundary term 
in $C_0$.  

\begin{Lemma} \label{lemma1}
    There exists a linear flux map $\Psi : C_0 \to C_1$ such that
    $c - \mu = \delta \Psi$.
\end{Lemma}

\begin{proof}[Proof of lemma \ref{lemma1}]
\nocite{Kellerer-64,Matus-88} 
Theorem \ref{thm:A}
faithfully characterizes
the homology classes $[h]$ of $C_0 / \delta C_1$ by their global 
hamiltonian  
$H_\Om = \sum_\aa h_\aa$ when $K$ is $\cap$-closed (see \cite[cor. 2.14]{phd}).
Therefore exactness of Bethe-Kikuchi energy  
\eqref{eq:U_Bethe},
implies that $h = \mu H$ and $c H$ 
are always homologous, so that 
the image of $c - \mu$ is contained in $\delta C_1$.
One may therefore construct an arbitrary $\Psi : C_0 \to C_1$ 
such that $c - \mu = \delta \Psi$ 
by linearity. The flux values taken by $\Psi$ are only constrained 
in $C_1 / \delta C_2$, as $\Ker(\delta)$
and $\Img(\delta)$ coincide on positive degrees
i.e. $C_\bullet$ is acyclic \cite[thm. 2.17]{phd}.
\end{proof}

The proofs of theorems \ref{thm:S} and \ref{thm:varF} 
(detailed in appendix)
may then be summarized as follows. Under consistency constraints 
on a critical belief $p \in \Gamma_0 \subseteq C_0^*$, 
the adjunction $d = \delta^*$ 
first implies that the 
variations cancel on $\Ker(d) = \Img(\delta)^\perp$. As linear forms on $C_0^*$,
$\partial_{p} \check S$ and $\partial_p \check \varF$
can therefore lie in $\delta C_1$ through Lagrange multipliers $\delta \psi$,
and we write $\delta \psi + \lambda \in  \delta C_1 + \Lambda$ 
for the general expression of differentials at a critical point. 
The space $\Lambda$ only
depends on the other constraints at hand : 
$\Lambda = R_0$ for problem \ref{S-crit}, and
$\Lambda = \R h + R_0$ for problem \ref{varF-crit} 
(additive constants in $R_0$ are dual to normalization constraints, 
and $h \in C_0$ is dual to the internal energy constraint $\langle p, h \rangle = \U$).

The form of Bethe-Kikuchi functionals naturally leads 
to express $\partial_p \check S$ and $\partial_p \check \varF$ 
as $c V$ for some $V \in C_0$, 
through standard computations of partial derivatives. 
Writing $V = \zeta v$, 
the remaining difficulty consists in showing that $c V = \delta \psi + \lambda$ 
is equivalent to $v = \delta \ph + \lambda$ for some $\ph \in C_1$. 
This difficult step is greatly eased by lemma \ref{lemma1},
which states that $[v] = [\mu V] = [c V]$ coincide 
as homology classes. 

In contrast, theorem \ref{thm:F} works under energy constraints 
on the potential $w \in [h] \subseteq C_0$, 
yielding local hamiltonians $V = \zeta w$, but  
the consistency of $p = \rho(V)$ is not enforced as a constraint. 
The linear form $\partial_V \check F$ will be written $c p$,
which lies in $\Ker(d \zeta^*) = \Img(\zeta\delta)^\perp$ when critical, 
because of the energy conservation constraint on $V \in \zeta(h + \delta C_1)$.  
This only implies $q = \zeta^*(c p) \in \Ker(d)$ in general, 
yet we will conclude that $p$ and 
$q$ must agree on all the 
regions $\bb$ where $c_\bb \neq 0$, so that 
the affine degeneracy of solutions (absent in \ref{thm:S} and \ref{thm:varF}) 
is completely supported by the non-maximal regions that cancel $c$.
The consistent beliefs $q \in \Gamma_0$ then solve \ref{varF-crit}.

%% file: fig-singularity.tex
\begin{figure}[H]
\begin{center}

\includegraphics[width=.35\textwidth]
{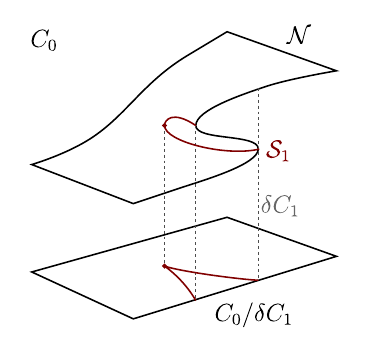}

    \caption{\label{fig:singularity}{\small
    Cuspidal singularity of the Bethe-Kikuchi variational free energy (red dot). 
    The stationary manifold $\fix$ is tangent to the space of gauge transformations 
    $\delta C_1$ on the singular subspace $\mathcal{S}_1$ (red line).
    At the cusp, $\mathcal{S}_1$ is tangent to both $\fix$ and $\delta C_1$. 
    }}
\end{center}
\end{figure} 

%% file: V-ii-singularities.tex
\subsection{Singularities}

Let us say that $v \in \fix$ is 
{\it singular} if $T_{v} \fix \cap \delta C_1 \neq 0$, and 
call {\it singular degree} of $v$ 
the number 
\begin{equation}
\cork_v = \dim( T_v \fix \cap \delta C_1).
\end{equation}
When $p = \rho(\zeta v)$, according to \eqref{eq:Beliefs-h} $\cork_v$ coincides 
with 
\begin{equation}
\cork_p = \dim(\Ker(d) \cap T_p \Beliefs_{v}).
\end{equation}
Both numbers measure the singularity of the canonical 
projection $\fix \to C_0 / \delta C_1$ onto homology classes, 
a submersion if and only if $\cork_v = 0$ everywhere on $\fix$. 

\begin{Definition}[Singular sets]
For all $k \in \N$, let 
\begin{enumerate}[label=(\alph*)]
\item ${\cal S}_k := \{ \cork_v = k \} \subseteq \fix$, 
\item $\Sigma_k := \{ \cork_p = k \} \subseteq \Gamma_0$,
\end{enumerate}
denote the \emph{singular stratifications} of $\fix$ and $\Gamma_0$ respectively.
\end{Definition}

Figure \ref{fig:singularity} depicts a situation where 
${\cal S}_1$ is non-empty. We refer the reader to Thom \cite{Thom-56} and 
Boardman \cite{Boardman-67} for more details on singularity theory.  

We show that the singular sets
$\Sigma^k$ are defined by polynomial 
equations in $\Gamma_0$. 
Singularities will therefore be located on the completion 
$\bar \Sigma^1$ of a smooth hypersurface 
$\Sigma^1 \subseteq \Gamma_0$, which may be possibly 
empty.  
In particular the intersections $\fix \cap [v]$ and
$\Gamma_0 \cap \Beliefs_v$ are
almost everywhere transverse. 

In what follows, we denote by $\R[\Delta_0]$ 
(resp. $\R(\Delta_0)$)
the algebra of polynomials (resp. rational functions) 
on the span of $\Delta_0 \subseteq C_0^*$, 
by $\R[\lambda]$ the algebra of polynomials in the 
real variable $\lambda$. 
For a vector space ${\cal E}$, the 
Lie algebra of its endomorphisms 
is denoted $\gl({\cal E})$.
 
\begin{Theorem}\label{thm:Sigma}
There exists a polynomial 
$\chi \in \R[\Delta_0] \otimes \R[\lambda]$, 
of degree $\dim(\delta C_1)$ in $\lambda \in \R$, such that for all $p \in \Gamma_0$:
\begin{equation} \label{eq:Sigma-chi}
p \in \Sigma^k \quad\Leftrightarrow \quad 
\chi_p(\lambda) {\rm\;has\;root\;} 1 {\rm\;of\;multiplicity\;} k 
\end{equation}
\end{Theorem}

\input{fig-contours_UV.tex}
 
We shall prove theorem \ref{thm:Sigma} by computing 
the corank of linearized diffusion restricted to $\delta C_1$. 

Assume given a faithful flux functional $\Phi : C_0 \to C_1$, satisfying 
the axioms of definition \ref{def:faithful}, for instance 
$\Phi_{GBP}$.  
The isothermal diffusion $X_\Phi = \delta \Phi$ then only fixes $\fix$, 
defined by $\Phi = 0$, 
while evolution remains parallel to $\delta C_1$. 
Linearizing $X_\Phi$ in the neighbourhood of $v \in \fix$ thus 
yields an 
endomorphism $\Tg_v X_{\Phi}$ on $C_0$, of kernel $\Tg_v \fix$, 
which stabilizes $\delta C_1$ by construction. 
The singular degree of $v$ therefore computes the corank of 
$\Tg_v X_{\Phi}$ restricted to boundaries,
\begin{equation} \label{eq:cork-v-TX}
\cork_v = \cork(\Tg_v X_{\Phi}|_{\delta C_1}).
\end{equation}
By faithfulness of $\Phi_{GBP} = - \DF \circ \zeta$, 
one may explicitly compute 
$\cork_v$ via minors of 
the sparse matrix
\begin{equation} \label{eq:Tv-XGBP}
    \Tg_v X_{GBP} = \delta \circ (\Tg_{\zeta v} \DF) \circ \zeta.
\end{equation} 
Letting $p = \rho(\zeta v)$, theorem \ref{thm:Sigma} 
will easily follow 
from lemma \ref{lemma:cork-p-v}, stating 
that $\cork_p = \cork_v$, 
and from proposition \ref{prop:TV-D}, 
which implies that $\Tg_v X_{GBP}$ has 
rational function coefficients in $p$. 

\begin{Lemma} \label{lemma:cork-p-v}
If $p = \rho(\zeta v) \in \Gamma_0$
then $\cork_p = \cork_v$.
\end{Lemma}

The proofs of lemma \ref{lemma:cork-p-v} and theorem \ref{thm:Sigma} are 
delayed to the end of this subsection. 
Taking a closer look at the linearized structure of $\Tg_v \fix$ 
before hand will yield an interesting description 
of singularities by conservation equations on 1-fields $\phi \in C_1$
(proposition \ref{prop:nabla-zeta-delta})
which we shall use to give an explicit expression for $\chi$ 
on binary graphs in the next subsection. 

\begin{Proposition} \label{prop:TV-D}
For all $V \in C_0$, the map 
$\Tg_V \DF : C_0 \to C_1$ is expressed 
in terms of $p = \rho(V)$ by
\begin{equation} \label{eq:TV-D}
[\Tg_V \DF \cdot V']_{\aa \to \bb}(x_\bb) 
= V'_\bb(x_\bb) - \E_{p_\aa}[V'_\aa \,|\,x_\bb]
\end{equation}
for all $V' \in C_0$, all $\aa \supset \bb$ in $K$ and all $x_\bb \in E_\bb$.
\end{Proposition}

\begin{proof}
    This computation may be found in \cite[prop. 4.14]{phd}. 
    It consists 
    in differentiating the conditional free energy term
    $- \ln \sum_{x_{\aa|\bb} = x_\bb} \e^{-V_\aa(x_\aa)}$ 
    with respect to $V_\aa \in \R^{E_\aa}$.
\end{proof}

Note that any $p\in \Gamma_0$ defines a 
family of local metrics $(\eta_{p_\aa})_{\aa \in K}$ such that 
$\eta_{p_\aa}(U_\aa, V_\aa) := \E_{p_\aa}[U_\aa  V_\aa]$, consistent 
in the sense that the restriction of $\eta_{p_\aa}$ to 
$\R^{E_\bb} \subseteq \R^{E_\aa}$ coincides with $\eta_{p_\bb}$ 
for all $\aa \supseteq \bb$ by consistency of $p$. 
The direct sum of the $\eta_{p_\aa}$ defines a scalar product $\eta_p$ 
on $C_\bullet$, which we denote by $\langle - , - \rangle_p$. 

Denote by 
$\E_p^{\aa \to \bb}$ 
the orthogonal projection $\R^{E_\aa} \to \R^{E_\bb}$ 
with respect to $\eta_{p_\aa}$, for all $\aa \supseteq \bb$. This 
is the conditional expectation operator on observables, adjoint 
of the embeddings $\R^{E_\bb} \subseteq \R^{E_\aa}$. 
Defining $\nabla_p : C_0 \to C_1$ 
by \eqref{eq:TV-D}:
\begin{equation} \label{eq:nabla-p}
    \nabla_p(V')_{\aa \to \bb} = V'_\bb - \E_p^{\aa \to \bb}[V'_\aa],
\end{equation}
propositions \ref{prop:d-D} and \ref{prop:TV-D} imply
that for all $p = \rho(\zeta v) \in \Gamma_0$, 
\begin{equation} \label{eq:Tfix-nabla}
    \Tg_v \fix = \Ker(\nabla_p \circ \zeta) = 
\Ker(\Tg_{v} (\DF \circ \zeta)).
\end{equation}
The restriction of $\eta_p$ to tangent fibers $\Tg_v \fix$ 
for $p = \rho(\zeta v)$ moreover makes $\fix$ 
a Riemannian manifold.

It is worth mentioning that $\nabla_p$ is the adjoint of $\delta$ 
for the metric $\eta_p$ on $C_\bullet$, and therefore extends 
to a degree 1 differential on $C_\bullet$ \cite[prop 5.8]{phd}.
This follows from adjunction of the projections $\E_p^{\aa \to \bb}$ with 
the inclusions $j_{\aa \leftarrow \bb}$,
as identifying $C_\bullet$ with its dual $C_\bullet^*$ through $\eta_p$, the operator 
$\nabla_p$ then represents $d$.  
However, note that $\delta C_1$ is
the orthogonal of $\Ker(\nabla_p) = \zeta(\Tg_v \fix)$ but not
of $\Ker(\nabla_p \circ \zeta) = \Tg_v \fix$, 
which may intersect $\delta C_1$. 

\begin{Proposition} \label{prop:nabla-zeta-delta}
For all $\phi \in C_1$ and all $p \in \Gamma_0$, one has: 
\begin{equation} \label{eq:nabla-zeta-delta}
    \begin{split}
\nabla_p\, \zeta (\delta \phi)_{\aa \to \bb} 
&= \int_{dK^\bb} \phi \,-\, \E_p^{\aa \to \bb} \bigg[
\int_{dK^\aa} \phi \bigg] \\[.3em]
&= \int_{K^{\aa\bb}} 
\phi \,-\, \E_p^{\aa \to \bb} \bigg[
    \int_{K^{\Om}_\aa \to K^\aa_\bb} \phi \bigg]
    \end{split}
\end{equation}
\end{Proposition}

\begin{proof}
Substituting the Gauss formula \eqref{eq:Gauss-0} into
\eqref{eq:nabla-p} yields the first line. 
We may then partition the coboundary $dK^\bb$ 
by source as $K^\Om_\bb \to K^\bb = (K^\Om_\aa \sqcup K^\aa_\bb) \to K^\bb$, 
and $dK^\aa$ by target as
$K^\Om_\aa \to K^\aa = K^\Om_\aa \to (K^\aa_\bb \sqcup K^\bb)$. 
Also note that $\int_{dK^\bb} \phi \in \R^{E_\bb}$ 
is fixed by $\E_p^{\aa \to \bb}$ to remove the redundant 
terms and obtain the second line 
(see definition \ref{def:hypercones} for notations : 
$K^\Om_\aa$ here denotes the complement of 
$K^\aa$, and $K^{\aa\bb} = K^\Om_\aa \to K^\bb$). 
\end{proof}

\begin{proof}[Proof of theorem \ref{thm:Sigma}]
For all $V \in C_0$, coefficients of the linear map $\Tg_V \DF = \nabla_p$ 
are rational functions of $p = \rho(V)$ in \eqref{eq:TV-D}.
The coefficients of $\E_{p}^{\aa \to \bb}$  
are indeed given according to the Bayes rule, for all 
$\aa \supseteq \bb$ in $K$, as 
\begin{equation} \label{eq:pab-rational}
    p_\aa(x_\aa|x_\bb) = 
    \frac {p_\aa(x_\aa)}
    {\displaystyle \sum_{y_{\aa|\bb} = x_\bb} p_\aa(y_\aa)} 
    \;\in\;\R(\Delta_0).
\end{equation}
As $\delta$ and $\zeta$ have integer coefficients, it follows
that $\Tg_v X_{GBP}$ given by \eqref{eq:Tv-XGBP}, 
is obtained by evaluating a rational function of linear 
maps $L \in \R(\Delta_0) \otimes \gl(C_0)$ 
at $p = \rho(\zeta v)$.

By $\Img(L) \subseteq \delta C_1$, we may define a restricted 
operator $L' \in \R(\Delta_0) \otimes \gl(\delta C_1)$. 
The characteristic polynomial map 
$\gl(\delta C_1) \to \R[\lambda]$ then gives  
an expression
\begin{equation} 
    \det{\lambda -1 + L'_p}_{\delta C_1} = 
    \frac{\chi(p, \lambda)}{Q(p)}
\end{equation}
It is clear from \eqref{eq:pab-rational} 
that the poles of $Q(p)$ lie on the boundary of 
$\Gamma_0$ as $p_\aa > 0$ for all $\aa \in K$ inside $\Gamma_0$; furthermore $Q(p)$ does not depend on $\lambda$. 

The multiplicity of the root $\lambda = 1$ in $\chi_p(\lambda) = \chi(p, \lambda)$
therefore computes the dimension of $\Ker(L'_p)$, which 
is precisely $\cork_v$ by 
\eqref{eq:cork-v-TX}. Lemma \ref{lemma:cork-p-v} 
finally implies that $\chi_p(1) = 0$ is a polynomial equation in $p$ of 
$\bar \Sigma^1 = \bigcup_{k \geq 1} \Sigma^k$, defined by $\cork_p \geq 1$.  
One may compute $\cork_p$ by evaluating derivatives $\partial^j \chi / \partial \lambda^j$ 
at $\lambda = 1$ to recover the singular stratification. 
\end{proof}

\begin{proof}[Proof of lemma \ref{lemma:cork-p-v}]
    Let us stress that both points of view 
    ($v \in \fix$ and $p \in \Gamma_0$) are meant to be identical, 
    were it not for the action of additive constants. 
    The Gibbs state map $\rho \circ \zeta$ induces a quotient diffeomorphism 
    $\Delta_0 \simeq C_0 / R_0$, sends $\fix$ to $\Gamma_0$ and 
    $v + \delta C_1$ to 
    $\Beliefs_v$ by definition. 
  
    Note that $\fix \cap R_0 = \mu(\R)$ is a supplement 
    of $\delta R_1$ (see appendix \ref{section:apx-fix}).
    The existence of a terminal element 
    (by $\cap$-closure assumption) 
    indeed implies that $\mu(1)$ sums to $\sum_\bb c_\bb = 1$ 
    (corollary \ref{cor:BK-coeffs}, see also appendix \ref{section:apx-fix}) 
    and that $\mu(\R)$ is not a coboundary of 
    $\delta C_1$ by theorem \ref{thm:Gauss}.

    This implies that $\delta R_1 \subseteq R_0$ 
    (acting on $[v]$ but trivially on $\Beliefs_v$)
    does not intersect $\fix$, 
    and that the intersection of $\delta C_1$
    with $\fix \cap R_0$ reduces to zero. 
    The intersections $\Tg_v \fix \cap \delta C_1$
    and $\Tg_p (\Gamma_0 \cap \Beliefs_v)$ 
    must therefore have same dimension.
    \end{proof}

%% file: fig-contours_UV.tex
\begin{figure*}[t]
\begin{center}

\includegraphics[width=\textwidth]
{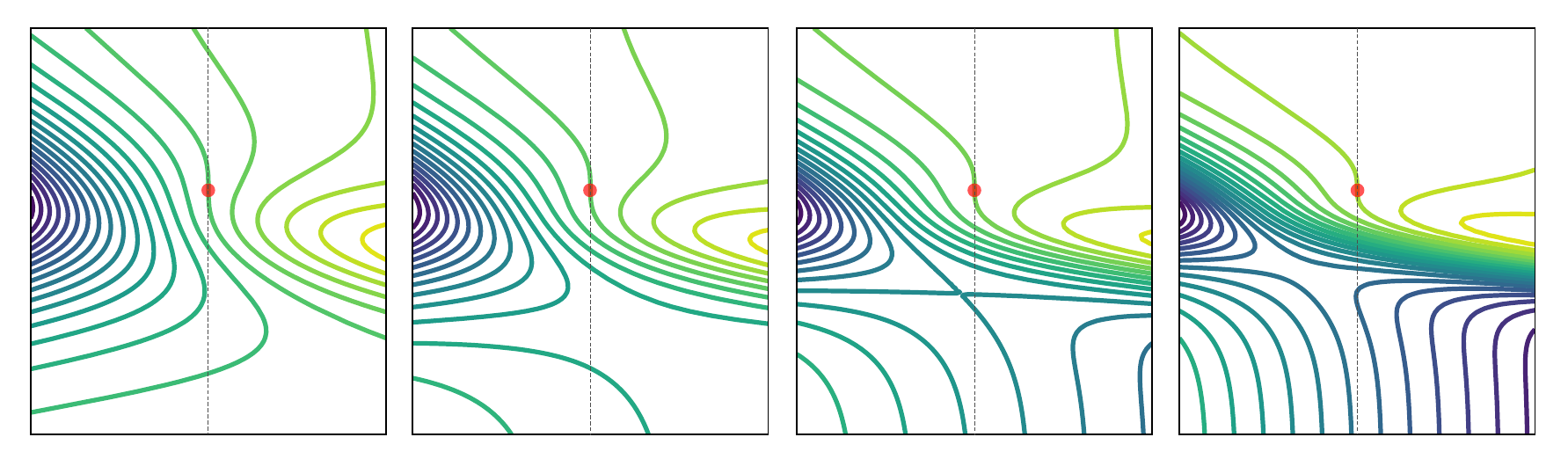}

\centering
\vspace{.2em}
    \caption{\label{fig:singular-contours}
        {\small
        Level curves of the Bethe free energy $\check F$ 
        accross a cuspidal singularity (red dot), for increasing values 
        of the external magnetic field (left to right). 
        On each plot, the horizontal axis $U \in C_0$ represents 
        variations in inverse temperature (i.e. a parameter),
        while the vertical axis 
        $V \in \zeta(\delta C_1)$ represents a 1D-fiber of 
        equivalent energies 
        (i.e. an optimization variable). 
        Convexity in the $V$-axis 
        is lost when temperature drops as two additional critical points 
        stem from the singularity (right of the dashed line).
        Also notice $\check F$ 
        acquires a sharp step when the magnetic field increases. 
    }}
\end{center}
\end{figure*} 

%% file: V-iii-graphs.tex
\subsection{Loopy Graphs}

\newcommand{\Khoff}{{\cal Z}_1^{(p)}}

In the case of graphs, we may give polynomial equations
for the singular strata $\Sigma^k$ explicitly. 
They are obtained as a loop series expansion 
by focusing on the action of diffusion on fluxes $\ph \in C_1$, 
via the remarkable Kirchhoff formula \eqref{eq:Kirchhoff} below.
The reader may find 
in \cite{Mooij-07,Watanabe-09,Mori-13,Sudderth-07}
very similar loop expansions 
for the Bethe-Kikuchi approximation error 
and the analysis of BP stability. 

When $K \subseteq \Pow(\Om)$ is a graph, we simply write $ij \in K$ for edges and $i \in K$ for
vertices (instead of $\{i, j\}$ and $\{i\}$). 
Our $\cap$-closure assumption usually 
implies that $\varnothing \in K$ and the nerve of $K$ 
is then a simplicial set of dimension 2. 
However, as $N_2 K$ only consists of chains $ij \to i \to \varnothing$ 
and $E_\varnothing$ is a point (unit for $\times$), 
$C_2$ is only spanned by additive constants 
and coincides with $R_2$. 

We denote by $K' = K \smallsetminus \{ \varnothing \}$ 
the associated graph in a more usual sense, 
whose nerve $N_\bullet K'$ is of dimension 1. 
The notation $i \frown j$ will indicate that 
$i$ is a neighbour of $j$ in $K$ and $K'$, 
whenever $ij \in K$. 

\begin{Proposition}[Kirchhoff formula] \label{prop:Kirchhoff}
Given $p \in \Gamma_0$, denote by $\Khoff \subseteq C_1$
the subspace defined by $\E_{p_\bb}[\phi_{\aa \to \bb}] = 0$ for all 
$\aa \to \bb \in N_1 K$
and
\begin{equation} \label{eq:Kirchhoff}
    \phi_{jk \to k} = 
    \E_p^{jk \to k}\bigg[ \sum_{i \frown j} \phi_{ij \to j} \bigg]
\end{equation} 
for all $jk \to k \in N_1 K$. Then $\cork_p = \dim \Khoff$. 
\end{Proposition}

\begin{proof} 
First assume that $\phi \in C_1$ is orthogonal to $R_1$ for $\eta_p$. 
Letting $(1_{\aa \to \bb})$ denote the canonical generators of $R_1$, 
we then have  
$\E_{p_\bb}[\phi_{\aa \to \bb}] = 0 = \eta_p(\phi, 1_{\aa \to \bb})$ 
for all $\aa \to \bb$ in $N_1 K$, and in particular
$\phi_{\aa \to \varnothing} = 0$ for all $\aa \in K$. 
Let $C'_1$ denote the space of such fields orthogonal to $R_1$. 

Assume now in addition that $\delta \phi \in \Ker(\nabla_p \circ \zeta)$. 
It then follows from proposition \ref{prop:nabla-zeta-delta} that 
for all $jk \to k \in N_1 K$,
\begin{equation} \label{eq:Khoff-jk}
\int_{K^{(jk)}_{(k)} \to K^{(k)}} \phi = \E_p^{jk \to k} \bigg[ 
\int_{K^\Om_{(jk)} \to K^{(jk)}_{(k)}} \phi \bigg]
\end{equation}
Equation \eqref{eq:Kirchhoff} is equivalent to cancelling the r.h.s. of 
\eqref{eq:nabla-zeta-delta} in its second form when $\phi \in C'_1$,
as all the $\phi_{\aa \to \varnothing}$ vanish. Its 
l.h.s. reduces to $\phi_{jk \to k}$ while the r.h.s. sums inbound 
fluxes $\phi_{ij \to j}$ over source
edges $ij \not \subseteq jk$, containing the target $j \subseteq jk, j \not\subseteq k$ 
(brackets in \eqref{eq:Khoff-jk} are used to avoid ambiguous interpretation of notations 
in definition \ref{def:hypercones}).

We showed that $\eqref{eq:Kirchhoff}$ 
describes $\Ker(\nabla_p \circ \zeta \circ \delta)$ under the assumption 
that $\phi \perp R_1$. 
Recalling that $\cork_p$ is the dimension of $\Ker(\nabla_p \circ \zeta) \cap \delta C_1$, 
we may compute $\cork_p$ as corank of the restriction of 
$\nabla_p \circ \zeta \circ \delta$ 
to a supplement of $\Ker(\delta)$. 
Now as $\delta = \nabla_p^*$ for the metric $\eta_p$, the subspace 
$\nabla_p\, C_0 \subseteq C_1$ is such a supplement.

Let us show that $\nabla_p C_0$ contains $C'_1$. 
Given $\phi \in C'_1$, the orthogonal 
projections of $\phi_{ij \to j} \in \R^{E_j}$ 
and $\phi_{ij \to i} \in \R^{E_i}$ 
onto $\R^{E_j} \cap \R^{E_i} = \R^{E_\varnothing} = \R$ vanish 
for every edge $ij \in K$ by the assumption $\phi \perp R_1$. 
We may thus   
choose $V_{ij}$ in $\R^{E_{ij}} \simeq \R^{E_i} \otimes\R^{E_j} \supset \R^{E_i} + \R^{E_j}$
that projects onto $-\phi_{ij\to i} \in \R^{E_i}$ and 
$-\phi_{ij \to j} \in \R^{E_j}$. 
Letting $V_i = 0$ for all vertex $i \in K$ and $V_\varnothing = 0$, 
we may get $V \in C_0$ such that 
$\nabla_p V = \phi$, as 
\begin{equation}
    \nabla_p(V)_{ij \to j} = 0 - \E_p^{ij \to j}[V_{ij}] = \phi_{ij \to j}
\end{equation}
for all $ij \to j \in N_1 K$, and $\nabla_p V_{\aa \to \varnothing} = 0$
for all $\aa \in K$. 
As $C'_1 \subseteq \nabla_p C_0$ 
consists of cocyles, 
$\cork_p$ is greater or equal than the corank 
of $\nabla_p \circ \zeta \circ \delta$ restricted to $C'_1$. 

The subspace $R'_1 = R_1 \cap \nabla_p C_0$ contains the remaining 
cocyclic degrees of freedom as 
\begin{equation}
\nabla_p C_0 = C'_1 \overset{\perp}{\oplus} R'_1.
\end{equation}
However, as we shown in the proof of lemma \ref{lemma:cork-p-v},
the subspace $\delta R'_1 \subseteq \delta R_1$ does not intersect 
$\Ker(\nabla_p \circ \zeta)$, while it is stable under diffusion. 
Therefore $R'_1$ does not contribute to $\cork_p$ and 
\begin{equation}
\cork_p = \cork(\nabla_p \circ \zeta \circ \delta_{|C'_1}).
\qedhere
\end{equation}
\end{proof}

The Kirchhoff formula \eqref{eq:Kirchhoff} will allow us 
to relate the emergence of singularities to the topology of $K$. 
Just like steady electric currents cannot flow accross open circuits, 
it is clear \eqref{eq:Kirchhoff} will not admit any non-trivial solutions 
when $K$ is a tree. However, unlike electric currents, 
the zero-mean constraint excludes scalar fluxes, fixed 
by conditional expectation operators. 
Multiple loops will thus need to collaborate 
for non-trivial solutions to appear. 

Fixing $p \in \Gamma_0$, denote by $C'_1 \subseteq C_1$
the orthogonal of local constants for $\eta_p$ as above. 
Choosing a configuration $o_i \in E_i$ and letting 
$E_i^* = E_i \smallsetminus \{ o_i \}$ for each vertex $i$, one has an isomorphism:
\begin{equation}
C'_1 \simeq \prod_{ij \to j \in N_1 K'} \R^{E_j^*}
\end{equation}
Let us also denote by ${\bf E}_p : C'_1 \to C'_1$ the 
{\it edge propagator} 
\begin{equation} \label{eq:E-propagator}
{\bf E}_p(\phi)_{jk \to k} = 
\E_p^{jk \to k}\bigg[ \sum_{i \frown j} \phi_{ij \to j} \bigg],
\end{equation}
so that \eqref{eq:Kirchhoff} 
is the eigenvalue equation $\phi = {\bf E}_p(\phi)$. 
One may recover $\cork_p$ in the characteristic polynomial of 
${\bf E}_p$, 
which we compute explicitly for binary variables. 

\begin{Definition}
Define a directed graph structure ${\cal G}$ on 
$N_1 K'$, 
by including all edges of the form $(ij \to j) \triangleright (jk \to k)$
for $i \neq k$. 
\end{Definition}

The edges of ${\cal G}$ describe all non-vanishing coefficients of the 
matrix ${\bf E}_p$. However note that coefficients of ${\bf E}_p$ are 
indexed by lifts of an edge $(ij \to j) \trito (jk \to k)$ 
to a pair of configurations $(x_j, x_k) \in E_j^* \times E_k^*$ in general. 

Let us now restrict to binary variables for simplicity, 
so that the edges of ${\cal G}$ are in bijection 
with the non-vanishing coefficients of 
${\bf E}_p \in \R(\Delta_0) \otimes \gl(C'_1)$. 
The coefficient of ${\bf E}_p$ attached to an 
edge $(ij \to j) \trito (jk \to k)$ actually does not depend on $i$, 
as it consists in projecting observables on $j$ to observables 
on $k$ with the metric induced by $p_{jk}$ by \eqref{eq:E-propagator}.
It may thus be denoted $\eta_{j k}(p)$ for now. 
An explicit form will be given by \eqref{eq:eta-jk} below, 
which is symmetric in $j$ and $k$. 

\begin{Definition}
    Denote by $\Perm^k {\cal G} \subset \mathfrak{S}(N_1 K')$ 
    the set of permutations 
    with exactly $m - k$ fixed points, compatible with ${\cal G}$.  
    Any $\gamma \in \Perm^k {\cal G}$ decomposes as a product 
    of $l(\gamma)$ disjoint cycles. 
    \end{Definition}

\begin{Theorem} \label{thm:chi-graph}
Assume $x_i \in E_i$ is a binary variable for all $i \in K$. 
Then $p \in \bar \Sigma^s$ 
if and only if $(\lambda - 1)^s$ divides the polynomial 
\begin{equation} \label{eq:chi-graph}
\chi_p(\lambda) 
= \sum_{k \geq 0}^{m}
\lambda^{m - k} \sum_{\gamma \in \Perm^k {\cal G}}
(-1)^{l(\gamma)} \,
{\bf \Lambda}_p[\gamma],
\end{equation}
where ${\bf \Lambda}_p[\gamma]$ is the product of coefficients 
of ${\bf E}_p$ accross $\gamma$
\begin{equation}
{\bf \Lambda}_p[\gamma] = \prod_{ij \to j \,\trito jk \to k}^\gamma \eta_{jk}(p),
\end{equation}
and where $\eta_{jk}(p)$ can be chosen as \eqref{eq:eta-jk} below, in an 
orthonormal system of coordinates for $\eta_p$. 
\end{Theorem}

Note that factorizing $\gamma$ as $\gamma_1 \dots \gamma_{l(\gamma)}$, one has:
\begin{equation}
{\bf \Lambda}_p[\gamma] = \prod_{s=1}^{l(\gamma)} {\bf \Lambda}_p [\gamma_s].
\end{equation}
We may call ${\bf \Lambda}_p[\gamma_s]$ the {\it loop 
eigenvalue} of $\gamma_s$. 

\input{fig-dumbell.tex} 

\begin{proof}[Proof of theorem \ref{thm:chi-graph}]
Proposition \ref{prop:Kirchhoff} implies that $\cork_p$ is 
the multiplicity of 1 as eigenvalue of ${\bf E}_p$, in other 
words the dimension of $\Ker(1 - {\bf E}_p) \subseteq C'_1$. 
Let us show that $\chi_p(\lambda)$ does 
compute the characteristic polynomial of ${\bf E}_p$. 

Consider a matrix $M : \R^m \to \R^m$, whose diagonal coefficients all vanish, 
and write $\Perm^k_m \subseteq \Perm_m$ for the set of
permutations having exactly $m - k$ fixed points. 
Using the Leibniz formula, one gets for
$\chi_M(\lambda) = \det{\lambda - M}$:
\begin{equation}
\chi_M(\lambda) = 
\sum_{k = 0}^m
\lambda^{m - k}
\sum_{\sigma \in \Perm^k_m} 
\eps(\sigma)
 (-1)^k
\prod_{e \not\in {\rm Fix}(\sigma)} M_{e, \sigma(e)}
\end{equation}
The signature of a length-$k$ cycle $\sigma \in \Perm^{k}_m$ has 
signature $\eps(\sigma) = (-1)^{k - 1}$. 
By multiplicativity of $\eps(\sigma)$, it follows that for
every product of disjoint cycles $\sigma = \sigma_1 \dots \sigma_l \in \Perm^k_m$,  
\begin{equation}
(-1)^k \eps(\sigma) = (-1)^{l} 
\end{equation}
Therefore \eqref{eq:chi-graph} 
computes the characteristic polynomial of ${\bf E}_p$, 
whose diagonal coefficients do  vanish.  
\end{proof}

Let us express the $\eta_{jk}(p)$
in an orthonormal system of coordinates, convenient for the 
symmetry in $j$ and $k$ it induces. 
The price to pay is that $\eta_{jk}(p)$ becomes a rational function 
in $\sqrt p$ and not in $p$, but this 
choice greatly simplifies computations. 
We write $x_i \in \{ \pm \}$ with $p_i^\pm = p_i(\pm)$ and 
$p_{ij}^{\pm\pm} = p_{ij}(\pm\pm)$, etc. 

For $\phi \in C'_1$, each 
flux term $\phi_{ij \to j}$ may be constrained to 
a 1-dimensional subspace $\R u_j$ 
of $\R^{E_j} \simeq \R^2$,  
chosen so that $\E_{p_j}[u_j] = 0$ and $\E_{p_j}[u_j^2] = 1$.
In the $(+, -)$ coordinates, we thus define $u_j$ as: 
\begin{equation}
u_j = \frac 1 {\sqrt{p_j^+ p_j^{-}}} \cdot 
\left(\begin{matrix}
    p_j^{-} \\  
    - p_j^+ 
\end{matrix}\right)
\end{equation}
Because $\E_p^{jk \to k}[u_j]$ has zero mean too, it must lie 
in $\R u_k$. 
Hence, $\E_p^{jk \to k}[u_j] = \eta_{jk}(p) \cdot u_k$ 
where $\eta_{jk}(p)$ is found as:
\begin{equation} \label{eq:eta-jk}
\eta_{jk}(p) = \frac{p_{jk}^{++} p_{jk}^{--} - p_{jk}^{+-} p_{jk}^{-+}}
{\sqrt{p_j^+ p_j^- \cdot p_k^+ p_k^-}}
\end{equation}
The scaling factor $\eta_{jk}(p)$
is symmetric in $j$ and $k$ as $\E_p^{jk \to k}$ 
is a (self-adjoint) projector in $\R^{E_{jk}} \simeq \R^4$. 
The reader may check \eqref{eq:eta-jk} 
by performing the 2x2 matrix-vector product $p_{kj} \cdot u_j$ 
and dividing the resulting vector by $p_k$, 
according to the Bayes rule,
then substituting for $p_j$ and $p_k$ the marginals 
of $p_{jk}$ when necessary. Again, we stress that the image of
$u_j$ must lie in $\R u_k$ by the zero-mean constraint 
$u_j, u_k \perp \R$.  

Letting $g_j = \sqrt{p_j^+ p_j^-}$ denote the geometric mean of $p_j$ 
and letting $\tilde u_j = g_j u_j$ for all $j$: 
\begin{equation}
\E_p^{jk \to k}[\tilde u_j] 
= \frac {g_j}{g_k} \eta_{jk}(p) \tilde u_k  
= \tilde{\eta}_{jk}(p) \tilde u_k
\end{equation}
leads to coefficients $\tilde \eta_{jk}(p)$ that are no longer 
symmetric, but are rational functions of ${p}$. 
Note that the products of the $\eta_{jk}(p)$ and 
$\tilde {\eta}_{jk}(p)$ accross a cycle $\gamma \in \Perm^k {\cal G}$ 
do coincide.  

%% file: fig-dumbell.tex
\begin{figure*}[t]
\begin{center}
\vspace{-2em}
\includegraphics[width=.98\textwidth]
{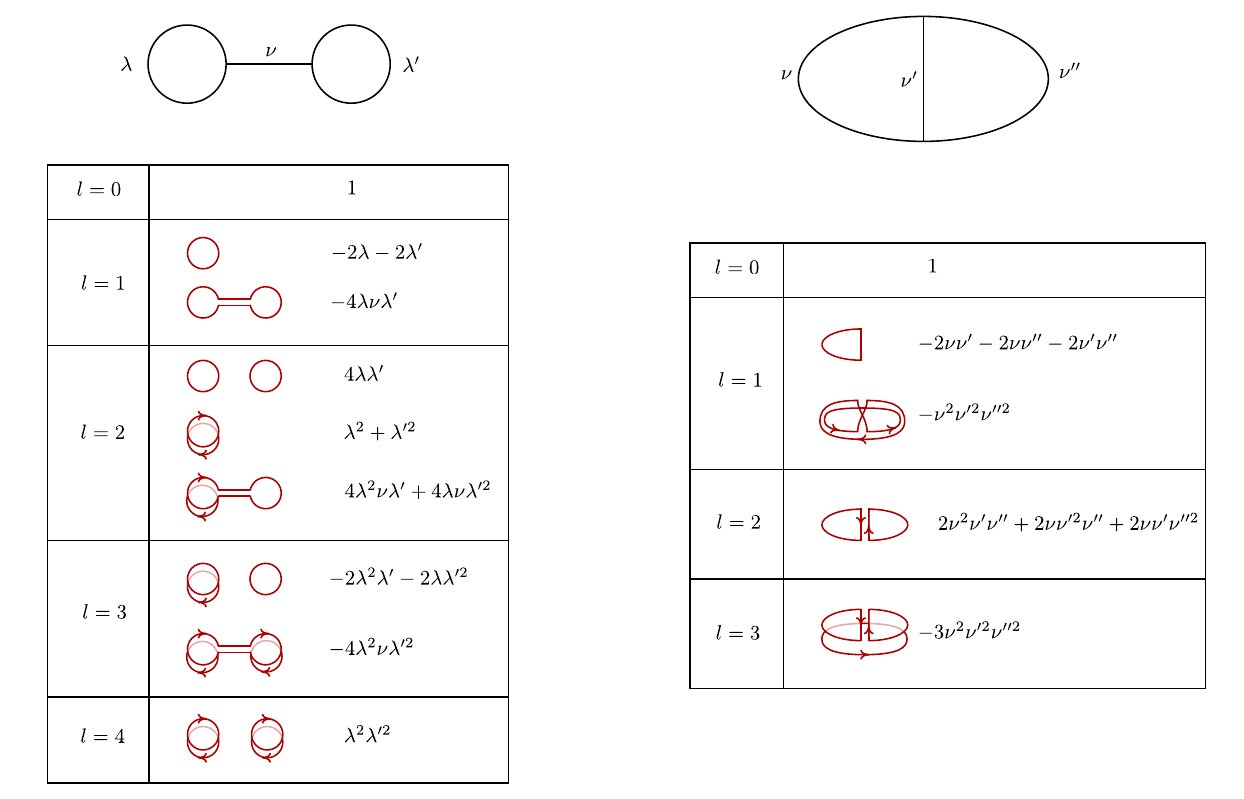}
\centering
\vspace{.2em}
    \caption{\label{fig:dumbell}{\small
        Computation of $\chi_p(1)$ on the dumbell and theta graphs. 
        The sum of listed monomials sums to $0$ 
        when $p \in \bar \Sigma^1$ is singular. Every edge $jk$ in the base graph $K'$
        gives rise to a pair $jk \to k$ and $jk \to j$ of vertices in ${\cal G}$. 
        This explains why every edge of $K'$ may be walked through in both directions, 
        although U-turns are not permitted as $jk \to j$ and $jk \to j$ are not adjacent in ${\cal G}$. 
        The precise topology of $K'$ (including its vertices) does not really matter here, 
        as loop eigenvalues ${\bf \Lambda}_p[\gamma]$ only compute 
        products of coefficients accross cycles.
    }}
\end{center}
\end{figure*} 

%% file: VI-conclusion.tex
\section{Conclusion}

Belief propagation algorithms 
and their relationship to Bethe-Kikuchi optimization problems were 
the cornerstone motivation for this work.
We produced the most comprehensive picture of this vast subject
that was in our power, although loop series expansions 
\cite{Watanabe-09, Mooij-07, Mori-13,Sudderth-07}
surely would have deserved more attention. 

The rich structure 
owned by the complex $(C_\bullet, \delta, \zeta)$ was 
a surprise to uncover, in a field usually dominated by statistics. 
We hope the length of this article will be 
seen as an effort to motivate the unfriendly 
ways of Möbius inversion formulas and homological algebra, 
and demonstrates their reach and expressivity in a localized theory 
of statistical systems. This intimate relationship 
between combinatorics and algebraic topology 
(culminating in theorems \ref{thm:Gauss} and \ref{thm:Greene}) 
has not been described to our knowledge, 
although both subjects are classical and well covered individually. 

From a practical perspective, we showed that belief propagation algorithms 
are time-step 1 Euler integrators 
of continuous-time diffusion equations on $(C_\bullet, \delta, \zeta)$.
We propose to call these ODEs {\it belief diffusions}, 
as "diffusions" suggest \ref{item:energy-conservation} a conservation equation of the form 
$\dot v = \delta \Phi(v)$, while "belief" recalls \ref{item:consistency} 
that their purpose 
is to converge on consistent pseudo-marginals, critical 
for Bethe-Kikuchi information functionals.  
The GBP algorithm offered a spatial compromise between precision 
and complexity in the choice of the hypergraph $K \subseteq \Pow(\Om)$; 
our diffusion equations offer a temporal compromise between 
runtime and stability. They thus allow for a wider range 
of initial conditions and applications.

Belief diffusions, in their isothermal and adiabatic form, 
solve localized versions of the max-entropy (\ref{S-crit}) 
and free energy principles (\ref{varF-crit} and \ref{F-crit}).
The associated Bethe-Kikuchi functionals have a piecewise constant number 
of critical points, whose discontinuities are located 
on the projection of $\bar{\cal S}_1 \subseteq \fix$ 
on the quotient space of parameters $C_0 / \delta C_1$. 
A stationary point in $\fix$ crossing $\bar{\cal S}_1$ will become 
unstable and forced onto a different sheet of the intersection 
with homology classes. This would appear as a 
discontinuous jump in the convex polytope $\Gamma_0$, 
happening anytime a consistent belief $p \in \Gamma_0$ crosses 
the singular space $\bar \Sigma_1$.

%% file: sectionA-gauss.tex
\section{Gauss-Greene Formulas} \label{section:apx-Gauss}

\begin{proof}[Proof of theorem \ref{thm:Gauss}]
Let us first show {\it (i)} $\Rightarrow$ {\it (ii)}, stating  
that $\sum_{\aa \in K} \delta \ph_\aa = 0$ vanishes for all $\ph \in C_1$. 
Recall that the canonical inclusions $\R^{E_\cc} \subset \R^{E_\bb} \subset \R^{E_\Om}$
are implicit in:
\begin{equation}\label{eq:phi-sum}
    \sum_{\bb\in K}\delta \ph_\bb 
    = \sum_{\aa \supset \bb} \ph_{\aa \to \bb} 
    - \sum_{\cc \subset \bb} \ph_{\bb \to \cc}
\end{equation} 
Both sums run over the strict 1-chains of $N_1 K$, 
therefore \eqref{eq:phi-sum} 
vanishes as a $K$-local observable in $\R^{E_\Om}$
and $\delta C_1$ is contained in the kernel 
of the total energy map $C_0 \to \R^{E_\Om}$. 

Showing {\it (ii)} $\Rightarrow$ {\it (i)} is more involved, 
as it relies on the interaction decomposition theorem 
(see \cite{Kellerer-64, Matus-88, agt23} and \cite[section 2.3.2]{phd}) 
which requires $K$ to be $\cap$-closed. This theorem states that one may 
choose a consistent family of {\it interaction subspaces} 
${\bf Z}_\aa \subseteq \R^{E_\aa}$ such that for all $\aa \in K$:
\begin{equation}
\R^{E_\aa} = \bigoplus_{\bb \subseteq \aa} {\bf Z}_\bb
\end{equation}
For every $\aa \supseteq \bb \in K$, let us denote by ${\bf Z}_{\aa \to \bb}$ 
the projection of $\R^{E_\aa}$ onto ${\bf Z}_\bb$ cancelling on 
every ${\bf Z}_\cc$ for $\cc \neq \bb$. 

Given $h \in C_0$ such that $\sum_\aa h_\aa = 0$, 
define $\ph \in C_1$ by letting 
$\ph_{\aa \to \bb} = - {\bf Z}_{\aa \to \bb}(h_\aa)$ for all $\aa \supset \bb$ 
in $K$. Then $\delta \ph$ is given by:
\begin{equation}\label{eq:phi-Z}
    \delta \ph_\bb = -\sum_{\aa \supseteq \bb} 
    {\bf Z}_{\aa \to \bb}(h_\aa) + \sum_{\cc \subseteq \bb} {\bf Z}_{\bb \to \cc}(h_\bb)
\end{equation}
The second 
sum over $\cc \subseteq \bb$ reconstructs $h_\bb \in \R^{E_\bb}$. 
To see that the first sum over $\aa \supseteq \bb$ vanishes, 
one should notice that $\R^{E_\bb} \cap \R^{E_{\bb'}} = \R^{E_{\bb \cap \bb'}}$ 
is a strict subspace of $\R^{E_\bb}$ whenever $\bb'$ does not contain $\bb$. 
This implies that its projection onto ${\bf Z}_\bb$ vanishes, and 
that the projection of $\sum_{\aa \in K} h_\aa$ onto ${\bf Z}_\bb$ 
is given by the first term of \eqref{eq:phi-Z}. 
Therefore $\delta \ph = h$ by assumption, 
and $\delta C_1$ coincides with the space of zero-sum potentials. 
\end{proof}

\begin{proof}[Proof of theorem \ref{thm:Greene}]
        Denote by ${\cal E}_j$ the set 
        of chains $\vec \aa \in N_{r+1} K$ such that $\vec \aa^{(j)}$ 
        lies in the hypercone $K^{\bb_0 \dots \bb_r}$, 
        for all $0 \leq j \leq r + 1$. 
        From definition of $\delta$, the l.h.s. of \eqref{eq:Gauss-r} then rewrites 
        \begin{equation} \label{eq:Gauss-Ej}
            \int_{K^{\bb_0 \dots \bb_r}} \delta \ph 
            = \sum_{j = 0}^{r+1} 
            \: \sum_{\vec \aa \in {\cal E}_j} (-1)^j \, \ph_{\vec \aa}.
        \end{equation}
        Let $\vec \cc = \vec \aa^{(j)} \in K^{\bb_0 \dots \bb_r}$ for $j > 0$.  
        Let us show that $\vec \aa$ also belongs to either ${\cal E}_{j-1}$
        or ${\cal E}_{j + 1}$. 
        By construction of ${\cal E}_{j}$, 
        agreeing to let $\cc_{j+1} = \bb_{j+1} = \varnothing$ and $\bb_0 = \Om$, 
        we have: 
        \begin{enumerate}[label=(\roman*)] 
            \item $\aa_{i} = \cc_i \in K^{\bb_i}_{\bb_{i+1}}$ for $i < j$,
            \item $\aa_{j-1} \supset \aa_j \supset \aa_{j+1}$,
            \item $\aa_{k+1} = \cc_k \in K^{\bb_k}_{\bb_{k+1}}$ for $k \geq j$.
        \end{enumerate}
        
        If $\aa_j \subseteq \bb_j$, then $\aa_j \supset \aa_{j+1} \not\subseteq \bb_{j+1}$ 
        implies $\aa_j \in K^{\bb_j}_{\bb_{j+1}}$. 
        Then $\vec \aa^{(j-1)} \in K^{\bb_0 \dots \bb_r}$, as 
        (iii) is satisfied for $k \geq j - 1$.
        
        If $\aa_j \not\subseteq \bb_j$, 
        then $\bb_{j-1} \supseteq \aa_{j-1} \supset \aa_j$ implies 
        $\aa_{j} \in K^{\bb_{j-1}}_{\bb_j}$. 
        Then $\vec \aa^{(j+1)} \in K^{\bb_0 \dots \bb_r}$, 
        as (i) is satisfied for $i \leq j$.
        
        We just proved 
        ${\cal E}_j \subseteq {\cal E}_{j-1} \cup {\cal E}_{j+1}$ 
        for all $0 < j \leq r + 1$, 
        and the reader may easily check that (i) and (iii) 
        do also imply that ${\cal E}_j \cap {\cal E}_k = \varnothing$ 
        for $|k - j| > 1$. 
        There only remains to determine which chains $\vec\aa \in {\cal E}_0$ 
        are not killed by the alternated sum of \eqref{eq:Gauss-Ej}, 
        i.e. do not lie in ${\cal E}_1$. 
        
        Consider $\vec \aa \in {\cal E}_0$, so that 
        $\vec\aa^{(0)} = \aa_1 \aa_2 \dots \aa_{r+1} \in K^{\bb_0 \dots \bb_r}$. 
        Then $\vec \aa^{(1)} = \aa_0 \aa_2 \dots \aa_{r+1} \in K^{\bb_0 \dots \bb_r}$
        if and only if $\aa_0 \subseteq \bb_0$. 
        It follows that $\vec \aa \in {\cal E}_0 \smallsetminus {\cal E}_1$ 
        if and only if $\aa_0 \not\subseteq \bb_0$ 
        and $\aa_{j+1} \in K^{\bb_j}_{\bb_{j+1}}$ for all $0 \leq j \leq r$, 
        which is equivalent to $\vec \aa \in dK^{\bb_0 \dots \bb_r}$ 
        by definition \ref{def:hypercones}. 
\end{proof}

\begin{proof}[Proof of theorem \ref{prop:mu-Gauss}]
    Denote by $\tilde K = K \cup \{\Om \}$ and by
    $\tilde C_\bullet = C_\bullet(\tilde K, \R^{E})$ 
    the hypergraph and complex obtained by 
    appending the global region $\Om$. 
    The complex $C_\bullet$ then naturally embeds into 
    $\tilde C_\bullet$ as a linear subspace, 
    defined by $i_\Om = 0$ when $\Om \not \in K$. 
    The first equality in \eqref{eq:mu-Gauss}, 
    \begin{equation} \label{apx:check-gauss}
        \check \delta \Phi_{\aa_1 \dots \aa_r}
        = \int_{dK^{\aa_1 \dots \aa_r}} \phi,
    \end{equation}
    simply consists in a Gauss formula \eqref{eq:Gauss-r} 
    on $K^{\aa_1 \dots \aa_r}$ for $\check \delta \Phi = \zeta(\delta \mu \Phi) = \zeta(\delta \phi)$. 
    Also recall that $dK^{\aa_1 \dots \aa_r} = \tilde K^{\Om \aa_1 \dots \aa_r}$ 
    from definition \ref{def:hypercones}. 
    
    First assume that $K = \tilde K$ contains $\Om$. 
    Then \eqref{eq:check-Phi} follows from 
    $\zeta = \mu^{-1}$ as in this case $c_\Om = 1$ and $c_\aa = 0$ for
    $\aa \neq \Om$:
    \begin{equation}
    \zeta(\delta \phi)_{\aa_1 \dots \aa_r} 
    = \zeta(\phi)_{\Om \aa_1 \dots \aa_r} 
    = \Phi_{\Om \aa_1 \dots \aa_r}.
    \end{equation}
    Assume now that $\Om \not\in K$. The Gauss formula 
    \eqref{apx:check-gauss} above then 
    consists in a truncated Möbius inversion instead 
    (i.e. a Bethe-Kikuchi approximation) which we may write
    \begin{equation} \label{eq:tilde-zeta-mu}
    \zeta(\delta \phi)_{\aa_1 \dots \aa_r} = 
    \tilde \zeta(\phi)_{\Om \aa_1 \dots \aa_r} 
    = \tilde\zeta (\mu \Phi)_{\Om \aa_1 \dots \aa_r} 
    \end{equation} 
    where $\tilde \zeta$ denotes the zeta automorphism of $\tilde C_\bullet$.
    This is not the inverse of the Möbius automorphism $\mu$ of $C_\bullet$, 
    and we write $\tilde \mu = \tilde \zeta^{-1}$. 
    
    By $\Phi \in C_\bullet \subset \tilde C_\bullet$, remember that 
    the evaluation $i_\Om$ of the first region on $\Om$ cancels $\Phi$. 
    Thus $\tilde \zeta (\tilde \mu \Phi)_{\Om \aa_1 \dots \aa_r} = 0$ so that
    \begin{equation} 
    \tilde \zeta (\mu \Phi)_{\Om \aa_1 \dots \aa_r} 
    = \tilde \zeta\big((\mu - \tilde \mu) \Phi\big)_{\Om \aa_1 \dots \aa_r}.
    \end{equation}
    Let $\tilde \phi = \tilde \mu \Phi$. The Möbius transform is local 
    in the sense that $(\tilde \mu \Phi)_{\aa_0 \dots \aa_r}$ only contains 
    terms $\Phi_{\bb_0 \dots \bb_r}$ with $\bb_0 \subseteq \aa_0$. 
    Therefore $\tilde \phi_{\aa_0 \dots \aa_r} = \phi_{\aa_0 \dots \aa_r}$ 
    whenever $\aa_0 \neq \Om$.
    From its definition \eqref{eq:mu-r}, when $\aa_0 = \Om$ we get: 
    \begin{equation} \label{eq:tilde-phi}
        \begin{split}
    \tilde \phi_{\Om \aa_1 \dots \aa_r} 
    = 
    &\sum_{\bb_r \subseteq \aa_r} 
    \mu_{\aa_r \to \bb_r}
    \dots \sum_{\substack{\bb_1 \subseteq \aa_1 \\
    \bb_1 \not\subseteq \bb_2}} 
    \mu_{\aa_1 \to \bb_1} \\
    & \sum_{\bb_0 \not \subseteq \bb_1} 
    \mu_{\Om \to \bb_0} 
    \Phi_{(\bb_0 \dots \bb_r)^\cap}.
        \end{split}
    \end{equation}
    Let us recall that 
    $c_{\bb_0} = - \tilde \mu_{\Om \to \bb_0}$ 
    (see \cite{Rota-64} or check that the inclusion-exclusion 
    principle is indeed solved by $\tilde \zeta \tilde \mu = 1$ 
    on $\tilde C_0$). 
    We now claim that $\phi - \tilde \phi$ 
    may be expressed as $\tilde \mu \check \Phi$, 
    where $\check \Phi_{\Om \aa_1 \dots \aa_r}$ is given 
    by \eqref{eq:check-Phi} 
    for all $\aa_1 \dots \aa_r \in N_{r-1} K$, and 
    $\check \Phi_{\aa_0 \dots \aa_r} = 0$ 
    for all $\aa_0 \dots \aa_r \in N_r K$.
    
    By locality of $\tilde \mu$, one easily checks 
    that $(\tilde \mu \check \Phi)_{\aa_0 \dots \aa_r} = 0$ 
    when $\aa_0 \neq \Om$, 
    as $\check \Phi$ vanishes on $N_r K$ by construction.  
    To see that $(-\tilde \mu \check \Phi)_{\Om \aa_1 \dots \aa_r}$ 
    is given by \eqref{eq:tilde-phi}, first write from 
    definition of the Möbius transform:
    \begin{equation} 
        \begin{split}
    (\tilde \mu \check \Phi)_{\Om \aa_1 \dots \aa_r} 
    = &\sum_{\bb_r \subseteq \aa_r} 
    \mu_{\aa_r \to \bb_r}
    \dots \sum_{\substack{\bb_1 \subseteq \aa_1 \\
    \bb_1 \not\subseteq \bb_2}} 
    \mu_{\aa_1 \to \bb_1} \\
    &
    \mu_{\Om \to \Om} \,
    \check \Phi_{(\Om \bb_1\dots \bb_r)^\cap},
    \end{split}
    \end{equation}
    then expand $\check \Phi_{(\Om \bb_1 \dots \bb_r)^\cap}$ 
    as a sum over $\bb_0 \not \subseteq \bb_1$ 
    of flux {\it cumulants} 
    $c_{\bb_0} \check \Phi_{(\bb_0 \dots \bb_r)^\cap} 
    = - \mu_{\Om \to \bb_0} \Phi_{(\bb_0 \dots \bb_r)^\cap}$, 
    by substituting \eqref{eq:check-Phi} defining $\check \Phi$. 
    Writing $(\Om \bb_1 \dots \bb_r)^\cap = \cc_0 \dots \cc_r$
    where $\cc_0 = \Om$ and $\cc_j = \bb_j \cap \cc_{j-1}$,
    associativity of the intersection indeed ensures
    that $\Phi_{(\bb_0 \dots \bb_r)^\cap}$ coincides 
    with $\Phi_{\bb_0 \cap (\bb_0 \cap \cc_1) \dots (\bb_0 \cap \cc_r)}$. 
    
    Having showed that $\phi - \tilde \phi = \tilde \mu \check \Phi$, 
    or equivalently that $\tilde \zeta(\phi - \tilde \phi) = \check \Phi$, 
    we may now conclude the proof using \eqref{eq:tilde-zeta-mu} to get:
    \begin{equation}
    \check \delta \Phi_{\aa_1 \dots \aa_r} 
    = \tilde \zeta(\phi - \tilde \phi)_{\Om\aa_1 \dots \aa_r} 
    = \check \Phi_{\Om \aa_1 \dots \aa_r}. \qedhere
    \end{equation}
    \end{proof}

%% file: sectionA-fix.tex
\section{Consistent Manifolds} \label{section:apx-fix}

Letting $R_\bullet \subseteq C_\bullet$ denote the subcomplex 
spanned by constant local observables, 
note that the Gibbs state map $\rho^\beta$ is invariant 
under the addition of energy constants, and
factorizes through a diffeomorphism
$C_0 / R_0 \to \Delta_0$. It follows 
that $\Fix^\beta \simeq \fix^\beta + R_0$ coincides with the
orbit of $\fix^\beta$ under the addition of local constants in $R_0$. 

However remark that $\mu(\R) \subseteq R_0$ 
acts on $\fix^\beta$, as local hamiltonians $H = \zeta(h)$ 
are then shifted by a constant term which rescales all the Gibbs
densities $\e^{-H}$ to a common mass. 
One may decompose $R_0$ as $\mu(\R) \oplus \delta R_1$, as 
\begin{equation} \label{eq:sum-mu-ab}
\sum_{\aa \supseteq \bb} \mu_{\aa \to \bb} = \sum_{\bb \supseteq \varnothing} c_\bb = 1
\end{equation}
follows from the existence of a terminal element (assumed here to be 
the empty set $\varnothing$) 
and corollary \ref{cor:BK-coeffs}. 
The fact that $\mu(1)$ does not sum to zero implies that it generates 
the only non-trivial homology class of $R_0 / \delta R_1$ 
by theorem \ref{thm:Gauss}. 
Hence, $\Fix^\beta \simeq \fix^\beta \times \delta R_1$. 

The following lemma provides an explicit retraction of $\Fix^\beta$ 
onto $\fix^\beta$. 

\begin{Lemma} \label{apx:fix-retraction}
Let $F^\beta : C_0 \to R_0$ denote the direct sum 
$\oplus F^\beta_\aa$ of local free energy maps (definition \ref{def:functionals}). 
Then a natural retraction $r^\beta : \Fix^\beta \to \fix^\beta$ is defined by 
letting for all $v \in \Fix^\beta$: 
\begin{equation}
r^\beta(v) = v - \mu \big( F^\beta(\zeta v)\big). 
\end{equation}
For all $v \in \Fix^\beta$, one has $v \in \fix^\beta$
if and only if $F^\beta(\zeta v) = 0$.
\end{Lemma}

\begin{proof}
The result is easily proved by observing that both $\zeta$ and $\mu$ 
are equivariant under the action of $R_0$, and that $F^\beta(V)_\aa$ 
computes $-\ln Z_\aa$, i.e. free energies vanish if and only if 
the Gibbs densities $\e^{-\beta V_\aa}$ are normalised. 
\end{proof}

\begin{Lemma} \label{apx:fix-beta}
    For all $\beta > 0$, 
    $\fix^\beta$ is the image of $\fix^1$ under the 
    scaling $v \mapsto \beta^{-1} v$. 
\end{Lemma}

\begin{proof}
This directly follows from definition \ref{def:Fix}, 
as $\rho^\beta(v) = \rho(\beta v)$. 
\end{proof}

Extension to zero or negative inverse temperatures 
would describe infinitesimal and reversed interactions respectively. 
Although not physical, negative temperatures describe
a natural symmetry of $\Delta_0$ which is diffeomorphic 
to the quotient Lie group $C_0 / R_0$ under $\rho$,  
energy scalings $\beta \in \R$ acting as exponents in $\Delta_0$. 

\begin{Proposition}
Let $u = r^1(0) \in \fix^1$ denote the consistent potentials 
having uniform Gibbs densities. Then $v \in \Tg_u \fix^1$ 
if and only if $V = \zeta v$ satisfies: 
\begin{equation} \label{eq:T_0-fix}
V_\bb(x_\bb) = \E_{\rho(0)}[V_\aa | x_\bb] = 
\frac{|E_\bb|}{|E_\aa|} \sum_{x_{\aa|\bb} = x_\bb} V_\aa(x_\aa)
\end{equation}
for all $\aa \supset \bb$ in $K$ and all $x_\bb \in E_\bb$. 
\end{Proposition}

Equation \eqref{eq:T_0-fix} describes consistency for systems 
having infinitesimal interactions with respect to the
energy scale (high temperature limit). In the non-interacting case,
a canonical choice of {\em interaction subspaces} 
${\bf Z}_\aa \subseteq \R^{E_\aa}$,
where 
\begin{equation}\label{eq:Za-canonical}
    {\bf Z}_\aa \subseteq \bigg(\sum_{\bb \subseteq \aa} \R^{E_\bb}\bigg)^\perp
    \quad{\rm for}\;  \eta_{\rho(0)} 
\end{equation}
further allows to describe $\Tg_u \fix^1$ as 
the set of those $v \in C_0$ such that $v_\aa \in {\bf Z}_\aa$ 
for all $\aa \in K$ \cite[section 5.2.1]{phd}.
The metric $\eta_{\rho(0)}$ is induced by
uniform beliefs (see section \ref{section:Equilibria}), 
so that conditional expectations of \eqref{eq:T_0-fix} 
are orthogonal projections for $\eta_{\rho(0)}$. 
Note that such an orthogonal decomposition,
compatible with the metric $\eta_p = \eta_{\rho(0)}$,
is only possible because all variables 
are independent in the high-temperature limit. 

%% file: sectionA.tex
\section{Faithfulness Proofs}

The faithfulness properties of a diffusion flux $\Phi : C_0 \to C_1$ 
(definition \ref{def:faithful})
relate the stationary potentials under diffusion, for which 
$\delta \Phi(v) = 0$, with consistent potentials 
which satisfy $\e^{-\zeta v} \in \Ker(d)$. The latter is equivalent to $\DF(\zeta v) = 0$
(where $\DF : C_0 \to C_1$ is the non-linear free energy gradient, definition \ref{def:DF}).

Although proving that any consistent potential is stationary is straightforward when 
$\Phi$ is built from $\DF \circ \zeta$, 
proving the converse is more difficult due to the non-linearity of $\DF$. 
Below we prove that $\Phi_{GBP}$ \eqref{eq:X-GBP} is indeed faithful, and 
that $\Phi_{BK}$ \eqref{eq:X-BK} is locally faithful (i.e. any stationary potential 
sufficiently close to the consistent manifold is consistent), although global 
faithfulness of $\Phi_{BK}$ should be conjectured. 

\begin{Proposition} \label{prop:GBP-faithfulness}
The GBP diffusion flux $\Phi_{GBP} = - \DF \circ \zeta$ 
is faithful. Equivalently for all $V = \zeta v \in C_0$ we have: 
\begin{equation}
\delta (\DF V) = 0 \quad \eqvl \quad \DF V = 0.
\end{equation}
\end{Proposition}

\begin{proof} 
By $d = \delta^*$, first recall that we have for all $\phi \in C_1$ and $q \in C_0^*$ 
the duality (integration by parts) formula: 
\begin{equation} \label{eq:A-ipp}
\langle q | \delta \phi \rangle \;=\;  
\langle dq | \phi \rangle.
\end{equation}
Assume now that $\phi = \DF V$ lies in $\Ker(\delta)$. 
Then $\langle dq | \phi \rangle = 0$ for all $q \in C_0^*$ by \eqref{eq:A-ipp}. 
Choosing local (unnormalized) Gibbs densities for $q = \e^{-V} \in C_0^*$, 
let us show that $\langle dq | \DF V \rangle = 0$ implies $\DF V = 0$ 
by a monotonicity argument. 

Expanding the duality bracket of $C_1^* \times C_1$ yields
\begin{equation} \label{eq:A-bracket}
\langle dq | \DF V \rangle 
= \sum_{\aa \to \bb} 
\sum_{x_\bb} 
dq_{\aa \to \bb}(x_\bb) \cdot \DF V_{\aa \to \bb}(x_\bb), 
\end{equation}
and we claim that the summands of \eqref{eq:A-bracket} identically vanish. Indeed we have 
from definition 
\begin{equation} 
dq_{\aa \to \bb}(x_\bb) = \e^{-V_\bb(x_\bb)} - \sum_{x_{\aa|\bb} = x_\bb} \e^{-V_\aa(x_\aa)},
\end{equation}
\begin{equation}
\DF V_{\aa \to \bb}(x_\bb) 
= V_\bb(x_\bb) - \bigg( - \ln \sum_{x_{\aa|\bb} = x_\bb} \e^{-V_\aa(x_\aa)} \bigg),
\end{equation}
and $y \mapsto - \ln y$ is strictly decreasing. The signs of $dq_{\aa \to \bb}(x_\bb)$ and 
$\DF V_{\aa \to \bb}(x_\bb)$ therefore disagree for every $\aa \to \bb$ in $N_1K$ and every $x_\bb \in E_\bb$, 
so that \eqref{eq:A-bracket} is a vanishing sum of non-positive terms, i.e. a sum of zeroes. 

Noticing that $dq_{\aa \to \bb}(x_\bb) = 0$ is moreover equivalent to 
$\DF V_{\aa \to \bb}(x_\bb) = 0$, we conclude that $\DF V = dq = 0$, i.e. 
that Gibbs densities are consistent whenever $\delta(\DF V) = 0$. 
\end{proof}

\begin{Proposition} \label{prop:BK-faithfulness}
The Bethe-Kikuchi diffusion flux $\Phi_{BK} = - \mu \circ \DF \circ \zeta$ 
is locally faithful: there exists an open neighbourhood $\cal V$ of $\fix$ such that 
for all $v \in {\cal V} \subseteq C_0$, 
\begin{equation}
\delta \Phi_{BK}(v) = 0 \eqvl v \in \fix. 
\end{equation}
\end{Proposition}

\begin{proof}
Given $v \in \fix$ and $w \in \Tg_v C_0 \simeq C_0$,
we denote by $p = \rho(\zeta v) \in \Gamma_0$ the consistent beliefs 
defined by $v$,  
and by $\Tg_v X_{BK} : \Tg_v C_0 \to \Tg_0 C_0$ the differential of 
$X_{BK}$ at $v$. 
Let us show that 
the linearized stationarity condition $\Tg_v X_{BK}(w) = 0$
implies $w \in \Tg_v \fix$. 

Letting $W = \zeta w \in \zeta(\Tg_v \fix) \subseteq C_0$,  
by linearizing the evolution of local hamiltonians  
\eqref{eq:BK-diffusion},
we claim that $\Tg_v X_{BK}(w) = 0$ 
if and only if for all $\bb \in K$ and all $x_\bb \in E_\bb$,  
\begin{equation} \label{eq:A-stationarity}
W_\bb(x_\bb) = \sum_{\aa \in K} c_\aa \,\E_{p_\aa}[W_\aa | x_{\bb | \aa \cap \bb}].
\end{equation}
The reader may check that
conditional free energies  of \eqref{eq:BK-diffusion} linearize to conditional expectations 
with respect to $p \in \Gamma_0$ 
(see \cite[prop. 4.14]{phd} otherwise).

In order to factorize the sum over regions inside a single 
conditional expectation, we choose a global 
measure $p_\Om \in \R^{E_\Om *}$ such that $p_\aa$ is the pushforward 
of $p_\Om$ for all $\aa \in K$. This is always possible by consistency of 
$p \in \Gamma_0$ and acyclicity of $(C_\bullet^*, d)$ \cite[thm. 2.17]{phd}. 
Altough $p_\Om$ cannot always be chosen non-negative, with 
a slight abuse of notation let us write
\begin{equation}
\E_{p_\Om}[W_\Om|x_\aa] = \sum_{x_{\Om|\aa} = x_\aa} 
\frac {p_\Om(x_\Om) W_\Om(x_\Om)} {p_\aa(x_\aa)}.
\end{equation}
so that $\E_{p_\aa}[W_\aa|x_{\aa \cap \bb}] = \E_{p_\Om}[W_\aa|x_{\aa \cap \bb}]$ 
for all $\aa,\bb \in K$, all $W_\aa \in \R^{E_\aa} \subseteq \R^{E_\Om}$, 
and all $x_{\aa \cap \bb} \in E_{\aa \cap \bb}$.  
Because $W_\aa$ does not depend on the variables outside $\aa$ we also have 
\begin{equation}
\E_{p_\Om}[W_\aa|x_{\bb | \aa \cap \bb}] = \E_{p_\Om}[W_\aa|x_\bb]
\end{equation}
for all $x_\bb \in E_\bb$,
where $\R^{E_{\aa \cap \bb}}$ is identified
with its image in $\R^{E_\bb}$.   
This lets us to rewrite the sationarity condition \eqref{eq:A-stationarity}
as 
\begin{equation}
W_\bb(x_\bb) = \E_{p_\Om} \bigg[ 
\sum_{\aa \in K} c_\aa W_\aa \bigg| x_\bb \bigg] = \E_{p_\Om}[W_\Om|x_\bb],
\end{equation}
which implies compatibility of $W \in C_0$ under the conditional 
expectation functor, i.e. $\E_{p_\bb}[W_\bb | x_\cc] = W_\cc(x_\cc)$ 
for all $\cc \subseteq \bb$ and $x_\cc \in E_\cc$. 
The constraint that local hamiltonians are conditional expectations 
of a common global energy $W_\Om = \sum_\aa c_\aa W_\aa = \sum_\aa w_\aa$ 
precisely defines the tangent manifold $\zeta (\Tg_v \fix)
= \Ker ( \Tg_{\zeta v} \DF)$ (see \cite[section 5.2.1]{phd}
or differentiate $\DF$).
Therefore $w \in \Tg_v \fix$. 

We showed that the stationary fiber $\Ker(\Tg_v X_{BK})$ 
coincides with the consistent tangent fiber $\Tg_v \fix$, assuming $v \in \fix$. 
The stationary manifold contains $\fix$ by consistency of $\Phi_{BK}$ 
but we do not know whether it is connected.  
The vector field $X_{BK} = \delta \Phi_{BK}$ being analytic,
one may still prove that connected components of the stationary manifold are 
separated to construct the desired neighbourhood ${\cal V} \supseteq \fix$. 
\end{proof}

%% file: sectionB.tex
\section{Correspondence Proofs} \label{apx:correspondence}

\begin{proof}[Proof of theorem \ref{thm2}]
    Given $\U \in \R$ and $h \in C_0$, let $p \in \Gamma_0$
    be a solution to problem~\ref{S-crit},
    with $\langle p, h \rangle = \U$. 
    Defining local energies $\bar V = \zeta \bar v \in C_0$ 
    by $\bar V_\bb = - \ln p_\bb$ for all $\bb$, 
    the consistency of $p$ implies $\bar v \in \Fix^1$
    as $p = \rho^1(\bar V) = \rho^1(\zeta \bar v)$.
    The same argument holds for any other potential $\bar v' \in \bar v + \R^K$, 
    equal to $\bar v$ up to additive energy constants.
    This follows from the invariance of 
    Gibbs states under additive constants,
    i.e. $\rho^1(\zeta \bar v') = \rho^1(\zeta \bar v)$ 
    is consistent for all
    $\bar v' \in \bar v + \R^K$, and $\bar v' \in \Fix^1$.   
    
    Let us now show that the constraints $\langle p, h \rangle = \U$
    and $dp = 0$  
    also imply
    $\bar v \in [\beta h]$ for some $\beta \in \R$. 
    By duality, recalling that $\delta = d^*$, we know that 
    $\langle q, u \rangle = 0$ for all $q \in \Ker(d)$ 
    if and only if $u \in \Img(\delta)$. 
    The consistency constraint hence yields Lagrange multipliers 
    of the form $\delta \ph \in \delta C_1$. 
    The mean energy constraint 
    $\langle p, h \rangle = \U$,
    as usual,  
    yields the inverse temperature $\beta$ 
    as a Lagrange multiplier for terms of the form $\beta h \in \R h$. 
    The normalisation constraints $\langle p_\aa, 1_\aa \rangle = 1$ 
    for all $\aa \in K$ finally yield additive constants $\lambda \in \R^K$. 
    
    A classical computation yields the differential 
    of local entropies as $\Tg_{p_\bb} S_{\bb} = \langle \:\cdot\:, - \ln(p_\bb)\rangle {\rm\;\mod\;} \R$.
    They coincide with $\langle \:\cdot\: , \bar V'_\bb\rangle$ 
    for all $\bar V'_\bb \in - \ln(p_\bb) + \R$ 
    on the tangent fibers of $\Delta_\bb = \Prob(E_\bb)$,
    due to the tangent normalisation constraint. 
    Therefore $p$ is critical if and only if 
    there exists $\beta \in \R$, $\ph \in C_1$ and $\lambda \in \R^K$ such that:
    \begin{equation} \label{dS}
    \check S_*(p)_\aa = c_\aa (\bar V_\aa - 1_\aa) = 
    \beta h_\aa + \delta \ph_\aa + \lambda_\aa 1_\aa.
    \end{equation}
    By lemma \ref{lemma1} we know that $c \bar V$ and $\bar v = \mu \bar V$
    are homologous: 
    $\bar v = c \bar V - \delta \Psi(c \bar V)$. Lemma \ref{lemma1} 
    therefore implies that 
    \eqref{dS} is equivalent to 
    $\bar v \in \beta h + \delta C_1 + \lambda$ 
    with $\lambda \in \R^K$.
    Theorem \ref{thm:A}, applied to the subcomplex of constant 
    local observables $R_\bullet \subseteq C_\bullet$, 
    implies that $R_0$ decomposes as $\R \oplus \delta R_1$.  
    Up to a boundary term of $\delta R_1 \subseteq \delta C_1$,
    the local constants 
    $\lambda \in \R^K = R_0$ 
    can therefore be absorbed into a single constant $\lambda' \in \R$ 
    so that $\bar v \in \beta h + \delta C_1 + \R$.
    Enforcing the constraint $\langle p, \bar v' \rangle = 
    \beta \langle p, h \rangle = \beta\, \U$, 
    on the equivalent potentials $\bar v' \in \bar v + \R$,
    we may 
    finally get $\bar v' \in \beta h + \delta C_1$ 
    by $\langle p, \delta C_1 \rangle = 0$.   
    Each critical $p \in \Gamma_0$
    hence defines a unique potential
    $\bar v' \in [\beta h] \cap \Fix^1$ 
    satisfying $p = \rho^1(\zeta \bar v')$ 
    and $\langle p, \bar v' \rangle = \beta \,\U$. 
    \end{proof}

    \begin{proof}[Proof of theorem \ref{thm1}]
    Very similar to that of \ref{thm2} above the precise proof can 
    be found in \cite[thm 4.22]{phd}.
    \end{proof}
        
    \begin{proof}[Proof of theorem \ref{thm3}]
    Given $H = \zeta h \in C_0$ and $\beta > 0$, 
    let $V = \zeta v \in C_0$ be a solution to problem \ref{F-crit}. 
    The energy constraint on $v \in h + \delta C_1$ 
    implies $\check F^\beta_*(V)_{|\Img(\zeta\delta)} = 0$ with
    the differential $\check F^\beta_* : C_0 \to C_0^*$ given by: 
    \begin{equation}
        \check \F^\beta_*(V)_\aa = c_\aa \: \rho^\beta(V_\aa).
    \end{equation}
    For any subspace $B \incl C_0$, let us write $B^\perp \incl C_0^*$ for
    the orthogonal dual (or annihilator) of $B$.
    Letting $p = \rho^\beta(V)$
    criticality is then equivalent to $\F^\beta_*(V) = c p \in \Img(\zeta \delta)^\perp$.
    Recall that $d = \delta^*$ by definition and denote by 
    $\zeta^* : q \mapsto q \circ \zeta$ the adjoint of $\zeta$: 
    
    \begin{equation}
        \begin{split}
        cp \in \Img(\zeta\delta)^\perp 
        &\quad\eqvl\quad   \zeta^*(cp) \in \Img(\delta)^\perp \\
        &\quad\eqvl\quad   \zeta^*(cp) \in \Ker(d).
        \end{split}
    \end{equation}  
    Therefore $V$ is critical for $\check \F^\beta_{|\Img(\zeta \delta)}$ 
    if and only if $\zeta^*(c p) = \zeta^*(c \,\rho^\beta(V))$ is consistent.
    
    Assume $p \in \Ker(d)$ is consistent. Then by the dual form of 
    lemma \ref{lemma1} involving $(\delta \Psi)^* = \Psi^* d$ 
    which vanishes on $\Ker(d)$, we have
    $c p = \mu^* p$ 
    and $\zeta^*(cp) = p$ is consistent as well. 
    This shows that any consistent potential $v \in [h] \cap \Fix^\beta$ 
    yields a critical point $V = \zeta v$ of the Bethe-Kikuchi energy 
    $\check \F^\beta(V)$.
    
    Reciprocally assume $q = \zeta^*(c p) \in \Ker(d)$, 
    then $\mu^*(q) = c q$ by lemma \ref{lemma1} and $c q = c p$,
    which means that $p_\bb$ and $q_\bb$ coincide on any $\bb \in K$ 
    such that $c_\bb \neq 0$. 
    Let us define energies $W = \zeta w \in C_0$ by: 
    \begin{equation} \label{W}
    W_\bb = - \frac 1 \beta \ln(q_\bb) + \F_\bb^\beta(V_\bb).
    \end{equation}
    Observing that $q = \rho^\beta(W)$ and
    $\F^\beta_\bb(W_\bb) = \F^\beta_\bb(V_\bb)$ for all $\bb$,
    one sees that $c p = c q$ implies $c V = c W$. 
    By lemma \ref{lemma1} we get
    $\mu V \in \mu W + \delta C_1$ so that the potentials $w = \mu W$ 
    satisfy the global energy constraint $[w] = [v] = [h]$.
    By the assumption $q \in \Ker(d)$ 
    this shows that $w \in [h] \cap \Fix^\beta$. 
    The relation $cV = cW$ implies
    $v - w = \mu (V - W) \in \Ker(c\zeta)$, 
    so that $v \in [h] \cap \Fixt^\beta$ 
    is weakly consistent by definition of $\Fixt^\beta$. 
    
    The retraction $r^\beta : \Fix_+^\beta \to \Fix^\beta$ 
    is described more succinctly via its conjugate 
    $R^\beta = \zeta \circ r^\beta \circ \mu$ acting
    on local energies $V = \zeta v \in C_0$:
    \begin{equation}
        R^\beta(V)_\bb = - \frac 1 \beta
        \ln \sum_{\aa \cont \bb} c_\aa \: \pi_*^{\aa \to \bb}
        \big(\rho^\beta_\aa(V_\aa) \big) 
        - \F_b^\beta(V_\bb),
    \end{equation}
    which amounts to applying the linear projection $\zeta^* c$ on beliefs
    before choosing energies as \eqref{W}.
    \end{proof}